\documentclass{article}
\usepackage{graphicx} 
\usepackage{amsmath,amsfonts, bm, amsthm}
\usepackage[left=2.5cm,right=2.5cm,top=2.5cm,bottom=2cm]{geometry}
\usepackage{xcolor}
\usepackage{amsthm}         
\usepackage{dsfont}			
\usepackage{hyperref}
\usepackage{multirow}
\usepackage{graphicx}
\usepackage{subcaption}
\usepackage{epsfig}
\usepackage{mathbbol}
\usepackage{booktabs} 
\usepackage{subcaption}  
\usepackage{comment}
\usepackage{diagbox}
\usepackage{booktabs}
\usepackage{natbib}
\usepackage{placeins}

\def\ab{{\bm{a}}}
\def\Xb{{\bm{X}}}
\def\Zb{{\bm{Z}}}
\def\xb{{\bm{x}}}
\def\yb{{\bm{y}}}
\def\epb{{\bm{\varepsilon}}}
\def\etab{{\bm{\eta}}}
\def\E{\operatorname{E}}
\def\var{\operatorname{var}}
\newcommand\smallO{{{\scriptscriptstyle\mathcal{O}}}} 

\newtheorem{theorem}{Theorem}
\newtheorem{remark}{Remark}
\newtheorem{lemma}[theorem]{Lemma}

\newcommand*{\Prob}[2]{\operatorname{P}_{#1} \left( {#2} \right)}

\begin{document}

\title{Asymptotically exact threshold for detecting anomalies in multivariate Gaussian data with application to time series}

\author{Marie Turčičová$^{a0}$,
Patrícia Martinková$^{a}$ \\
		\\
\small {\textit{$^a$Institute of Computer Science, Czech Academy of Sciences,}}\\
\small{\textit{Pod Vodárenskou věží 271/2, Prague, 182 00,  Czech Republic}}}
\date{} 
	
\footnotetext{Corresponding author. Email address: turcicova@cs.cas.cz}
\maketitle

\begin{abstract}
In this paper, we propose a new thresholding technique for detecting anomalies in multivariate normal random samples, under the assumption that anomalous observations are sparse and differ from the rest of the data in their mean. The mean vector of the non-anomalous data is assumed to be zero, while the covariance matrix is unknown. We derive conditions on the mean shift of the anomalous observations, as well as on the covariance matrix and its estimator, under which the proposed procedure achieves asymptotically exact detection, meaning that the expected number of misclassified observations converges to zero as the sample size increases. In addition, we establish conditions under which exact anomaly detection is impossible for any procedure. The performance of the proposed method is illustrated through an extensive simulation study and compared with other widely used anomaly detection methods. Real-data analyses involving wearable activity measurements and air pollution time series provide an assessment of its performance in real-world settings. 
\end{abstract}

\bigskip
\noindent \textit{Keywords:} exact anomaly detection, {threshold}, {multivariate normal distribution}, {time series}, {Huber covariance estimator}

\section{Introduction}
Anomalies, often referred to as outliers in the literature, are observations that deviate significantly from the rest of the data (see, for example, \cite{Barnett-1998, Chandola-2009, Hawkins-1980, Schmidl-2022}). Such deviations can take various forms, but the most common is a change in the mean, as considered in this paper. 
More specifically, we focus on anomalies represented by individual measurements that differ from the rest of the data by their expected value. These typically arise as transient measurement errors or disturbances that do not imply a permanent change in the process mean.
Anomaly detection is then a task of identifying such nonconforming observations within the dataset and can be viewed as a binary classification problem. In statistical analysis, anomalous observations present a serious challenge, as they can lead to issues such as poor model fit, model misspecification, or biased parameter estimation, as noted by \cite{Laurikkala-2000, Rousseeuw-1987, Tsay-2000}. In practice, their detection is highly valuable because they often indicate critical events such as manufacturing defects, system failures, or health abnormalities (see, for example, \cite{Chandola-2009, Schmidl-2022, Hill-2010, Lu-2007} and references therein). 

Characterizing anomalies in a multivariate setting is considerably more challenging than in the univariate case. In addition to distorting measures of location and scale, anomalous observations may also affect the dependence structure of the data, as reflected in the correlations between variables (see, e.g., \cite{Gnanadesikan-1972}). 
A multivariate observation may be considered anomalous either due to a gross error in one of its components or because an unusual combination of otherwise typical component values makes it inconsistent with the rest of the data, as discussed by \cite{Barnett-1998, Gnanadesikan-1972, Laurikkala-2000}. Multivariate outlier detection aims to identify these atypical combinations of values, rather than focusing solely on extreme behavior in individual variables. A standard approach is to employ distance functions, which reduce multivariate responses to scalar measures \cite[Chapter~7]{Barnett-1998}. A widely used and effective class of such measures consists of positive semi-definite quadratic forms of the type $\boldsymbol{e}^\top M \boldsymbol{e}$, where $\boldsymbol{e}$ denotes a residual vector (or a vector of deviations) and $M$ is a matrix constructed from these residuals (see, for example, \cite[Chapter~7]{Barnett-1998} and \cite{Gnanadesikan-1972}).

Despite the wide range of existing detection techniques, thresholding methods remain important, particularly for detecting anomalies caused by shifts in the mean. However, determining an appropriate threshold is a statistically nontrivial task. In some cases, thresholds are chosen subjectively or based on domain-specific engineering knowledge, as discussed by \cite{Basu-2007}. Other methods define thresholds as fixed multiples of the standard deviation estimated from the data or from model residuals, an approach considered by \cite{Chandola-2009, Rousseeuw-1987, Shewhart-1931, Zhou-2016}. A related approach is the boxplot rule, which classifies observations as outliers if they lie more than 1.5 times the interquartile range below the first quartile or above the third quartile, as described in \cite{Chandola-2009, Laurikkala-2000}.

In several frameworks, the threshold is defined not as a fixed value but as a probability level or risk parameter chosen by the researcher. Such approaches typically require some a priori knowledge of the data distribution. A common practice is to select a probability level (for example, 0.05 or 0.01) and to use a quantile of a suitable distribution (e.g., normal, Student’s~t, or beta distribution) or a function thereof to determine the corresponding threshold value, as adopted by \cite{Chandola-2009, Gnanadesikan-1972, Siotani-1959, Liu-1991}. Similarly, in conformal prediction-based approaches considered by \cite{Laxhammar-2014}, a new observation is flagged as anomalous if its estimated p-value, computed from a nonconformity score, falls below a chosen threshold~$\epsilon$ (e.g., 0.005, 0.01, or 0.02).

Another popular thresholding approach is based on extreme value theory. In this approach, an initial threshold is selected as a high empirical quantile (e.g., 98\%), so no prior knowledge of the data distribution is required. Observations exceeding this threshold are modeled using an extreme-value distribution (e.g., generalized Pareto distribution), and its quantiles then serve as thresholds for identifying the actual anomalies. This method is commonly referred to as the Peaks-Over-Threshold (POT) approach (see, for example, \cite{Siffer-2017}). 

Most existing thresholding methods rely on subjective choices of probability levels of the corresponding distributional or empirical quantiles. Such choices play a crucial role in determining the number of detected anomalies and can substantially affect overall performance. In the absence of prior knowledge about the expected number of outliers, selecting an appropriate probability level becomes difficult. 

To address this limitation, we propose a thresholding procedure for detecting outliers in multivariate Gaussian samples that requires minimal subjective tuning. Our approach builds on the results of \cite{Butucea-2018} for univariate data coming from a normal distribution with known variance, where the authors introduced a threshold based solely on distributional properties, without requiring the specification of quantiles. 
That threshold was developed in the context of variable selection in a Gaussian sequence model and was shown (under certain conditions) to achieve asymptotically exact identification of the relevant components of an unknown parameter vector. In this paper, we extend these results to the multivariate setting with unknown covariance and adapt them to anomaly detection.

The proposed method requires only limited subjective input: the sole tuning parameter is constrained by an asymptotic condition, which significantly restricts its variability. Moreover, neither the form of the threshold nor the choice of this parameter depends on prior knowledge of the sparsity level of the anomalies and thus, the method is adaptive. We also address the problem of covariance estimation and specify conditions on the data and the covariance estimator under which the proposed method achieves asymptotically exact detection, i.e., the expected number of misclassified observations converges to zero as the sample size increases. The method can be applied directly to multivariate Gaussian samples with known mean, as well as to residuals from statistical models that are approximately Gaussian and centered. We illustrate its performance using residuals from ARIMA and VAR models. 

The paper is organized as follows. In Section~\ref{sec:method}, we review commonly used thresholding methods for outlier detection in univariate and multivariate settings. The main results are presented in Section \ref{sec:AOT}, where we propose an asymptotically exact threshold for detecting anomalies in normally distributed data, applicable to both univariate and multivariate settings. We also provide theoretical guarantees describing its properties (Theorems~\ref{theorem:upper_bound} and~\ref{theorem:lower_bound}), with proofs given in the Appendix. Section~\ref{Sec:simul} presents simulation studies, demonstrating theoretical properties of the proposed approach, and comparing it with existing methods. 
In Section~\ref{sec:realdata}, we investigate wearable-device measurements and air-pollution data using the proposed methodology, highlighting its performance relative to existing approaches.
The paper concludes with a discussion.


\section{Methods} \label{sec:method}

Assume that the standard data have zero expected value and the anomalies have an expected value $\ab=(a_1,\ldots,a_p)^\top \in \mathbb{R}^p\setminus \{\boldsymbol{0}\}$. In particular, consider that observations $\Xb_i = (X_{i1},\ldots,X_{ip})^\top$ are described by the model
\begin{equation} \label{model}
	\Xb_i =  \eta_i \ab + \epb_i, \quad i=1,\ldots,n,
\end{equation}
where $\epb_i \sim N_p(\bm{0},\Sigma)$ are independent random vectors representing noise (with positive definite $p \times p$ covariance matrix $\Sigma$) and $\eta_i$'s are anomaly indicators, i.e.  for $\eta_i=0$, $\Xb_i \sim N_p(\bm{0},\Sigma)$ is a standard observation and for $\eta_i =1$, $\Xb_i \sim N_p(\ab,\Sigma)$ is an anomaly. 
In \cite{Barnett-1998}, the distribution $N_p(\ab,\Sigma)$ is called a location-slippage alternative model. Such data $\Xb_1,\ldots,\Xb_n$ can represent residuals from regression models, Kalman filters, and many other statistical models, in which normality is a common assumption. 
Based on these data, our goal is to construct an estimator $\hat{\etab} = (\hat{\eta}_i)_{i=1}^n$ of $\etab = (\eta_i)_{i=1}^n$ that accurately identifies the anomalous observations.

In \cite{Butucea-2018}, the estimator $\hat{\etab}$ is called \textit{a selector} because it selects relevant variables of the unknown vector of parameters in a regression model, which can also be seen as a problem of signal recovery. In contrast, the term ``detection" is typically reserved for the task of deciding whether there are any relevant components or the signal is identically zero.
From this perspective, the problem considered in this paper is more precisely described as anomaly selection, since our goal is to select observations exhibiting anomalous behavior. However, the term ``anomaly detection'' is firmly established in the statistical literature and practical applications. Hence, within the framework of model~(\ref{model}), we adopt this terminology and refer to the vector $\hat{\etab}$ as \textit{a detector}.

\subsection{Anomaly detection thresholding methods for univariate data} \label{sec:existing_methods}
We begin by reviewing methods commonly used for anomaly detection in a univariate random sample $X_1,\ldots, X_n$ coming from the Gaussian distribution $N(\mu,\sigma^2)$. In the case of univariate observations, the detection process is straightforward, because the extremeness can be easily measured by the magnitude of the observation, and we only have to set a limit for this magnitude - the threshold. 

Probably the most widely used method is the \textit{$3\sigma$ rule}, a classical approach considered by \cite{Chandola-2009, Rousseeuw-1987, Shewhart-1931}. By using this rule, we mark as anomalous observations $X_i$ that are distant from $\mu$ by more than $3\sigma$, i.e.
\begin{equation} \label{3sigma}
	|X_i| > \bar{X} + 3s,
\end{equation}
where the unknown $\mu$ and $\sigma$ are estimated by sample mean $\bar{X}$ and sample standard deviation~$s$. This rule comes from the fact that a normally distributed random variable has only 0.3~\% probability of lying beyond~$3\sigma$. 

Another standard technique, commonly implemented in statistical software (e.g., R, SAS, SPSS) as part of basic descriptive statistics, is the \textit{boxplot rule} (see, e.g., \cite{Chandola-2009, Laurikkala-2000}). This rule is based on the interquartile range, defined as $IQR = Q_3 - Q_1$, where $Q_1$ and $Q_3$ denote the first and third quartiles, respectively. By this rule, an observation $X_i$ is classified as anomalous if it lies more than $1.5 \cdot IQR$ below $Q_1$ or more than $1.5 \cdot IQR$ above $Q_3$, that is,
\begin{equation} \label{boxplot_rule}
	X_i < Q_1 - 1.5 \cdot IQR \quad \text{or} \quad X_i > Q_3 + 1.5 \cdot IQR.
\end{equation}
This technique is applicable to data coming from an arbitrary distribution. A normal random variable lies in the region between $Q_1 - 1.5 \cdot IQR$ and $Q_3 + 1.5 \cdot IQR$ with probability 0.993, so for Gaussian data, this technique provides similar results as the $3\sigma$ rule. 

A technique known as the \textit{maximum normed residual test} is described in \cite{Grubbs-1969, Stefansky-1971, Stefansky-1972} and it marks an observation $X_i$ coming from a normal distribution as anomalous if 
\begin{equation} \label{Grubbs}
	e_i = \frac{|X_i-\bar{X}|}{s} > \frac{n-1}{\sqrt{n}} \sqrt{\frac{qt_{n-2}(\alpha/(2n))^2}{n-2+qt_{n-2}(\alpha/(2n))^2}},
\end{equation}
where $s$ is the standard deviation of the data and $qt_r(\alpha)$ denotes the $\alpha$ quantile of distribution~$t_r$. This threshold for a single normed residual $e_i$ was derived by directly using its monotone relationship with a Student $t$-statistic and using Bonferroni's method. Following the terminology of \cite{Chandola-2009}, we refer to this method as the \textit{Grubbs test} throughout this paper.

\subsection{Anomaly detection thresholding methods for multivariate data} \label{sec:existing_methods_multivariate}

For multivariate observations $\mathbf{X}_1,\ldots,\mathbf{X}_n$, the detection problem becomes more challenging because multivariate outliers are not necessarily univariate outliers. Even unusual combinations of moderate values may render an observation anomalous in the multivariate sense. Detecting such atypical value combinations requires introducing some ordering principle that measures the extremeness of the observed vectors. This principle is usually implemented using a distance measure that transforms the multi-response observations into a single scalar value. 
A popular choice of such a distance function is the Mahalanobis squared distance~(see, e.g., \cite{Anderson-2003})
\begin{equation} \label{Mahalanobis_true}
	M_{\Sigma}(\Xb_i) = (\Xb_i - \boldsymbol{\mu})^\top \Sigma^{-1}(\Xb_i - \boldsymbol{\mu}).
\end{equation}
One of the key advantages of the Mahalanobis distance is its scale invariance because each variable is standardized to a mean of zero and a variance of one. The sample Mahalanobis distance is defined by 
\begin{equation} \label{Mahalanobis}
	M_C(\Xb_i) = (\Xb_i - \bar{\Xb})^\top C^{-1}(\Xb_i - \bar{\Xb})
\end{equation}
with $C$ being a covariance estimator based on $\Xb_1,\ldots,\Xb_n$. 

Denote by $S_1$ the sample covariance matrix. Authors of
\cite{Gnanadesikan-1972} showed that for normally-distributed observations
$\frac{n}{(n-1)^2}M_{S_1}(\Xb_i) \sim  Beta\left( \frac{p}{2}, \frac{n-p-1}{2} \right).$
Thus, $\Xb_i$ can be classified as anomalous if 
\begin{equation} \label{qBeta}
	M_{S_1}(\Xb_i) > \frac{(n-1)^2}{n} qBeta\left(1-\alpha, \frac{p}{2}, \frac{n-p-1}{2} \right),
\end{equation}
where $qBeta\left(1-\alpha, d_1, d_2 \right)$ denotes an $(1-\alpha) 100\%$ quantile of beta distribution with $d_1$ and $d_2$ degrees of freedom. 

Authors of \cite{Laurikkala-2000} propose applying a standard boxplot-rule for $M_S(\Xb_1), \ldots,M_S(\Xb_n)$, and mark $\Xb_i$ as anomalous if
\begin{equation} \label{boxplot_forM}
	M_{S_1}(\Xb_i) > Q_3(M_{S_1}(\mathbb{X}))+1.5*IQR(M_{S_1}(\mathbb{X})),
\end{equation}
where $M_{S_1}(\mathbb{X}) = (M_{S_1}(\Xb_1),\ldots,M_{S_1}(\Xb_n))$ and $Q_3(M_{S_1}(\mathbb{X}))$ is its third quartile.

Another remarkable method was proposed by \cite{Siotani-1959}. An observation $\Xb_i$ is called anomalous if
\begin{equation} \label{Siotani}
	M_C(\Xb_i) > (n-1)\frac{\nu}{n} \left( \frac{1}{qBeta\left( \frac{\alpha}{d},\frac{\nu+1-p}{2},\frac{p}{2} \right)}-1 \right),
\end{equation}
where $C$ is an unbiased estimator of $\Sigma$ with $\nu$ degrees of freedom. Mahalanobis distance is proportional to Hotelling's $T^2$ statistic, which maps monotonically to an F (and hence Beta) pivot. For the maximum over $n$ observations, The author of \cite{Siotani-1959} applied the Bonferroni method to obtain the overall critical value.  
Since this derivation is analogous to that of (\ref{Grubbs}), we can view this technique as a generalization of the Grubbs test.

The performance of both univariate and multivariate methods described above is very dependent on the quality of the covariance estimator. In what follows, we also address the problem of estimating the covariance matrix in the presence of outliers. A broader overview of alternative covariance estimators can be found in \cite[Chapter~7]{Barnett-1998}.

\subsection{Setting an asymptotically exact threshold} \label{sec:AOT}
Most of the thresholding methods discussed above rely on distributional quantiles, but they do not specify an optimal choice of the probability level $\alpha$. The optimal choice of $\alpha$ depends on the number of outliers, which is usually unknown in practice. In what follows, we propose a thresholding approach that avoids the need to subjectively select the quantile probability.

\paragraph*{Sparsity conditions}
Assume that the anomalies occur in model (\ref{model}) only rarely so $\sum_{i=1}^n \eta_i = \lfloor n^{1-\beta} \rfloor$, where $\beta \in (0,1)$ is a~\textit{sparsity index}. The larger $\beta$ is, the less frequently anomalies occur. The index $\beta$ is considered unknown. Further, define sets $H_{n,\beta}$ and $H_{n,\beta}^\pm$ of $n$-dimensional indicators satisfying our sparsity condition as follows 
\begin{align*}
	H_{n,\beta} &= \left\{ \etab = (\eta_1,\ldots,\eta_n): \eta_i \in \{0,1\}, \sum_{i=1}^n \eta_i \leq c_1 n^{1-\beta} \right\}, \\
	H_{n,\beta}^\pm &= \left\{ \etab = (\eta_1,\ldots,\eta_n): \eta_i \in \{0,1\}, c_0 n^{1-\beta} \leq \sum_{i=1}^n \eta_i \leq c_1 n^{1-\beta} \right\},
\end{align*}
for some constants $0 < c_0 < 1 <c_1 < \infty.$ Clearly, $H_{n,\beta}^\pm \subset H_{n,\beta}$.

\paragraph*{Distance measure}
As mentioned in Section \ref{sec:existing_methods}, the multivariate observation $\Xb$ has to be transformed to a single number capturing the magnitude of individual components. Unlike the methods described in Section \ref{sec:existing_methods_multivariate}, which use the Mahalanobis distance (\ref{Mahalanobis}), 
we reduce $\Xb$ to a scalar quantity $R(\Xb)$ by function 
\begin{equation} \label{def:R}
	R(\Xb) = \lVert L^{-1}\Xb\rVert_1,
\end{equation}
where $L$ is a Cholesky factor of $\Sigma$, so $LL^\top = \Sigma$, and $\lVert \cdot \rVert_1$ is the $\ell_1$-norm. Hence, $R(\Xb)$ is a sum of absolute values of entries of the standardized version of $\Xb$. We can then order our sample $\Xb_1,\ldots,\Xb_n$ in terms of the values $R(\Xb_i)$, $i=1,\ldots,n$. Observations $\Xb_i$ corresponding to the largest values $R(\Xb_i)$ are candidates for being classified as anomalous. Due to multiplying each observation by the Cholesky factor $L$, the detection mechanism respects the dependence of the components of the random vector. 

\paragraph*{Covariance estimator} 
For an unknown $\Sigma$, which is usually the case in practice, we replace~(\ref{def:R}) by
\begin{equation} \label{def:Rstar}
	\hat{R}(\Xb) = \lVert \hat{L}^{-1}\Xb \rVert_1,
\end{equation}
where $\hat{L}$ is such that $\hat{L}\hat{L}^\top = S$, where $S$ is an estimator of $\Sigma$ based on $\Xb_1,\ldots,\Xb_n$. In order to be able to prove the asymptotical optimality of the new detection method in Theorem \ref{theorem:upper_bound} below, we need the estimator $S$ to possess the following property. 
Assume that the elements of $S$ allow for the decomposition 
\begin{equation} \label{cond_S}
	s_{jk} = \frac{1}{n} \sum_{i=1}^n g_{jk}(\Xb_i) + b_{jk}(n) + R_{jk},
\end{equation}
where 
\begin{itemize}
	\item[(A1)]  $g_{jk}:\mathbb{R}^p \to \mathbb{R}$ is a measurable function such that $\E g_{jk}(\Xb_1) = \sigma_{jk}$ and $\E |g_{jk}(\Xb_1) - \sigma_{jk}|^8 < \infty$,
	\item[(A2)] $b_{jk}(n) = O(n^{-1})$ is a deterministic bias, and
	\item[(A3)] $R_{jk}$ is a small centered random remainder, in particular, we require $\E R_{jk}=0$ and $\E |R_{jk}|^8 = O(n^{-4})$. 
\end{itemize}
Assumptions (A1)–(A3) are technical conditions required in the proof of Theorem~\ref{theorem:upper_bound} to ensure that $S$ converges to $\Sigma$ at a sufficiently fast rate.
Decompositions of the form (\ref{cond_S}) are common in the statistical literature. For instance, when $s_{jk}$ is a U-statistic generated by a~symmetric kernel and for a specific form of $g_{jk}, R_{jk}$ and $b_{jk}=0$, the representation (\ref{cond_S}) coincides with the Hoeffding decomposition (see, for example \cite[Chapter~11]{Vaart-1998}). Also, if $g_{jk} - \sigma_{jk}$ corresponds to the influence function of the covariance functional, then due to assumptions (A2) and (A3), the decomposition (\ref{cond_S}) can be interpreted as an influence function expansion (see, for example \cite[p. 58 and 278]{Vaart-1998}). While estimator decompositions are typically not written by explicitly isolating the three components $g_{jk},\;b_{jk}(n)$ and $R_{jk}$, as in (\ref{cond_S}), such a separation is essential here in order to obtain sharp asymptotic control over each part of the estimator.
Observe that a covariance estimator with components (\ref{cond_S}) is consistent (a proof is provided in Section A.1 in the Appendix). 

The decomposition (\ref{cond_S}) accommodates many of the commonly used covariance estimators, including centered and uncentered sample covariance, and certain robust estimators. 
For example, when $S_2=\frac{1}{n-1} \sum_{i=1}^n (\Xb_i - \boldsymbol{\mu})(\Xb_i - \boldsymbol{\mu})^\top$ is the standard sample covariance computed, for example, from some historical or training data with a known mean vector $\boldsymbol{\mu} = (\mu_1,\ldots,\mu_p)^\top$ and in the absence of outliers, then the components of the decomposition are $g_{jk}(\Xb_i) = (X_{ij}-\mu_j)(X_{ik}-\mu_k)$, $b_{jk}(n) = \frac{\sigma_{jk}}{n-1}$ and $R_{jk} = \frac{1}{n(n-1)}\sum_{i=1}^n \left( (X_{ij}-\mu_j)(X_{ik}-\mu_k) - \sigma_{jk} \right)$. Verification of assumptions (A1)--(A3) in this case is straightforward. 

When no historical or training data are available, and $\Sigma$ must be estimated from data that may contain outliers, a robust covariance estimator, such as a Huber-type estimator, can be employed.
The M-estimator $\hat{\Sigma}$ with Huber-type weights (see, e.g., \cite{Maronna-1976, Zu-2010}) based on a sample $\Xb_1,\ldots,\Xb_n$ with known mean $\boldsymbol{\mu}=(\mu_1,\ldots,\mu_p)^\top$ is defined as a~solution of the equation
\begin{align} \label{def:Huber_scatter}
	\hat{\Sigma} &= \frac{1}{n} \sum_{i=1}^n u_2 (M_{\hat{\Sigma}}(\Xb_i)) (\Xb_i - \boldsymbol{\mu}) (\Xb_i - \boldsymbol{\mu})^\top,
\end{align}
where $M_{\hat{\Sigma}}(\Xb_i) = (\Xb_i - \boldsymbol{\mu}) \hat{\Sigma}^{-1}(\Xb_i - \boldsymbol{\mu})^\top$ is the sample Mahalanobis distance of $\Xb_i$ from~$\boldsymbol{\mu}$ (cf. (\ref{Mahalanobis})), and $u_2$ is a Huber-type weight function of the following form
\begin{align}
	\begin{split}
		u_2(d) &= \frac{1}{c}, \quad d \leq r \\
		&= \frac{r^2}{cd}, \quad d>r, 
	\end{split} \label{def:Huber_u2}
\end{align}
for some $r>0$ given by $\operatorname{P}\left(\chi_p^2 > r^2\right) = \kappa$ such that $\kappa$ is the proportion of cases one wants to downweight. The quantity $c$ is a normalizing constant chosen so that the robust estimator $\hat{\Sigma}$ is unbiased, i.e.  $\E \left[ u_2(M_{\Sigma}(\Xb))(\Xb-\boldsymbol{\mu})(\Xb - \boldsymbol{\mu})^\top \right] = \Sigma$, under the assumption that $\Xb \sim N_p(\boldsymbol{\mu},\Sigma)$. At convergence, the $j,k$-th entry of the robust covariance estimator can be written in the form (\ref{cond_S}) with
\begin{align}
	g_{jk}(\Xb_i) &= u_2(M_{\Sigma}(\Xb_i))(X_{ij} - \mu_{j})(X_{ik}-\mu_k) \label{Huber_g}\\
	\begin{split}
		R_{jk} &= \frac{1}{n} \sum_{i=1}^n \left( u_2(M_{\hat{\Sigma}}(\Xb_i)) - u_2(M_{\Sigma}(\Xb_i)) \right) (X_{ij} - \mu_{j})(X_{ik}-\mu_k)- \\ 
		& \quad \quad - \frac{1}{n} \E \left[ \sum_{i=1}^n \left( u_2(M_{\hat{\Sigma}}(\Xb_i)) - u_2(M_{\Sigma}(\Xb_i)) \right) (X_{ij} - \mu_{j})(X_{ik}-\mu_k) \right] 
	\end{split} \label{Huber_R} \\
	b_{jk}(n) &= \frac{1}{n} \E \left[ \sum_{i=1}^n \left( u_2(M_{\hat{\Sigma}}(\Xb_i)) - u_2(M_{\Sigma}(\Xb_i)) \right) (X_{ij} - \mu_{j})(X_{ik}-\mu_k) \right], \label{Huber_b}
\end{align}
where $M_{\Sigma}(\Xb_i)$ is the Mahalanobis distance of $\Xb_i$ from $\boldsymbol{\mu}$ as in (\ref{Mahalanobis_true}). Note that quantities (\ref{Huber_g})--(\ref{Huber_b}) satisfy assumptions (A1)--(A3) (for a rigorous proof, see Section A.2 in the Appendix). 

\paragraph*{Asymptotically exact anomaly detection}
As specified in model (\ref{model}), we hereafter assume $\boldsymbol{\mu} = \boldsymbol{0}$. The estimator of vector~$\etab$, which selects the anomalous observations in model (\ref{model}), is built on the following tail probability property of multivariate normal distribution (a proof is provided in Section B in the Appendix).

\begin{lemma} \label{lemma:max_norm_N0I}
	Let $\Zb_1,\ldots, \Zb_n$ be a random sample from $N_p(\boldsymbol{0}, I)$. Then, for any $\Delta > 0$,
	\begin{equation} \label{prob:N0I}
		\operatorname{P}\left(\max_{j=1,\ldots,n} \lVert \Zb_j \rVert_1 \geq \sqrt{2 (p+\Delta) \log n}\right) = \smallO(1), \quad n \to \infty.
	\end{equation}
\end{lemma}
In other words, the probability that the maximum $\ell_1$-norm of a standard normal random vector exceeds the given threshold tends to zero with growing $n$.

Based on this lemma, we propose a detector of the form
\begin{equation} \label{def:eta}
	\hat{\etab} = (\hat{\eta}_1,\ldots,\hat{\eta}_n), \quad \quad \hat{\eta}_i = \mathds{1}\left(\hat{R}(\Xb_i) > \sqrt{2(p+\delta) \log n}\right),
\end{equation}
where $\delta>0$ as in Lemma \ref{lemma:max_norm_N0I} and we assume that 
\begin{equation} \label{cond:delta}
	\delta=\delta(n) \to 0 \quad \text{ and } \quad  \delta \log n \to \infty \text{ as } n \to \infty.  
\end{equation}
Notice that the estimator (\ref{def:eta}) does not depend on $\beta$ and that for $p=1$ it reduces to the corresponding result of \cite{Butucea-2018, Cui-2014}. The following theorem shows that under some mild conditions, the expected number of errors in our identification in a random sample of size~$n$ tends to zero with $n \rightarrow \infty $.  

\begin{theorem} \label{theorem:upper_bound}
	Let $\beta \in (0,1)$ and $p \in \mathbb{N}$ be fixed numbers. If the parameters $\ab=\ab(n)$ and $\Sigma=LL^\top >0$ in model~(\ref{model}) satisfy the condition
	\begin{equation} \label{cond:a_UB}
		\liminf_{n \to \infty} \frac{\lVert L^{-1}\ab \rVert_1}{\sqrt{p\log n}} > \sqrt{2}(1+\sqrt{1-\beta}),
	\end{equation}
	then the estimator (\ref{def:eta})--(\ref{cond:delta}) satisfies
	\begin{equation} \label{eq:Hamming_risk_T2}
		\liminf_{n \to \infty} \sup_{\etab \in H_{n,\beta}} \E_{\etab} |\etab - \hat{\etab}| = 0.
	\end{equation}
\end{theorem}

Equivalently, the estimator (\ref{def:eta})--(\ref{cond:delta}) provides exact detection of the anomalous observations. The quantity $\E_{\etab} |\etab - \hat{\etab}|$ is called Hamming risk, and it represents the expected number of misclassified observations. 

\begin{remark}
	Theorem \ref{theorem:upper_bound} remains valid also for an estimator of $\etab$ based on $R$ defined in (\ref{def:R}). Moreover, the result is also valid for $\boldsymbol{\mu} \neq \boldsymbol{0}$ by considering the centered observations $\boldsymbol{Y}_i = \Xb_i - \boldsymbol{\mu}$, $i = 1,\ldots,n$, in place of $\boldsymbol{X}_i$, $i = 1,\ldots,n$. 
\end{remark}

The next theorem provides conditions under which such an exact identification cannot be achieved for any detector based on $\Xb_1,\ldots, \Xb_n$. Note that since $H^\pm_{n,\beta} \subset H_{n,\beta}$, the detector $\hat{\etab}$ is also exact with respect to the maximum risk when the maximum is taken over $H^\pm_{n,\beta}$. 

\begin{theorem} \label{theorem:lower_bound}
	Let $\beta \in (0,1)$ and $p \leq \lfloor \frac{8}{\beta} (1+\sqrt{1-\beta})-4 \rfloor$ be fixed numbers. If the parameters $\ab=\ab(n)$ and $\Sigma=LL^\top$ in model~(\ref{model}) satisfy the condition
	\begin{equation} \label{cond:a_LB}
		\liminf_{n \to \infty} \frac{\lVert L^{-1}\ab \rVert_1}{\sqrt{p\log n}} < \sqrt{2}(1+\sqrt{1-\beta}),
	\end{equation}
	then
	\begin{equation}
		\liminf_{n \to \infty} \inf_{\tilde{\etab}} \sup_{\etab \in H^\pm_{n,\beta}} \E_{\etab} |\etab - \tilde{\etab}| > 0,
	\end{equation}
	where the infimum is taken over all estimators $\tilde{\etab}$ of $\etab$ based on $\Xb_1,\ldots,\Xb_n$.
\end{theorem}

Proofs of Theorems \ref{theorem:upper_bound} and \ref{theorem:lower_bound} are given in Section B in the Appendix. 

\begin{remark}
	For $p=1$, known $\sigma^2$ and $a>0$, Theorems \ref{theorem:upper_bound} and \ref{theorem:lower_bound} coincide with the results of \cite{Butucea-2018} and \cite{Cui-2014}. 
\end{remark}

\begin{remark}
	The condition $p \leq \lfloor \frac{8}{\beta} (1+\sqrt{1-\beta})-4 \rfloor$ is necessary in the proof of Theorem \ref{theorem:lower_bound}, and it says that the closer $\beta$ is to 1, the closer the maximal admissible value of $p$ approaches 4. The restriction on $p$ ensures that detection cannot rely on high-dimensional aggregation effects. If the dimension was too large, many weak components could accumulate and together form a detectable signal. To keep the minimax risk bounded away from zero for larger values of $p$, condition~(\ref{cond:a_LB}) would need to be strengthened so that anomalies remain difficult to detect even when more components are present. Another possibility would be to restrict the subset of components allowed to carry the mean shift and thus to contain anomalies. This phenomenon, namely the need for sparser (or weaker) anomalies as the dimension increases, was also observed by \cite{Laurent-2018, Verzelen-2017}.
\end{remark}

\subsection{Thresholding methods for time series} \label{sec:time_series}
The proposed method, as well as the approaches described in Section \ref{sec:existing_methods} and \ref{sec:existing_methods_multivariate}, can be applied directly to detect anomalies in a centered Gaussian random sample. However, observations $\Xb_1,\ldots,\Xb_n \sim N_p(\bm{0},\Sigma)$ may also arise as residuals from a variety of statistical models, including regression models, Kalman filters, and other filtering algorithms, as well as time series models.
In that case, anomaly detection typically proceeds as follows:
\begin{enumerate}
	\item A model is fitted to the observed data.
	\item The covariance matrix of the residuals is estimated, usually as a part of the fitting procedure.
	\item For each residual $\mathbf{X}_i$, the indicator $\hat{\eta}_i$ defined in~(\ref{def:eta})--(\ref{cond:delta}) is computed to determine whether the observation is anomalous. 
\end{enumerate}

In an adaptive (online) setting, each incoming observation is classified using an updated threshold. The matrix $\hat{L}$ used to compute $\hat{\boldsymbol{\eta}}$ can be updated at each step, or it may be estimated once from a sufficiently large training sample and kept fixed thereafter.


\section{Simulation study}\label{Sec:simul}

The performance of the detector $\hat{\etab}$ defined in (\ref{def:eta})--(\ref{cond:delta}) is demonstrated by a set of simulations in a multivariate ($p=4$) case. Its performance is compared with other commonly used methods for multivariate anomaly detection presented in Section~\ref{sec:existing_methods_multivariate}. In both thresholds (\ref{qBeta}) and (\ref{Siotani}), we chose $\alpha = 0.05$, the value $\delta$ in (\ref{def:eta}) was chosen as $\delta = 1/\sqrt{\log n}$ in order to satisfy condition (\ref{cond:delta}). Besides method comparison, the statement of Theorem~\ref{theorem:upper_bound} is illustrated as well. The simulation study was conducted in the statistical software~\texttt{R}. 
To ensure comparability, the same covariance estimator was used across all detection methods. The sample covariance matrix is not a good estimator in the presence of outliers, so the covariance matrix was estimated by the robust Huber-type estimator (\ref{def:Huber_scatter})--(\ref{def:Huber_u2}) with $\boldsymbol{\mu} = \boldsymbol{0}$, and its practical computation was based on an \texttt{R}~function \texttt{cov\_Huber()} from package \texttt{robmed} with $r$ equal to 95\%-quantile of $\chi^2_p$ distribution, $\kappa = 0.05$ and $c = (-2f_{\chi^2_p}(r)+\kappa)\frac{r}{p} + 1- \kappa$, where $f_{\chi^2_p}$ stands for the density of $\chi^2_p$, which is the default setting in  \texttt{cov\_Huber()} (for details, see \cite{Zu-2010}).


\subsection{Methods comparison}
In the first simulation study, we illustrate that the Hamming risk $\E_{\etab} |\etab - \hat{\etab}|$ of the proposed asymptotically exact detector (\ref{def:eta})--(\ref{cond:delta}) converges to zero as stated in Theorem~\ref{theorem:upper_bound}. We also compare the magnitude of the estimated Hamming risk with that of the methods described in Section~\ref{sec:existing_methods_multivariate}. As mentioned previously, we consider a~multivariate case with $p=4$. In each of 100 runs of this simulation,  we draw a random sample $\Xb_1,\ldots, \Xb_n,$ of size $n \in \{300, 310, 320, \ldots, 1500 \}$, where each observation $\Xb_i, i=1,\ldots,n$, follows $N(\eta_i \ab_n, \Sigma)$ with parameters
\begin{align}  \label{sim_distr}
	\ab_n&= \begin{bmatrix}
		a_{n1} \\
		a_{n2} \\
		a_{n3} \\
		a_{n4}
	\end{bmatrix} &
	\Sigma &= \begin{bmatrix}
		1.12& -0.44 & 0.17& -0.73\\
		-0.44 & 1.97 &-0.26 &-0.42\\
		0.17& -0.26 & 0.56& -0.12\\
		-0.73 &-0.42& -0.12 & 1.13
	\end{bmatrix},
\end{align}
where $\ab_n$ depends on $n$ as seen in Figure \ref{fig:a_n}. Values of $a_{n1},a_{n2}$ and $a_{n3}$ were sampled from the uniform distribution on the interval $[-c_n,c_n]$, where $c_n=\sqrt{2p\log(n)}(1+\sqrt{1-\beta})/3$, and $a_{n4} = 1.3(c_n - \sum_{k=1}^3 |a_{nk}|)$, so the condition~(\ref{cond:a_UB}) was satisfied. The value of $\ab_n$ was kept the same in all simulation runs for each value of $n$. The anomalies (indicated by $\eta_i=1$) were placed at locations randomly sampled from $\{1,\ldots,n \}$ and their total number was equal to the integer part of $n^{1-\beta}$ with the sparsity parameter $\beta=0.4$. 
Recall that $\Sigma$ was estimated by the Huber-type estimator $\hat{\Sigma}$ defined by (\ref{def:Huber_scatter})--(\ref{def:Huber_u2}) and note that the stability condition $\lVert \Sigma^{-1} \rVert_2 \lVert \Sigma - \hat{\Sigma} \rVert_2 <1$ was satisfied for all simulation runs under all settings. 

\begin{figure}[h]
	\centering
	\includegraphics[width=\textwidth]{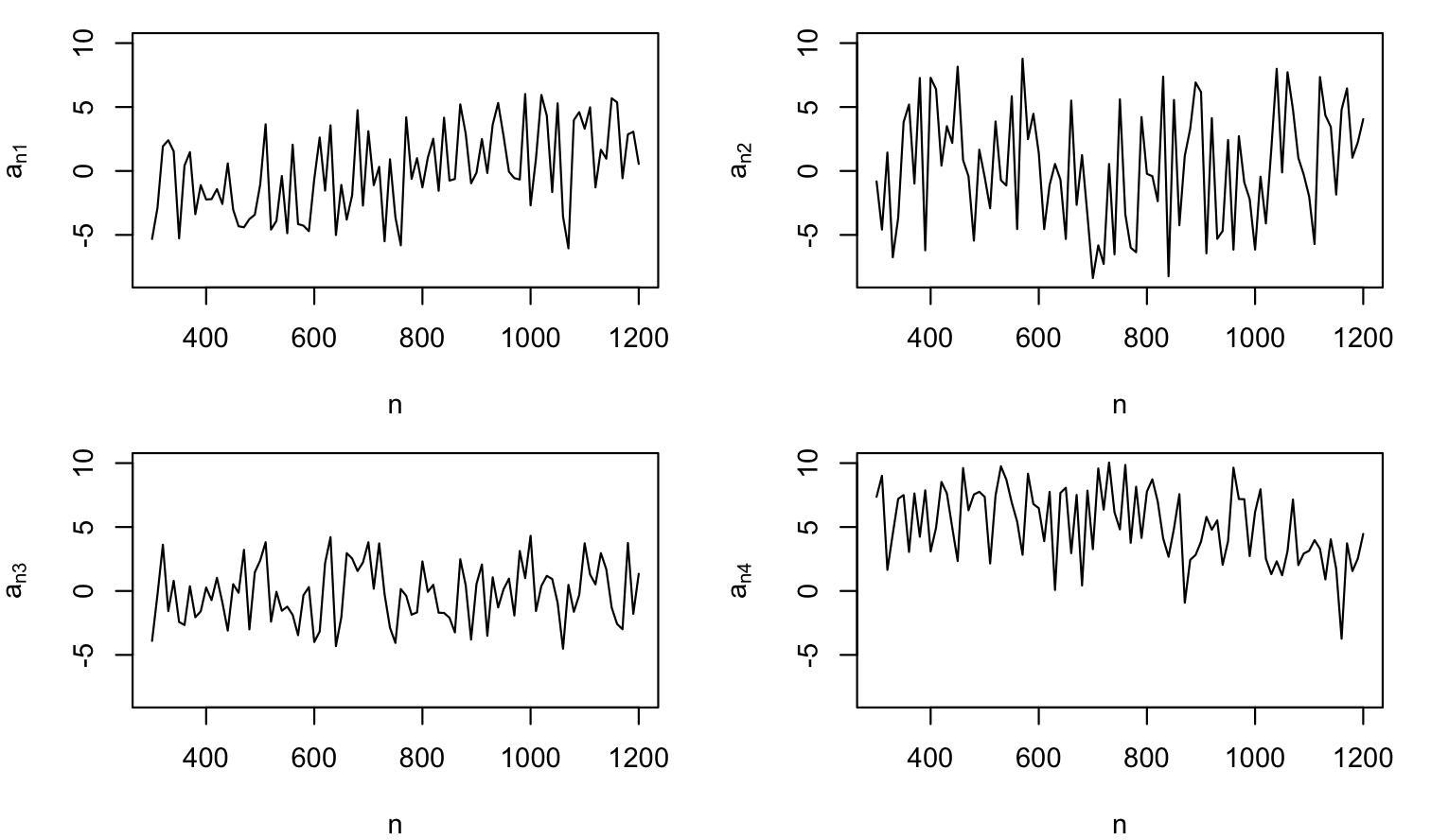}
	\caption{Components of the expected value $\ab_n = (a_{n1},a_{n2},a_{n3},a_{n4})^\top$ of the anomalies depending on the sample size $n$.}
	\label{fig:a_n}
\end{figure}

For each sample, all the detection methods from Section~\ref{sec:existing_methods_multivariate} were applied, and the total number of misclassified items (false negatives plus false positives) was computed for each method. The estimate of the Hamming risk $\E_\etab |\etab - \hat{\etab}|$ for each $n$ was calculated as the ratio of the number of misidentified observations and the number of replications (100). The resulting estimated Hamming risks for all methods are shown in Figure \ref{fig:Hamming_risk}. 
\begin{figure}[h]
	\centering
	\includegraphics[width=\textwidth]{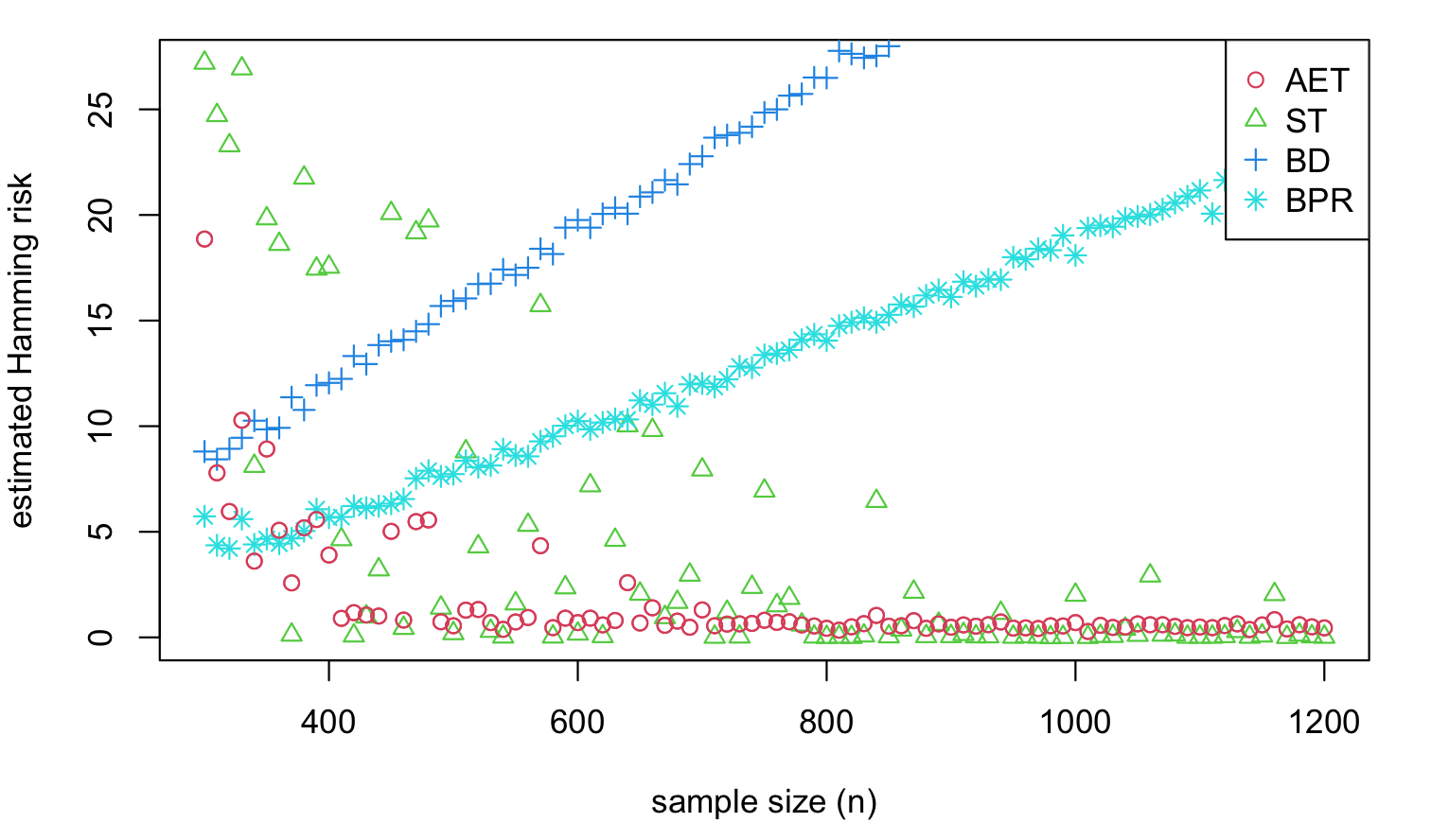}
	\caption{Average number of misclassified observations (false positives + false negatives), which estimates the Hamming risk, based on 100 replications. The samples were drawn from normal distributions with parameters given by (\ref{sim_distr}), where $\ab_n$ is displayed in Figure \ref{fig:a_n}. Covariance is considered unknown and estimated by the Huber-type estimator (\ref{def:Huber_scatter})--(\ref{def:Huber_u2}). The detection methods: AET: proposed asymptotically exact threshold (\ref{def:eta})--(\ref{cond:delta}), BD: threshold based on beta distribution (\ref{qBeta}), BPR: boxplot rule threshold (\ref{boxplot_rule}), ST: Siotani threshold (\ref{Siotani}). The simulation illustrates the statement of Theorem \ref{theorem:upper_bound}.}
	\label{fig:Hamming_risk}
\end{figure}

In Figure \ref{fig:Hamming_risk}, we observe that the estimated Hamming risk of the proposed asymptotically exact threshold (AET) defined in (\ref{def:eta})--(\ref{cond:delta}) decreases with increasing sample size $n$, in agreement with Theorem \ref{theorem:upper_bound}. Overall, the proposed AET method performs very well, yielding consistently low detection error. For smaller samples ($n<400$), it is occasionally outperformed by the threshold (\ref{qBeta}) based on beta distribution (BD) and the boxplot rule (BPR) method (\ref{boxplot_forM}); however, the estimated Hamming risk of these two methods increases as $n$ grows, and for $n>600$, their performance deteriorates substantially. The Siotani threshold (ST) defined in~(\ref{Siotani}) also exhibits a decreasing number of misidentified components, but it performs worse than AET for $n<800$. For larger sample sizes, it achieves very good detection accuracy — often slightly better than AET — although isolated peaks of higher error suggest some instability in specific scenarios. For such large $n$, AET appears to be highly robust and stable.

\subsection{Detection in one particular realization of a Gaussian random sample} \label{subsec:One_realization}

In the second simulation, the performance of our method and methods from Section \ref{sec:existing_methods_multivariate} is demonstrated using a single realization of a Gaussian random sample of size $n = 800$. For comparison, the univariate methods from Section \ref{sec:existing_methods} are also considered, applied to each dimension separately. 
We consider two scenarios differing in the structure of the mean shift of the anomalies. In the first scenario, the anomalies are caused by a large magnitude of one component (the mean shift is set to the value $\boldsymbol{a}_1$ below). In the second scenario, the anomalies are caused by mild deviations in all components (the mean shift equals $\boldsymbol{a}_2$ below) that might be missed by the univariate methods. 
The sparsity index was set to $\beta = 0.5$ and hence the number of anomalies was $\lfloor n^{1-\beta} \rfloor = 28$.

Specifically, we consider a multivariate random sample $\Xb_1,\ldots,\Xb_n$  of size $n = 800$ with each observation $\Xb_j=(X_{j1}, \ldots, X_{j4})^\top$ coming from the distribution $N_4(\eta_j \ab, \Sigma)$, where the mean shift vector $\ab$ was chosen as either
\begin{align}  \label{def:norm_simul_multiD}
	\ab_1= \begin{bmatrix}
		1.38 \\ 
		1.20 \\
		0.95\\
		5.33 
	\end{bmatrix}, \qquad \text{or} \qquad 
	\ab_2= \begin{bmatrix}
		4.14 \\ 
		3.61 \\
		2.87\\
		-2.87 
	\end{bmatrix},
\end{align}
and the covariance matrix $\Sigma$ was the same as in (\ref{sim_distr}).
The parameter values were chosen so that the condition (\ref{cond:a_UB}) is fulfilled. Throughout the simulation study, the covariance matrix was considered unknown and estimated by the Huber-type covariance estimator (\ref{def:Huber_scatter})--(\ref{def:Huber_u2}) as described at the beginning of Section \ref{Sec:simul}. 

In total, eight anomaly detection methods ware compared. We considered four multivariate methods as in the previous simulation study: BD: threshold (\ref{qBeta}) based on beta distribution, BPR: boxplot rule (\ref{boxplot_forM}), AET: asymptotically exact threshold (\ref{def:eta})--(\ref{cond:delta}), ST: Siotani's threshold~(\ref{Siotani}), and four univariate methods, applied separately to each component: AET1: asymptotically exact threshold (\ref{def:eta})--(\ref{cond:delta}) with $p=1$, GT: Grubbs test (\ref{Grubbs}), 3s: 3$\sigma$ rule (\ref{3sigma}), and BPR1: boxplot rule (\ref{boxplot_rule}). 
Before applying the univariate methods to each component, the multivariate sample was standardized using the same Huber-type covariance matrix estimate as that employed in the multivariate detection methods.
Note that the univariate version of the proposed method given by (\ref{def:eta})--(\ref{cond:delta}) with $p = 1$ generalizes the result of \cite{Butucea-2018} to the case of unknown variance.

The methods were compared in the total number of misclassified observations (false negatives + false positives). 
The results are listed in Table \ref{tab:N_I_multi} for multivariate methods and Table~\ref{tab:N_I_uni} for univariate methods.

\begin{table}[h!]
	\centering
	\begin{tabular}{l|cc}
		\toprule
		Method & $\mathbf{a}_1$ & $\mathbf{a}_2$ \\
		\midrule
		AET  & 0  & 0 \\
		ST   & 0  & 0 \\
		BD   & 32 (32/0) & 28 (28/0)\\
		BPR  & 18  (18/0) & 12 (12/0) \\
		\bottomrule
	\end{tabular}
	\caption{Number of incorrectly classified observations obtained by multivariate methods in a random sample $\Xb_1,\ldots,\Xb_n$ drawn from $N_p(\ab_1, \Sigma)$ and $N_p(\ab_2, \Sigma)$ with $n=800$, $p=4$. The numbers in parentheses denote false positives / false negatives.}
	\label{tab:N_I_multi}
\end{table}

\begin{table}[ht]
	\centering
	\begin{tabular}{l|cccc|cccc}
		\toprule
		& \multicolumn{4}{c|}{$\mathbf{a}_1$}
		& \multicolumn{4}{c}{$\mathbf{a}_2$} \\
		\cmidrule(lr){2-5} \cmidrule(lr){6-9}
		Method
		& $k=1$ & $k=2$ & $k=3$ & $k=4$
		& $k=1$ & $k=2$ & $k=3$ & $k=4$ \\
		\midrule
		AET1 & 28 (0/28) & 28 (0/28) & 28 (0/28) & 0 & 19 (0/19) & 23 (0/23) & 25 (0/25)& 26 (0/26)\\
		GT   & 28 (0/28) & 28 (0/28)& 28 (0/28) & 0 & 19 (0/19) & 23 (0/23) & 25 (0/25)& 26 (0/26)\\
		3s & 29 (2/27)& 26 (0/26) & 27 (1/26)& 0 & 8 (2/6)& 14 (0/14) & 18 (0/18)& 19 (1/18) \\
		BPR1 & 28 (3/25) & 26 (0/26) & 29 (3/26) & 7 (7/0) & 6 (3/3)& 14 (0/14) & 21 (5/16)& 20 (8/12)\\
		\bottomrule
	\end{tabular}
	\caption{Number of incorrectly classified observations in each component $k=1,\ldots,p$ obtained by univariate methods for a sample $\Xb_1,\ldots,\Xb_n$ drawn from $N_p(\ab_1, \Sigma)$ and $N_p(\ab_2, \Sigma)$ with $n=800$, $p=4$. The numbers in parentheses denote false positives / false negatives.}
	\label{tab:N_I_uni}
\end{table}

The sample $\Xb_1,\ldots,\Xb_n \sim N_p(\ab_t,\Sigma)$ is displayed in Figure~\ref{fig:sim_1realization_a1} for $t=1$ and in Figure~\ref{fig:sim_1realization_a2} for $t=2$. For clarity, only the subset $(X_{600,k},\ldots,X_{700,k})$ of each component $k=1,\ldots,p$ is shown, which keeps the figures uncluttered while ensuring that several anomalies are visible.

In the first scenario (see Figure \ref{fig:sim_1realization_a1} and the first column of Table \ref{tab:N_I_multi}), we see that AET and ST classify all observations correctly. Methods BPR and BD performs substantially worse - in both cases, all misclassified observations were false positives.  These results are consistent with Figure \ref{fig:Hamming_risk}. As expected, the univariate methods (see Figure \ref{fig:sim_1realization_a1} and the left-hand part of Table~\ref{tab:N_I_uni}) detected only anomalies in the fourth component of $\Xb$, where the largest mean shift was present.

In the second scenario of the mild deviations in all components of $\Xb$ (see Figure \ref{fig:sim_1realization_a2} and the right-hand part of Table \ref{tab:N_I_uni}), the univariate methods AET1 and GT are able to identify only a~few anomalies, but most of them remained undetected. Methods 3s and BPR1 perform a~bit better but the error is still substantial. This is not surprising, as univariate methods cannot identify composite anomalies. The performance of the multivariate methods (see Figure \ref{fig:sim_1realization_a2} and the second column of Table \ref{tab:N_I_multi}) is similar to that in the first scenario and in accordance with the results displayed in Figure \ref{fig:Hamming_risk}.

\begin{figure}[p]
	\centering
	\begin{subfigure}[t]{0.98\textwidth}
		\centering
		\includegraphics[width=\textwidth]{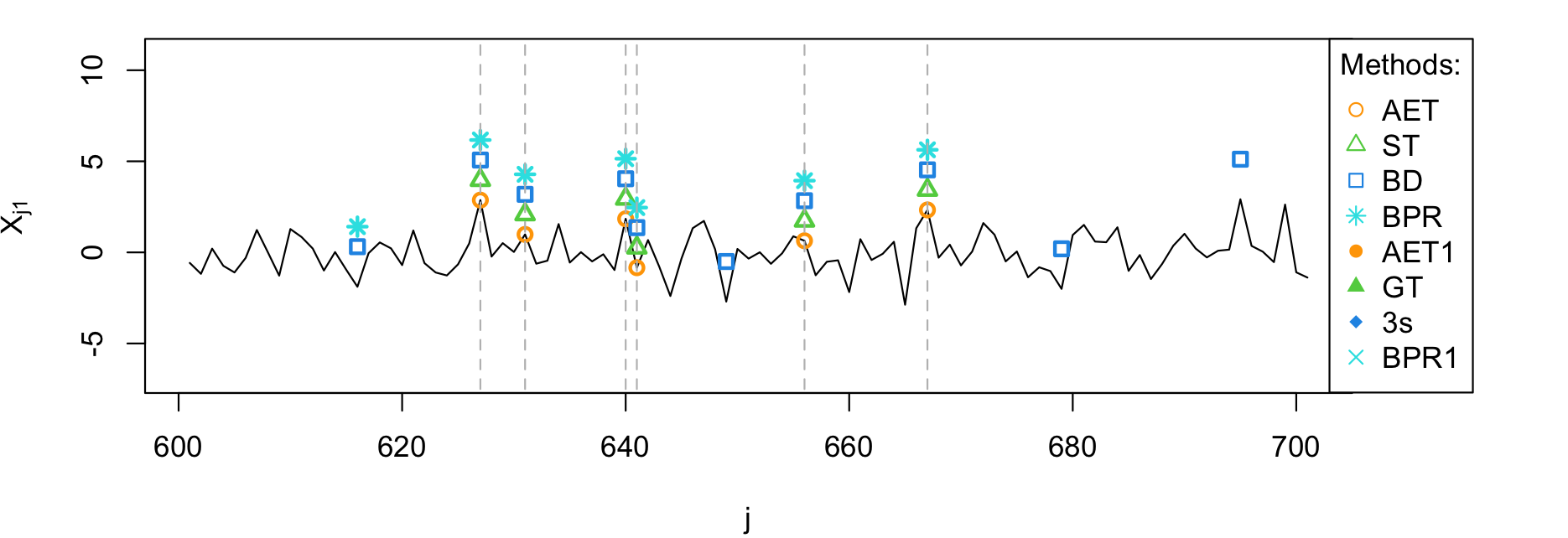}
	\end{subfigure}
	~ \vspace{-8mm} \\
	\begin{subfigure}[t]{0.98\textwidth}
		\centering
		\includegraphics[width=\textwidth]{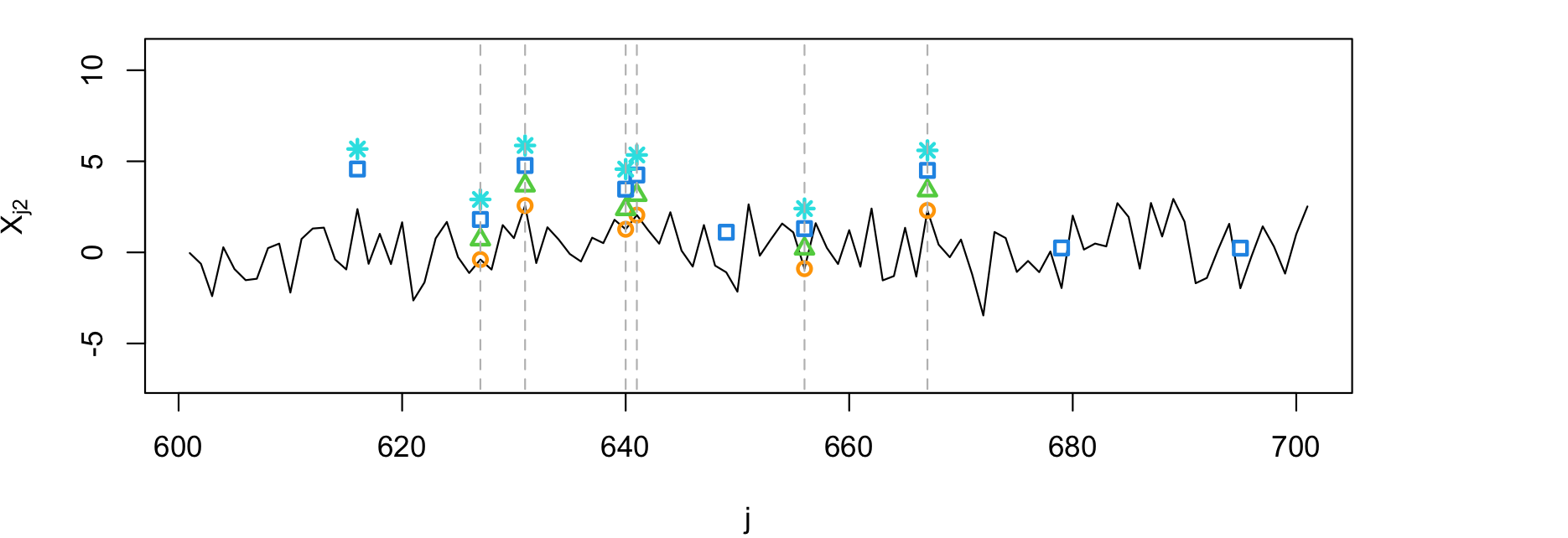}
	\end{subfigure}
	~ \vspace{-8mm} \\
	\begin{subfigure}[t]{0.98\textwidth}
		\centering
		\includegraphics[width=\textwidth]{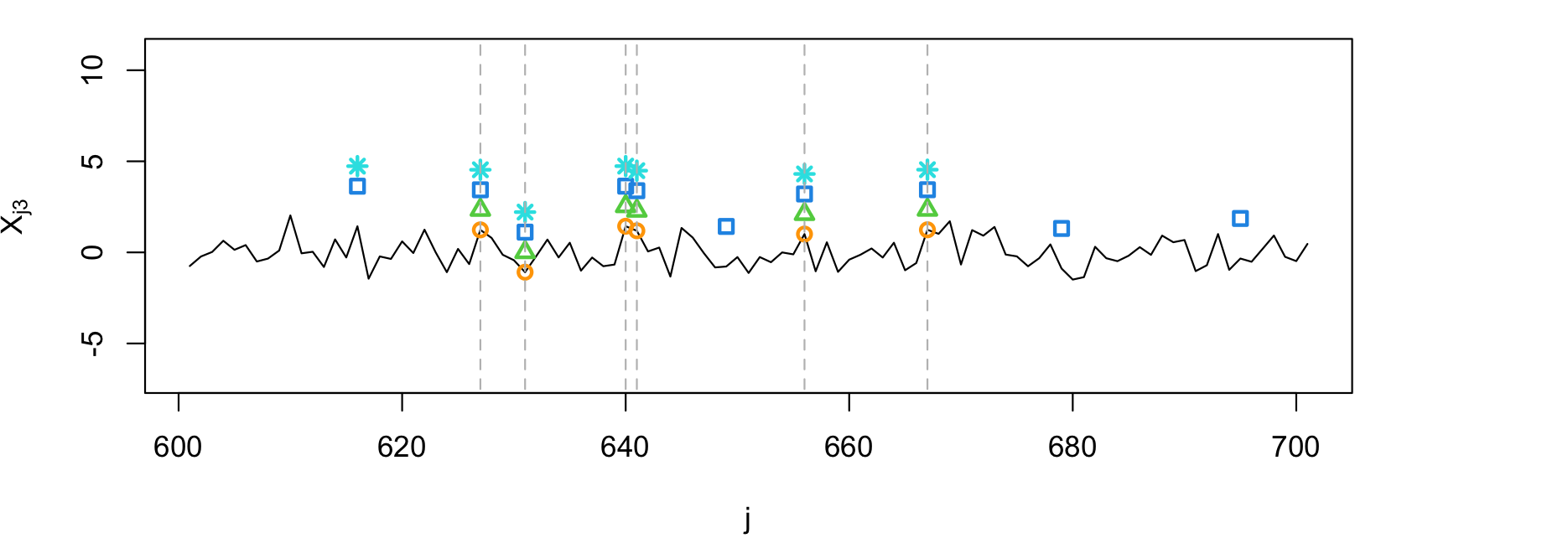}
	\end{subfigure}
	~ \vspace{-8mm} \\
	\begin{subfigure}[t]{0.98\textwidth}
		\centering
		\includegraphics[width=\textwidth]{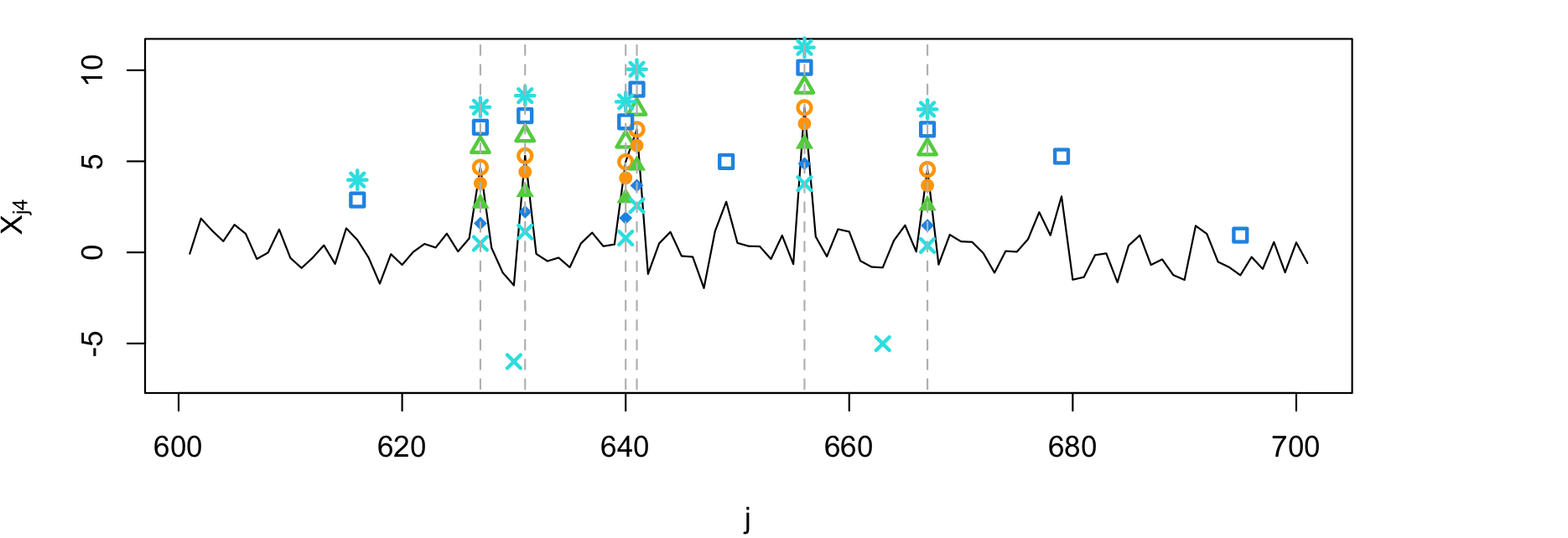}
	\end{subfigure}
	\caption{Detection in a Gaussian sample $\{\Xb_j = (X_{j1},\ldots,X_{j4})^\top, j=1,\ldots,800\},$ with mean $\boldsymbol{a}_1$ and covariance specified in (\ref{def:norm_simul_multiD}). Compared multivariate methods: BD: threshold (\ref{qBeta}) based on beta distribution, BPR: boxplot rule (\ref{boxplot_forM}), AET: asymptotically exact threshold (\ref{def:eta})--(\ref{cond:delta}), ST: Siotani's threshold (\ref{Siotani}), and univariate methods: AET1: asymptotically exact threshold (\ref{def:eta})--(\ref{cond:delta}) with $p=1$, GT: Grubbs test (\ref{Grubbs}), 3s: 3$\sigma$ rule (\ref{3sigma}), and BPR1: boxplot rule (\ref{boxplot_rule}). The true anomalies are marked by a dashed gray line.}
	\label{fig:sim_1realization_a1}
\end{figure}

\begin{figure}[p]
	\centering
	\begin{subfigure}[t]{0.98\textwidth}
		\centering
		\includegraphics[width=\textwidth]{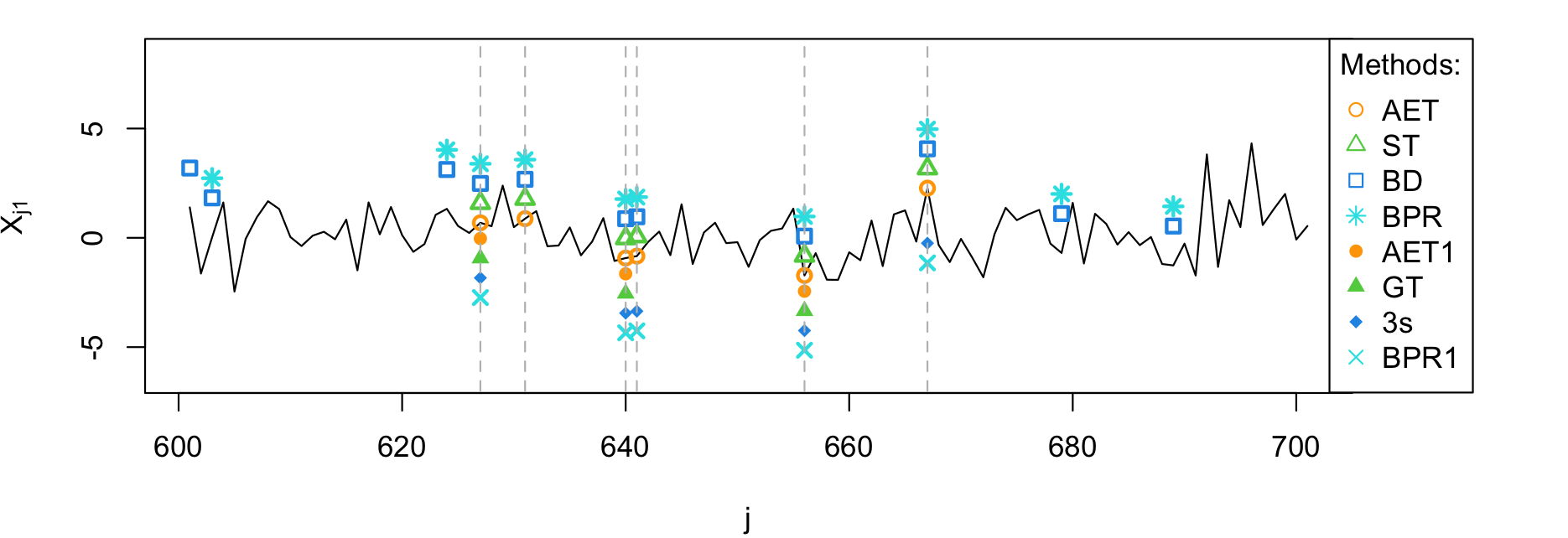}
	\end{subfigure}
	~ \vspace{-8mm} \\
	\begin{subfigure}[t]{0.98\textwidth}
		\centering
		\includegraphics[width=\textwidth]{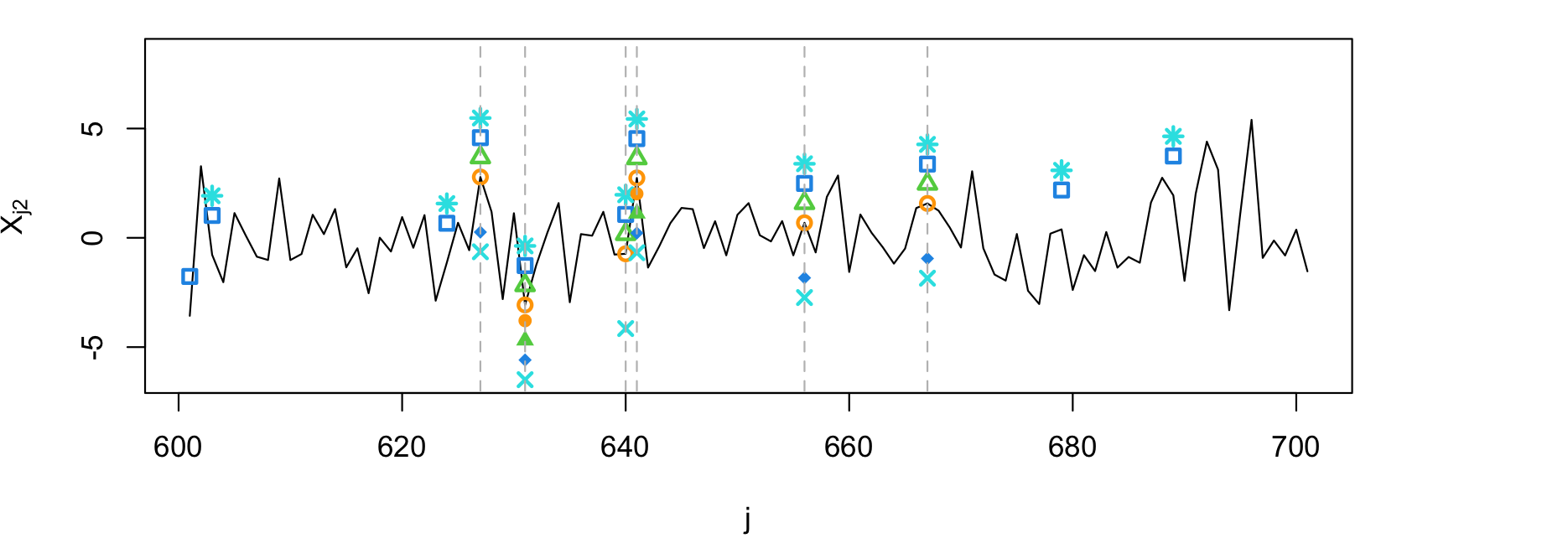}
	\end{subfigure}
	~ \vspace{-8mm} \\
	\begin{subfigure}[t]{0.98\textwidth}
		\centering
		\includegraphics[width=\textwidth]{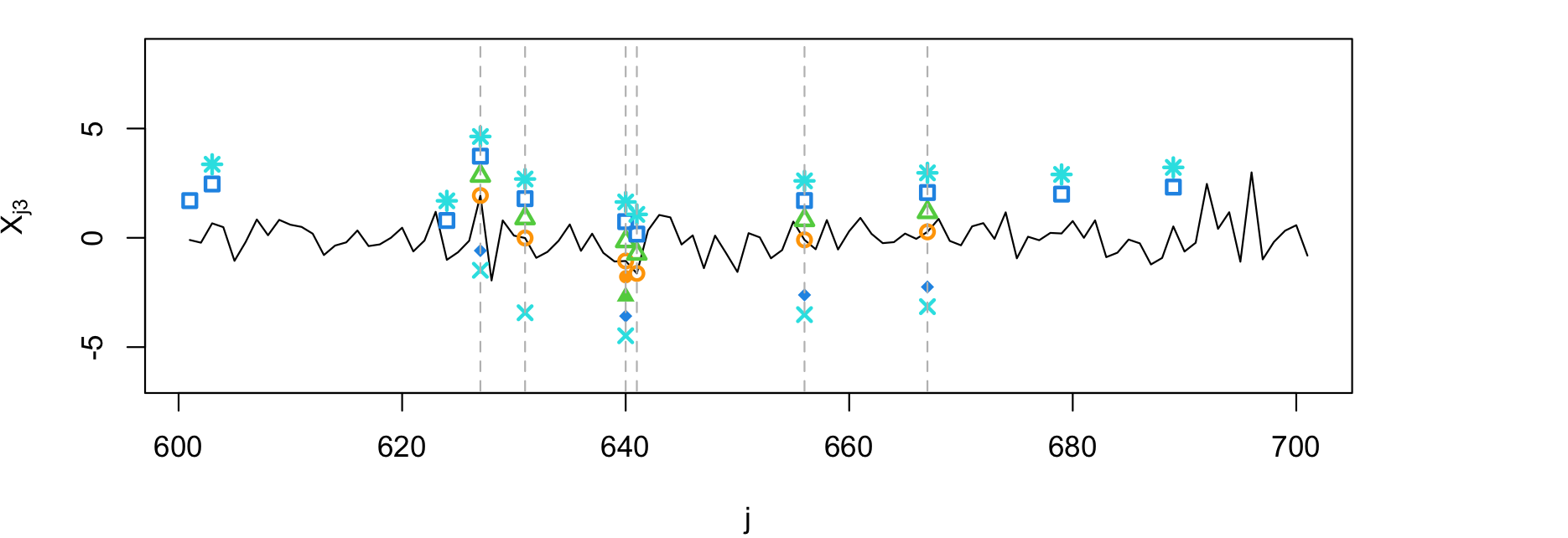}
	\end{subfigure}
	~ \vspace{-8mm} \\
	\begin{subfigure}[t]{0.98\textwidth}
		\centering
		\includegraphics[width=\textwidth]{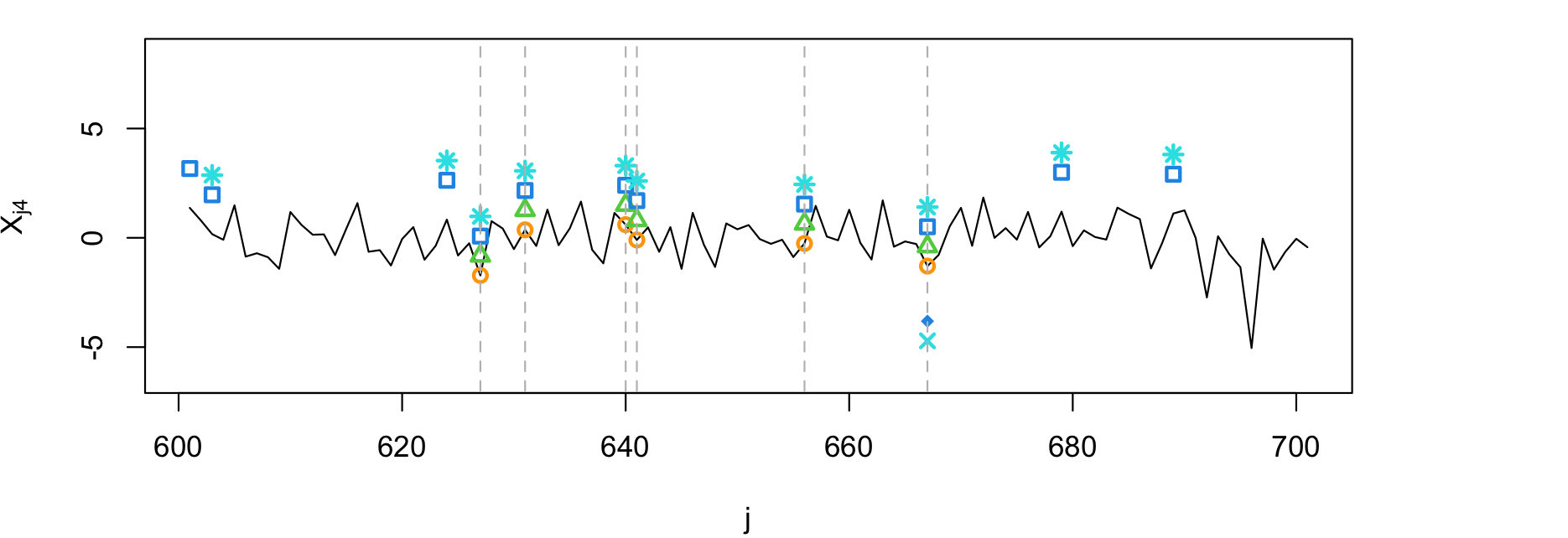}
	\end{subfigure}
	\caption{Detection in a Gaussian sample $\{\Xb_j = (X_{j1},\ldots,X_{j4})^\top, j=1,\ldots,800\},$ with mean $\boldsymbol{a}_2$ and covariance specified in (\ref{def:norm_simul_multiD}). Compared multivariate methods: BD: threshold (\ref{qBeta}) based on beta distribution, BPR: boxplot rule (\ref{boxplot_forM}), AET: asymptotically exact threshold (\ref{def:eta})--(\ref{cond:delta}), ST: Siotani's threshold (\ref{Siotani}) and univariate methods: AET1: asymptotically exact threshold (\ref{def:eta})--(\ref{cond:delta}) with $p=1$, GT: Grubbs test (\ref{Grubbs}), 3s: 3$\sigma$ rule (\ref{3sigma}) and BPR1: boxplot rule (\ref{boxplot_rule}). The true anomalies are marked by a dashed gray line.}
	\label{fig:sim_1realization_a2}
\end{figure}


\section{Real data examples}\label{sec:realdata}
\subsection{Daily step counts (univariate time series)}
In this subsection, we analyze physical activity data collected from a wearable device as part of the study described in detail by \cite{Elavsky-2025}. The analysis serves to assess the practical performance of the proposed anomaly detection procedure and to compare its findings with those of existing approaches. 
The dataset contains daily step counts from 1315 participants. For illustration, we selected three participants (IDs 999, 1429, and 2110) with long observation periods and no missing data. The recorded periods span 409 days for ID~999, 412 days for ID~1429, and 373 days for ID~2110.

Since daily step counts form a univariate time series, we applied the procedure described in Section~\ref{sec:time_series} for $p=1$ in combination with the univariate methods from Section~\ref{sec:existing_methods}. Each time series was modeled using an ARIMA model fitted via the~\texttt{R} function \texttt{auto.arima()}, which selects the best ARIMA model according to the corrected Akaike information criterion (AICc). The anomaly detection methods were subsequently applied to the resulting model residuals. As in Section \ref{Sec:simul}, we chose $\alpha=0.05$ in (\ref{Grubbs}), and $\delta=1/\sqrt{\log n}$ in (\ref{def:eta}).
The residuals were normalized by using the Huber variance estimator defined in~(\ref{def:Huber_scatter})–(\ref{def:Huber_u2}), with $\mu = 0$ and $\hat{\Sigma} = \hat{\sigma} > 0$. The computation of the estimator was based on the~\texttt{cov\_Huber()} function from the \texttt{robmed} package, with parameters set as described at the beginning of Section~\ref{Sec:simul}. Some of the considered detection methods assume normality and independence, which are also standard theoretical assumptions for ARIMA residuals (after accounting for daily and weekly seasonality). These assumptions were assessed using QQ-plots and autocorrelation function plots (see Section C.1 in the Appendix). In all presented cases, the QQ-plots showed mild tail deviations from normality, which is acceptable in a real-data setting. The autocorrelation functions revealed no significant dependence in any of the three cases.

The analyses of time series are displayed in Figures~\ref{fig:steps_id999}–\ref{fig:steps_id2110}, together with the detected anomalous days. The detection methods used were AET1: asymptotically exact threshold (\ref{def:eta})--(\ref{cond:delta}) for $p=1$, GT: Grubbs test (\ref{Grubbs}), 3s: $3\sigma$ rule (\ref{3sigma}) and BPR1: boxplot rule (\ref{boxplot_rule}). Although the true outliers are unknown, the identified observations appear reasonable. Methods AET1 and GT identified the highest peaks, methods 3$\sigma$ and especially BPR1 also flagged many smaller peaks. As was seen in Section \ref{Sec:simul}, the latter methods have performance comparable to AET1 and GT for $n \approx 400$, but they tend to produce some false positives. 

\begin{figure}[h!]
		\centering
		\includegraphics[width=\textwidth]{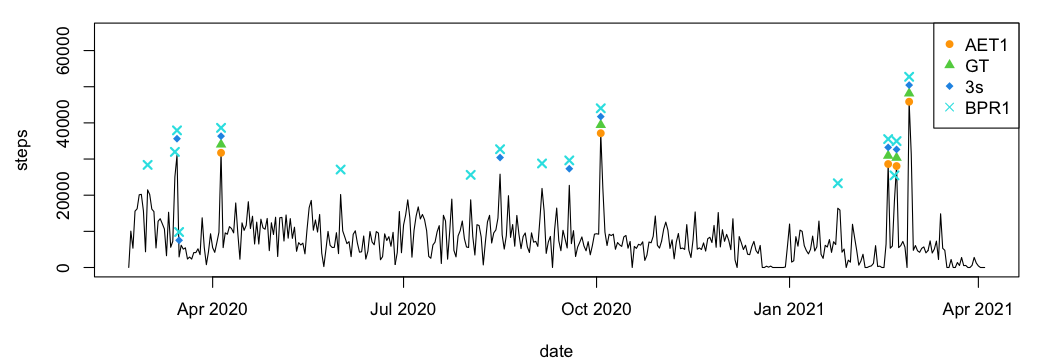}
	\caption{Daily step counts with detected anomalies for the participant with ID 999. Time period 409 days long. Time series modeled by ARIMA(1,1,2)(0,0,1)[7]. Detection methods: AET1: asymptotically exact threshold (\ref{def:eta})--(\ref{cond:delta}) for $p=1$, GT: Grubbs test (\ref{Grubbs}), 3s: $3\sigma$ rule (\ref{3sigma}) and BPR1: boxplot rule (\ref{boxplot_rule}).}
	\label{fig:steps_id999}
\end{figure}

\begin{figure}[h]
		\centering
		\includegraphics[width=\textwidth]{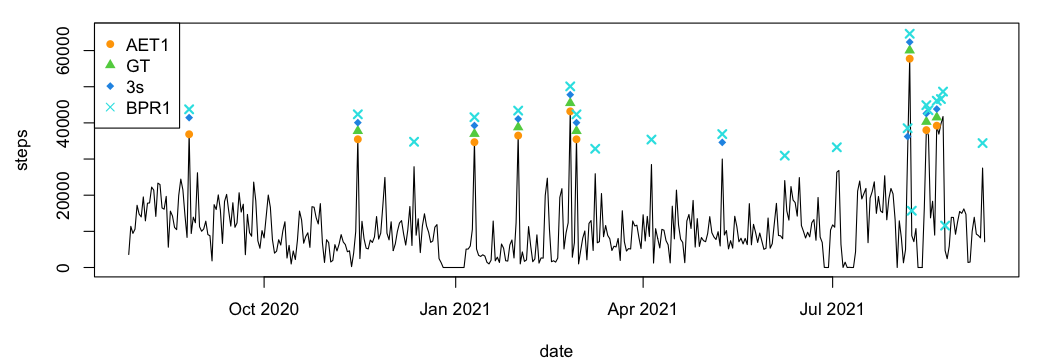}
	\caption{Daily step counts with detected anomalies for the participant with ID 1429. Time period 412 days long. Time series modeled by ARIMA(1,1,1)(0,0,1)[7]. Detection methods: AET1: asymptotically exact threshold (\ref{def:eta})--(\ref{cond:delta}) for $p=1$, GT: Grubbs test (\ref{Grubbs}), 3s: $3\sigma$ rule (\ref{3sigma}) and BPR1: boxplot rule (\ref{boxplot_rule}).}
	\label{fig:steps_id1429}
\end{figure}

\begin{figure}[h]
		\centering
		\includegraphics[width=\textwidth]{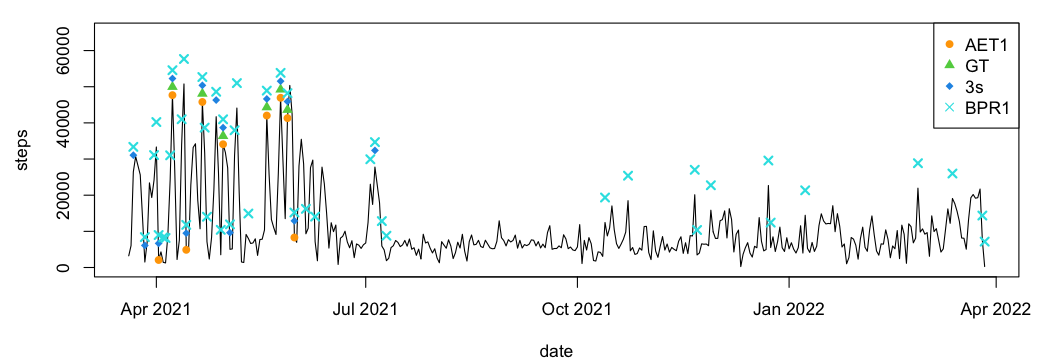}
	\caption{Daily step counts with detected anomalies for the participant with ID 2110.  Time period 373 days long. Time series modeled by ARIMA(5,1,0)(2,0,0)[7]. Detection methods: AET1: asymptotically exact threshold (\ref{def:eta})--(\ref{cond:delta}) for $p=1$, GT: Grubbs test (\ref{Grubbs}), 3s: $3\sigma$ rule (\ref{3sigma}) and BPR1: boxplot rule (\ref{boxplot_rule}).}
	\label{fig:steps_id2110}
\end{figure}

\subsection{Air pollution data (multivariate time series)}

To demonstrate the proposed anomaly detection method on a multivariate dataset, we used concentrations of five pollutants SO$_2$, NO$_2$, O$_3$, PM$_{10}$ (particles with a diameter of 10 micrometers or less), and PM$_{2.5}$ (particles with a diameter of 2.5 micrometers or less), which play a major role in air quality. The data were collected from January 1, 2015, to December 31, 2017, in the meteorological station Praha 4-Libuš. 
This period was chosen because it covers two important episodes of a smog situation (August 2015 and January/February 2017) and also a period of standard air quality (2016). The observation period covers three years and contains 11 \% of missing values. Because the VAR model used for subsequent time-series analysis does not accommodate missing values, the missing observations were first imputed using a state-space Kalman smoothing approach that explicitly accounts for a time-varying trend and weekly and annual seasonal effects. These imputed values were excluded from the subsequent anomaly detection procedure.

Since the data represent a five-dimensional time series, we model them using a vector autoregression (VAR) model fitted to the complete dataset after imputing the missing observations. The lag order was selected using the \texttt{VARselect()} function, which reports several information criteria and the final prediction error for increasing lag orders. According to the Akaike information criterion (AIC), a lag of 4 days was selected. 
Then, the anomaly detection methods described in Section~\ref{sec:existing_methods_multivariate}, together with the proposed method~(\ref{def:eta})--(\ref{cond:delta}) with $p=5$, were applied to the residuals of the fitted model (after excluding those corresponding to the imputed values), following the procedure outlined in Section~\ref{sec:time_series}. As above, we chose $\alpha=0.05$ in (\ref{qBeta}) and (\ref{Siotani}), and $\delta=1/\sqrt{\log n}$ in (\ref{def:eta}).

All five components of the time series, together with the detected anomalous days, are displayed in Figure~\ref{fig:pollutants_withAQI}(a). For clarity of presentation, we display only the results obtained using the proposed asymptotically exact threshold~(\ref{def:eta})--(\ref{cond:delta}) and Siotani's threshold~(\ref{Siotani}). The threshold~(\ref{qBeta}) based on the beta distribution and the boxplot rule~(\ref{boxplot_forM}) applied to the Mahalanobis distance identified nearly 10 \% of observations as anomalous, which would make the figure cluttered. This is consistent with the findings from the simulation study in Section \ref{Sec:simul} shown in Figure \ref{fig:Hamming_risk}, where these methods exhibited an increasing number of false positives as the sample size grew. This effect is already visible in Table \ref{tab:N_I_multi} for $n=800$. Since the present dataset contains $n=1096$ observations, the increase in false positives is expected to be strong. 
Since the true anomalies are unknown in this real-data example, we include for comparison the Air Quality Index (AQI) used by \cite{CHMI_IKO} 
and computed for the same location (Praha 4-Libuš). 
The AQI is shown as a time series in Figure~\ref{fig:pollutants_withAQI}(b), together with the official air quality levels used in meteorological classification. This comparison should be interpreted with caution, as the AQI is computed from a slightly different set of pollutants: PM$_{10}$, NO$_2$, SO$_2$, and O$_3$ during April–September, and PM$_{10}$, NO$_2$, and SO$_2$ during October–March. The exact formula is 
\begin{equation} \label{def:IKO}
	AQI = \frac{\sum_{i:c_i>C_i} \frac{c_i}{C_i}}{I} + \frac{\sum_{j:c_j \leq C_j} \frac{c_j}{C_j}}{J},
\end{equation}
where $I$ denotes the number of pollutants whose concentrations $c_i$ exceed their reference values $C_i$ and $J$ denotes the number of pollutants whose concentrations $c_j$ are less than or equal to $C_j$. If $J$ equals zero, the second term is set to zero. In the case $I=0$, the pollutant $j$ with the largest ratio $c_j/C_j$ is taken from the second term and moved to the first term. When the required set of pollutants (PM$_{10}$, NO$_2$, SO$_2$, and O$_3$ during April–September, and PM$_{10}$, NO$_2$, and SO$_2$ during October–March) is not available, the index cannot be computed.
The reference values are 200 $\mu$g/m$^3$/hour for NO$_2$, 350~$\mu$g/m$^3$/hour for SO$_2$, 120 $\mu$g/m$^3$/hour for O$_3$ and 90 $\mu$g/m$^3$/hour for PM$_{10}$. 
In Figures~\ref{fig:pollutants_withAQI}(a)--(b), smog situations are highlighted in gray. These are declared by the Czech Hydrometeorological Institute when threshold values specified by the Air Protection Act are exceeded.

As noted above, the AQI is calculated using only three or four pollutants, while PM$_{2.5}$ is not considered separately, and O$_3$ is not included throughout the entire year. However, concentrations of both these pollutants are known to contribute substantially to overall air quality, and we therefore based our anomaly detection on all five measured pollutants. This choice is also supported by the set of pollutants included in the European Air Quality Index (see \cite{EEA_AQI, Gonzalez-2025}), which is used to enable comparisons of air quality across the European Union. The index values are accessible only via an online interface and are not available for download. 

The detection methods assume that the residuals are independent and approximately normally distributed. These assumptions were assessed using QQ-plots and autocorrelation function plots of the residuals (see Section C.2 of the Appendix). The QQ-plots exhibit mild deviations from normality in the tails, while the autocorrelation function plots do not indicate any significant dependence.

We observe that the asymptotically exact threshold and Siotani’s threshold yield similar results, and both capture well the smog situation at the beginning of 2017. However, the asymptotically exact threshold also detects the smog situation in August 2015 and a peak of elevated AQI at the beginning of 2015, whereas Siotani's threshold misses these events. 
Note that anomaly detection procedures identify extreme observations in both directions, and therefore some days with exceptionally low pollution levels were also labeled as anomalous. Some short smog situations in 2015 and 2017 were not detected by either method, while several regular days were identified as anomalies. 
These discrepancies may be explained by the fact that the AQI is based on a slightly different set of pollutants, as well as by the influence of the fitted time-series model. More accurate results would likely require a larger dataset, allowing a more flexible model that captures the seasonal component more precisely. Overall, however, the anomalous days identified by the asymptotically exact threshold in Figure~\ref{fig:pollutants_withAQI}(a) appear to correspond well to periods of poor air quality.

\begin{figure}[h!]
	\centering
	\begin{subfigure}[t]{\textwidth}
		\centering
		\includegraphics[width=\textwidth]{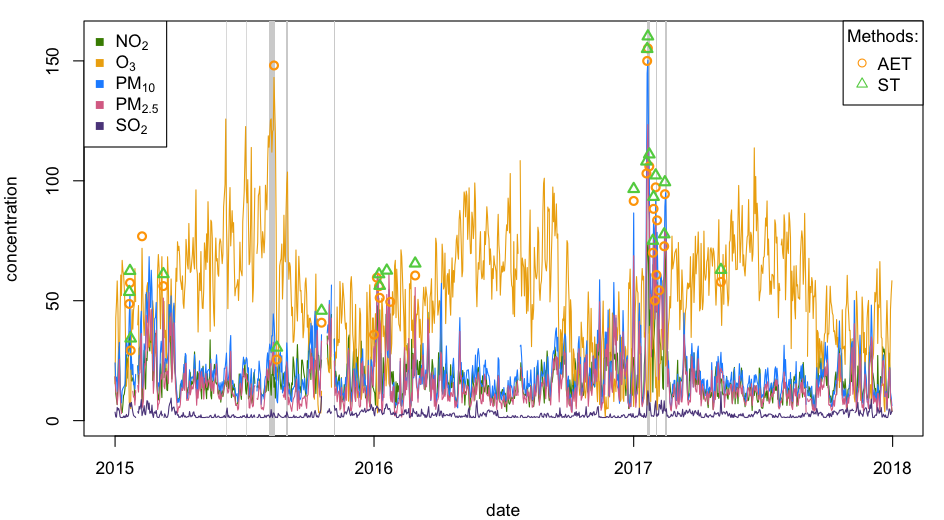}
		\caption{Anomalies detected by AET: asymptotically exact threshold (\ref{def:eta})--(\ref{cond:delta}) and ST: Siotani's threshold~(\ref{Siotani}). Smog situations indicated by gray color.}
	\end{subfigure}
	~ 
	\begin{subfigure}[t]{\textwidth}
		\centering
		\includegraphics[width=\textwidth]{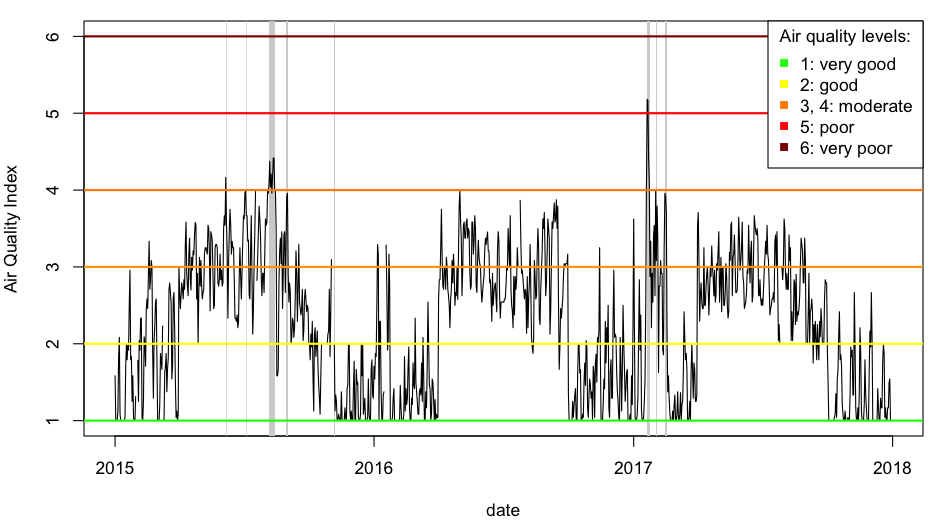}
		\caption{Air Quality Index including levels and smog situations (gray color).}
	\end{subfigure}
	\caption{Air pollutant concentrations with detected anomalous days and the official Air Quality Index in Praha 4-Libuš (Czech Republic),
		 January 1, 2015 – December 31, 2017.}
	\label{fig:pollutants_withAQI}
\end{figure}

\section{Discussion}
In this article, we proposed a method for detecting anomalies in samples of multivariate Gaussian observations that asymptotically achieves exact identification. The procedure is based on a simple thresholding rule, which makes it easy to implement and computationally efficient. For both the sample covariance matrix (computed from historical data without anomalies) and the Huber-type estimator (based on data contaminated by anomalies), we showed that the expected number of misclassified observations converges to zero, provided the anomaly signal is sufficiently strong (as specified in Theorem~\ref{theorem:upper_bound}). In addition, we established conditions under which exact identification cannot be achieved, which clarifies the limits of the detection problem. These theoretical findings were supported by a simulation study, in which the number of observations incorrectly classified by the proposed method decreased towards zero, in agreement with Theorem~\ref{theorem:upper_bound}. Moreover, the proposed method was compared with other commonly used thresholding approaches, and, for sample sizes greater than 400, it showed the best overall performance in terms of the number of misidentified observations while remaining very stable. The threshold based on the beta distribution and the boxplot rule applied to the Mahalanobis distance produced stable results, but with an increasing number of errors as the dimension increased. Siotani’s threshold exhibited a decreasing average number of errors but also showed considerable instability, occasionally producing a larger number of misclassified observations in certain scenarios. The proposed method was also used to analyze two real-data examples: a univariate time series of daily step counts and a multivariate time series of air pollutant concentrations.

Recall that if $p=1$ and the variance is known, then Theorem~\ref{theorem:upper_bound} and Theorem~\ref{theorem:lower_bound} reduces to results of \cite[Section 5]{Butucea-2018}. These results were further extended to the functional case in Theorems 1 and 2 of \cite{Stepanova-2025_EJS} and for an even more general case in Theorems 2 and 3 of \cite{Stepanova-2025_SISP}. The case of almost full selection, when the number of misidentifications does not tend to zero but remains small compared to the total number of active components, is addressed in~\cite{Stepanova-2025_ArXive}.

Detecting anomalies in a multivariate setting is a challenging problem, since a vector observation may appear extreme either due to a large deviation in a single component or through moderate deviations across many components. Our approach is very effective in identifying both types of anomalies: those occurring in individual components and those arising from their composition. The proposed statistic is based on the $\ell_1$-norm, which aggregates contributions from all coordinates and transforms numerous small shifts into a substantial overall signal. 
By standardizing the observations using the estimated covariance, the dependence structure of the data is removed, which makes the statistic insensitive to the choice of coordinate system and allows for detecting deviations in any direction. In contrast, marginal statistics respond only to shifts aligned with individual coordinates. The method can be applied either online, with the threshold and covariance estimate updated at each step, or in an offline setting using previously recorded data.  
Another practical advantage is that new variables can be incorporated into the procedure with minimal effort.

A potential limitation of the proposed method is its reduced sensitivity to anomalies that affect only a small subset of components. Because the statistic based on $\ell_1$-norm spreads the effect of each coordinate across the entire sum, it is naturally less influenced by a single outlying value unless that value is exceptionally large. While this property contributes to a more stable and robust overall measure of deviation, it may also cause certain sparse anomalies to be overlooked. Conditions (\ref{cond:a_UB}) and (\ref{cond:a_LB}) describe the minimal strength of the anomaly that is required for detection in terms of $\lVert L^{-1}\ab \rVert_1$ and so define the demarcation line between what is possible and impossible in this detection problem. A practical way to address this issue is to pair the $\ell_1$-based detector with its one-dimensional version applied, for example, to $\lVert \hat{L}^{-1}\Xb \rVert_{\infty}$, thereby increasing sensitivity to large deviations occurring in individual coordinates.

Our results offer an efficient method for identifying outlying observations not only in independent data but also in time series, when applied to model residuals. We illustrate this approach on real datasets from biomechanics (measurements from a wearable device) and meteorology. In the biomechanics example, we analyze the daily number of steps as a univariate time series and show that our method successfully highlights days with unusually high or low activity. 
When applied to a univariate time series, the proposed method can serve as an attractive alternative to the procedure of \cite{Shewhart-1931} or the EWMA~procedure of \cite{Roberts-1959} and can support, for instance, just-in-time adaptive interventions (JITAI) by providing timely warnings of atypical behavior, similarly to their use in~\cite{Schat-2026}. Moreover, in the multivariate setting, it offers a natural extension of classical statistical process control methods.

In the second real-data example, we analyzed a~multivariate time series of measurements of five air pollutants. The anomalies detected by our procedure matched the days classified as highly polluted according to the Air Quality Index commonly used in meteorological practice, and also identified the main smog situations.

Model (\ref{model}) is a standard additive noise model that can be applied to both individual observations and residuals from a wide range of statistical models, making our method broadly applicable. The assumption of a Gaussian distribution is common in statistical practice and is routinely imposed on residuals in many modeling settings. Other models that may be worth exploring include regression frameworks with non-Gaussian, heavy-tailed noise, as well as models designed for noncontinuous outcomes such as count or binary data.

The performance of each method depends strongly on the quality of the covariance estimator, since inaccurate covariance estimates can lead to a higher number of misidentified anomalies. In practice, however, estimating covariance in the presence of outliers is challenging. The optimality properties of the proposed detector are guaranteed for estimators satisfying decomposition~(\ref{cond_S}) together with assumptions (A1)--(A3). We showed that both the sample covariance (computed from uncontaminated historical data) and the Huber-type robust estimator fulfill these requirements. 

The main limitation of the proposed method is the assumption of independence and normality of the underlying observations. Although these assumptions are common in many standard statistical procedures, they may not always hold in practical applications, particularly in settings with temporal dependence or heavy-tailed noise. Future work could therefore focus on extending the methodology to more flexible distributional frameworks, such as heavy-tailed models based on the Student $t$-distribution with low degrees of freedom. 

In summary, the proposed method offers a flexible, computationally efficient framework for anomaly detection in multivariate settings while avoiding strong subjective choices, such as the specification of probability levels required by methods based on empirical or distributional quantiles. The method is supported by theoretical results and validated through extensive simulations and real-world applications. Its ability to capture both component-wise and composite deviations makes it broadly applicable across domains, while its simplicity allows for easy integration into existing procedures and routine screening of large datasets.



\section*{Acknowledgements}
We would like to thank the project “Healthy Aging in Industrial Environment HAIE \\ CZ.02.1.01/0.0/0.0/16\_019/0000798”, co-financed by the European Union, for the use of the data.
The air pollution data were provided by the Czech hydrometeorological institute (CHMI). The authors also gratefully acknowledge Ing. Hana Škáchová from CHMI for valuable consultations on the meteorological aspects of the study and for reviewing the corresponding section of the manuscript.
\section*{Funding}
The study is from the project ``Research of Excellence on Digital Technologies and Wellbeing \\ CZ.02.01.01/00/22\_008/0004583", which is co-financed by the European Union.

\bibliographystyle{plain} 
\bibliography{references}       

@book{Abra-1972,
  author    = {Abramowitz, Milton and Stegun, Irene A.},
  title     = {Handbook of Mathematical Tables with Formulas, Graphs, and Mathematical Tables},
  edition   = {10th},
  series    = {Applied Mathematics Series},
  number    = {55},
  publisher = {National Bureau of Standards},
  address   = {Washington, DC, USA},
  year      = {1972}
}

@book{Anderson-2003,
  author    = {Anderson, Theodore Wilbur},
  title     = {An Introduction to Multivariate Statistical Analysis},
  edition   = {3rd},
  publisher = {John Wiley \& Sons},
  address   = {New Jersey},
  year      = {2003}
}

@book{Barnett-1998,
  author    = {Barnett, Vic and Lewis, Toby},
  title     = {Outliers in Statistical Data},
  edition   = {3rd},
  publisher = {John Wiley \& Sons},
  address   = {Chichester},
  year      = {1998}
}

@article{Basu-2007,
  author  = {Basu, Sabyasachi and Meckesheimer, Martin},
  title   = {Automatic outlier detection for time series:
an application to sensor data},
  journal = {Knowledge and Information Systems},
  volume  = {11},
  number  = {2},
  pages   = {137--154},
  year    = {2007}
}

@mastersthesis{Cui-2014,
  author       = {Cui, Xin},
  title        = {Optimal Component Selection in High Dimension},
  school       = {Carleton University},
  address      = {Ottawa, Canada},
  year         = {2014}
}

@book{Massart-2013,
  author    = {Boucheron, St{\'e}phane and Lugosi, G{\'a}bor and Massart, Pascal},
  title     = {Concentration Inequalities: A Nonasymptotic Theory of Independence},
  edition   = {1st},
  publisher = {Oxford University Press},
  address   = {Oxford},
  year      = {2013}
}

@article{Butucea-2018,
  author  = {Butucea, Cristina and Ndaoud, Mohamed and Stepanova, Natalia and Tsybakov, Alexandre},
  title   = {Variable Selection with {H}amming Loss},
  journal = {The Annals of Statistics},
  volume  = {46},
  number  = {5},
  pages   = {1837--1875},
  year    = {2018}
}

@article{Chandola-2009,
  author  = {Chandola, Varun and Banerjee, Arindam and Kumar, Vipin},
  title   = {Anomaly Detection: A Survey},
  journal = {ACM Computing Surveys},
  volume  = {41},
  number  = {3},
  pages   = {15},
  year    = {2009}
}

@misc{CHMI_IKO,
  author       = {{Czech Hydrometeorological Institute}},
  title        = {Index kvality ovzduší},
  year         = {2026},
  url          = {https://www.chmi.cz/kvalita-ovzdusi/imise-informacni-system-hodnoceni-kvality-ovzdusi/podklady-pro-hodnoceni-ko/hodnoceni-ko-index-kvality-ovzdusi},
  note         = {Accessed: 2026-04-20}
}

@article{Elavsky-2025,
	author = {Elavsky, Steriani and Brabec, Marek and Maly, Marek and Knapova, Lenka and Kastovska, Barbora and Sebera, Michal and Ely, Marcela and Jandackova, Vera K and Keller, Jan and Pavel, Misha},
	title = {The temporal dynamics of the association between daily physical activity and life satisfaction},
	journal = {Annals of Behavioral Medicine},
	volume = {59},
	number = {1},
	pages = {1--14},
	year = {2025}
}

@misc{EEA_AQI,
  author       = {{European Environment Agency}},
  title        = {European Air Quality Index},
  year         = {2026},
  url         = {https://airindex.eea.europa.eu/AQI/index.html},
  note = {Accessed: 2026-04-20}
}

@article{Gnanadesikan-1972,
  author  = {Gnanadesikan, Ramanathan and Kettenring, John R.},
  title   = {Robust Estimates, Residuals, and Outlier Detection with Multiresponse Data},
  journal = {Biometrics},
  volume  = {28},
  number  = {1},
  pages   = {81--124},
  year    = {1972}
}

@techreport{Gonzalez-2025,
  author       = {Gonz{\'a}lez Ortiz, Alberto and Soares, Joana and Targa, Jaume and Colina, Maria and Azc{\'o}n, Laia and Banyuls, Lorena and Tobollik, Myriam and Monge, Silvia and Finnbjornsdottir, Ragnhildur},
  title        = {EEA's Revision of the European Air Quality Index Bands},
  institution  = {European Topic Centre on Human Health and the Environment},
  type         = {ETC HE Report},
  number       = {2024/17},
  year         = {2025}
}

@article{Grubbs-1969,
  author  = {Grubbs, Frank E.},
  title   = {Procedures for Detecting Outlying Observations in Samples},
  journal = {Technometrics},
  volume  = {11},
  number  = {1},
  pages   = {1--21},
  year    = {1969}
}

@book{Hawkins-1980,
  author    = {Hawkins, Douglas M.},
  title     = {Identification of Outliers},
  publisher = {Chapman \& Hall},
  year      = {1980}
}

@book{Higham-2008,
  author    = {Higham, Nicholas J.},
  title     = {Functions of Matrices: Theory and Computation},
  publisher = {Society for Industrial and Applied Mathematics},
  address   = {Philadelphia},
  year      = {2008}
}

@article{Hill-2010,
  author  = {Hill, David J. and Minsker, Barbara S.},
  title   = {Anomaly detection in streaming environmental sensor data: A~data-driven modeling approach},
  journal = {Environmental Modelling \& Software},
  volume  = {25},
  number = {9},
  pages   = {1014--1022},
  year    = {2010}
}

@book{Horn-2013,
  author    = {Horn, Roger A. and Johnson, Charles R.},
  title     = {Matrix Analysis},
  edition   = {2nd},
  publisher = {Cambridge University Press},
  address   = {United States of America},
  year      = {2013}
}

@article{Laurent-2018,
  author  = {Laurent, Béatrice and Marteau, Clément and Maugis-Rabusseau, Cathy},
  title   = {Multidimensional two-component Gaussian mixtures detection},
  journal = {Annales de l’Institut Henri Poincaré - Probabilités et Statistiques},
  volume  = {54},
  number = {2},
  pages   = {842–-865},
  year    = {2018}
}

@inproceedings{Laurikkala-2000,
  author    = {Laurikkala, Jorma and Juhola, Martti and Kentala, Erna},
  title     = {Informal Identification of Outliers in Medical Data},
  booktitle = {Proceedings of the Fifth International Workshop on Intelligent Data Analysis in Medicine and Pharmacology},
  year      = {2000}
}

@article{Laxhammar-2014,
  author  = {Laxhammar, Rikard and Falkman, G\"{o}ran},
  title   = {Online Learning and Sequential Anomaly Detection in Trajectories},
  journal = {IEEE Transactions on Pattern Analysis and Machine Intelligence},
  volume  = {36},
  number  = {6},
  pages   = {1158--1173},
  year    = {2014}
}

@book{Lehmann-2005,
  author    = {Lehmann, Erich L. and Romano, Joseph P.},
  title     = {Testing Statistical Hypotheses},
  edition = {3rd},
  publisher = {Springer},
  address   = {New York},
  year      = {2005}
}

@article{Liu-1991,
  author  = {Liu, Jen-Pei and Weng, Chung-Sing},
  title   = {Detection of outlying data in bioavailability/bioequivalence studies},
  journal = {Statistics in medicine},
  volume  = {10},
  pages   = {1375--1389},
  year    = {1991}
}

@article{Lu-2007,
  author  = {Lu, Chang-Tien and Kou, Yufeng and Zhao, Jiang and Chen, Li},
  title   = {Detecting and tracking regional outliers in meteorological data},
  journal = {Information Sciences},
  volume  = {177},
  pages   = {1609--1632},
  year    = {2007}
}

@article{Maronna-1976,
  author  = {Maronna, Ricardo A.},
  title   = {Robust {M}-Estimators of Multivariate Location and Scatter},
  journal = {The Annals of Statistics},
  volume  = {4},
  number  = {1},
  pages   = {51--67},
  year    = {1976}
}

@book{Petrov-1995,
  author    = {Petrov, Valentin V.},
  title     = {Limit Theorems of Probability Theory: Sequences of Independent Random Variables},
  publisher = {Clarendon Press},
  address   = {Oxford},
  year      = {1995}
}

@article{Roberts-1959,
  author  = {Roberts, S.W.},
  title   = {Control chart tests based on geometric moving averages},
  journal = {Technometrics},
  volume  = {1},
  number  = {3},
  pages   = {239--250},
  year    = {1959}
}

@book{Rousseeuw-1987,
  author    = {Rousseeuw, Peter J. and Leroy, Annick M.},
  title     = {Robust Regression and Outlier Detection},
  publisher = {John Wiley \& Sons},
  address   = {New York},
  year      = {1987}
}

@article{Schat-2026,
  author  = {Schat, Evelien and Schreuder, Marieke J. and Ceulemans, Eva},
  title   = {Statistical process control for real-time monitoring in clinical psychology: State of the art and future research agenda},
  journal = {Neuroscience Applied},
  volume  = {5},
  eid = {106991},
  year = {2026},
}

@article{Schmidl-2022,
  author  = {Schmidl, Sebastian and Wenig, Phillip and Papenbrock, Thorsten},
  title   = {Anomaly Detection in Time Series: A Comprehensive Evaluation},
  journal = {Proceedings of the VLDB Endowment},
  volume  = {15},
  number  = {9},
  pages   = {1779--1797},
  year    = {2022}
}

@book{Shewhart-1931,
  author    = {Shewhart, Walter A.},
  title     = {Economic Control of Quality of Manufactured Product},
  publisher = {D. Van Nostrand Company, Inc.},
  address   = {New York},
  year      = {1931}
}

@article{Siffer-2017,
  author  = {Siffer, Alban and Fouque, Pierre-Alain and Termier, Alexandre and Largouet, Christine},
  title   = {Anomaly Detection in Streams with Extreme Value Theory},
  journal = {Proceedings of the 23rd ACM SIGKDD International Conference on Knowledge Discovery and Data Mining},
  volume  = {},
  pages   = {1067--1075},
  year    = {2017}
}

@article{Siotani-1959,
  author  = {Siotani, Minoru},
  title   = {The Extreme Value of the Generalized Distances of the Individual Points in the Multivariate Normal Sample},
  journal = {Annals of the Institute of Statistical Mathematics},
  volume  = {10},
  pages   = {183--208},
  year    = {1959}
}

@article{Stefansky-1971,
  author  = {Stefansky, Wilhelmine},
  title   = {Rejecting Outliers by Maximum Normed Residual},
  journal = {The Annals of Mathematical Statistics},
  volume  = {42},
  number  = {1},
  pages   = {35--45},
  year    = {1971}
}

@article{Stefansky-1972,
  author  = {Stefansky, Wilhelmine},
  title   = {Rejecting Outliers in Factorial Designs},
  journal = {Technometrics},
  volume  = {14},
  number  = {2},
  pages   = {469--479},
  year    = {1972}
}

@article{Stepanova-2025_EJS,
  author  = {Stepanova, Natalia and Turcicova, Marie},
  title   = {Exact Variable Selection in Sparse Nonparametric Models},
  journal = {Electronic Journal of Statistics},
  volume  = {19},
  pages   = {2001--2032},
  year    = {2025}
}

@article{Stepanova-2025_SISP,
  author  = {Stepanova, Natalia and Turcicova, Marie},
  title   = {Adaptive Exact Recovery in Sparse Nonparametric Models},
  journal = {Statistical Inference for Stochastic Processes},
  volume  = {28},
  eid = {15},
  pages   = {1--26},
  year    = {2025}
}

@misc{Stepanova-2025_ArXive,
  author       = {Stepanova, Natalia and Turcicova, Marie and Zhao, Xiang},
  title        = {Adaptive Almost Full Recovery in Sparse Nonparametric Models},
  howpublished = {Preprint on arXiv:2512.10488},
  year         = {2025}
}

@book{Stewart-1990,
  author    = {Stewart, Gilbert W. and Sun, Ji-guang},
  title     = {Matrix Perturbation Theory},
  publisher = {Academic Press},
  edition = {1st},
  address   = {San Diego},
  year      = {1990}
}

@article{Sun-1991,
  author  = {Sun, Ji-guang},
  title   = {Perturbation Bounds for the Cholesky and {QR} Factorizations},
  journal = {BIT Numerical Mathematics},
  volume = {31},
  pages   = {341--352},
  year    = {1991}
}

@article{Tsay-2000,
  author  = {Tsay, Ruey S. and Pe\~{n}a, Daniel and Pankratz, Alan E.},
  title   = {Outliers in multivariate time series},
  journal = {Biometrika},
  volume  = {87},
  number  = {4},
  pages   = {789--804},
  year    = {2000}
}

@article{Tyler-1983,
  author  = {Tyler, David E.},
  title   = {Robustness and Efficiency Properties of Scatter Matrices},
  journal = {Biometrika},
  volume  = {70},
  number  = {2},
  pages   = {411--420},
  year    = {1983}
}

@book{Vaart-1998,
  author    = {van der Vaart, Adrianus Willem},
  title     = {Asymptotic Statistics},
  publisher = {Cambridge University Press},
  edition = {1st},
  year      = {1998}
}

@book{Vershynin-2018,
  author    = {Vershynin, Roman},
  title     = {High-Dimensional Probability: An Introduction with Applications in Data Science},
  publisher = {Cambridge University Press},
  year      = {2018}
}

@article{Verzelen-2017,
  author  = {Verzelen, Nicolas and Arias-Castro, Ery},
  title   = {Detection and Feature Selection in Sparse Mixture Models},
  journal = {The Annals of Statistics},
  volume  = {45},
  number = {5},
  pages   = {1920--1950},
  year    = {2017}
}

@INPROCEEDINGS{Zhou-2016,
  author={Zhou, Zeng-Guang and Tang, Ping},
  booktitle={2016 IEEE International Geoscience and Remote Sensing Symposium (IGARSS)}, 
  title={Improving time series anomaly detection based on exponentially weighted moving average (EWMA) of season-trend model residuals}, 
  year={2016},
  volume={},
  number={},
  pages={3414--3417},
}

@article{Zu-2010,
  author  = {Zu, Jiyun and Yuan, Ke-Hai},
  title   = {Local Influence and Robust Procedures for Mediation Analysis},
  journal = {Multivariate Behavioral Research},
  volume  = {45},
  number  = {1},
  pages   = {1--44},
  year    = {2010}
}

\appendix

\section{Decomposition (12) of a covariance estimator}
Consider a sample $\Xb_1,\ldots,\Xb_n$ from $N_p(\boldsymbol{\mu},\Sigma)$, where $\Sigma=(\sigma_{jk})_{j,k=1}^p$ is unknown.	In Section 2.3, we assume the covariance estimator $S=(s_{jk})_{j,k=1}^p$ to allow for the decomposition \begin{equation} \label{cond_S}
	s_{jk} = \frac{1}{n} \sum_{i=1}^n g_{jk}(\Xb_i) + b_{jk}(n) + R_{jk},
\end{equation}
where 
\begin{itemize}
	\item[(A1)]  $g_{jk}:\mathbb{R}^p \to \mathbb{R}$ is a measurable function such that $\E g_{jk}(\Xb_1) = \sigma_{jk}$ and $\E |g_{jk}(\Xb_1) - \sigma_{jk}|^8 < \infty$,
	\item[(A2)] $b_{jk}(n) = O(n^{-1})$ is a deterministic bias, and
	\item[(A3)] $R_{jk}$ is a small centered random remainder, in particular, we require $\E R_{jk}=0$ and $\E |R_{jk}|^8 = O(n^{-4})$. 
\end{itemize}

\subsection{Consistency of estimator (\ref{cond_S})}
Here, we show that such a covariance estimator is consistent. Indeed, by applying Markov's inequality and the inequality relating the spectral norm $\lVert \cdot \rVert_2$ and the Frobenius norm $\lVert \cdot \rVert_F$, we obtain, for all $\epsilon > 0$, 
\begin{align} \label{eq:Sconsistent_v1}
	\begin{split}
	\Prob{}{\lVert S-\Sigma \rVert_2 > \epsilon} \leq  \Prob{}{\lVert S-\Sigma \rVert_F > \epsilon} &=  \Prob{}{\lVert S-\Sigma \rVert_F^2 > \epsilon^2} \\
	& \leq \frac{\E \lVert S-\Sigma \rVert_F^2}{\epsilon^2} = \frac{\sum_{j,k=1}^p \E ( s_{jk}-\sigma_{jk} )^2}{\epsilon^2}.
		\end{split}
\end{align}
From assumptions (A1) and (A2), we have
\begin{align} \label{eq:Sconsistent_v2}
	\E ( s_{jk}-\sigma_{jk} )^2 &= \var (s_{jk}-\sigma_{jk}) = \var \left( \frac{1}{n} \sum_{i=1}^n (g_{jk}(\Xb_i) - \sigma_{jk}) + b_{jk}(n) + R_{jk} \right) \nonumber \\
	&= \var \left( \frac{1}{n} \sum_{i=1}^n (g_{jk}(\Xb_i) - \sigma_{jk} )+ R_{jk} \right) \nonumber \\
	& \leq 2 \var \left( \frac{1}{n} \sum_{i=1}^n (g_{jk}(\Xb_i) - \sigma_{jk} )\right) + 2 \var \left( R_{jk} \right).
\end{align}
Because our observations are iid, we can further write
\begin{align} \label{eq:Sconsistent_v3}
	\var \left( \frac{1}{n} \sum_{i=1}^n (g_{jk}(\Xb_i) - \sigma_{jk} )\right) & = \frac{1}{n^2} \sum_{i=1}^n \var \left( g_{jk}(\Xb_i) - \sigma_{jk} \right) \nonumber \\ 
	&= \frac{1}{n} \E \left( g_{jk}(\Xb_i) - \sigma_{jk} \right)^2 = \mathcal{O}(n^{-1}),
\end{align}
since by Lyapunov inequality and assumption (A1), we have $$\E \left( g_{jk}(\Xb_i) - \sigma_{jk} \right)^2 \leq \left( \E \left( g_{jk}(\Xb_i) - \sigma_{jk} \right)^8 \right)^{1/4} < \infty.$$
By using the same arguments and assumption (A3), we get
\begin{align} \label{eq:Sconsistent_v4}
	\var \left( R_{jk} \right) = \E R_{jk}^2 \leq \left( \E |R_{jk}|^8 \right)^{1/4} = \left( \mathcal{O}(n^{-4})\right)^{1/4} = \mathcal{O}(n^{-1}).
\end{align}
By substituting (\ref{eq:Sconsistent_v3}) and (\ref{eq:Sconsistent_v4}) into (\ref{eq:Sconsistent_v2}) and then into (\ref{eq:Sconsistent_v1}), we obtain for all $\epsilon >0$
\begin{align} \label{eq:Sconsistent}
	\Prob{}{\lVert S-\Sigma \rVert_2 > \epsilon} \leq  \frac{p^2}{\epsilon^2} \mathcal{O} \left( \frac{1}{n} \right) = \mathcal{O}(n^{-1}) = \smallO(1), 
\end{align}
and therefore, $S$ with entries given by (\ref{cond_S}) is a consistent estimator of $\Sigma$ in the spectral norm.

\subsection{Huber-type estimator allows for decomposition (\ref{cond_S})}
Recall that, for a sample $\Xb_1,\ldots,\Xb_n$ with known mean $\boldsymbol{\mu}=(\mu_1,\ldots,\mu_p)^\top$, the M-estimator $\hat{\Sigma}$ with Huber-type weights as in \cite{Maronna-1976, Zu-2010} is defined as a solution to the equation
\begin{align} \label{def:Huber_scatter}
	\hat{\Sigma} &= \frac{1}{n} \sum_{i=1}^n u_2 (M_{\hat{\Sigma}}(\Xb_i)) (\Xb_i - \boldsymbol{\mu}) (\Xb_i - \boldsymbol{\mu})^\top,
\end{align}
where $M_{\hat{\Sigma}}(\Xb_i) = (\Xb_i - \boldsymbol{\mu}) \hat{\Sigma}^{-1}(\Xb_i - \boldsymbol{\mu})^\top$ is the sample Mahalanobis distance of $\Xb_i$ from $\boldsymbol{\mu}$, and $u_2$ is the Huber-type weight function of the following form
\begin{align}
	\begin{split}
		u_2(d) &= \frac{1}{c}, \quad d \leq r \\
		&= \frac{r^2}{cd}, \quad d>r, 
	\end{split} \label{def:Huber_u2}
\end{align}
for some $r>0$ given by $\operatorname{P}\left(\chi_p^2 > r^2\right) = \kappa$ such that $\kappa$ is the proportion of cases one wants to downweight. The quantity $c$ is a normalizing constant chosen so that the robust estimator $\hat{\Sigma}$ is unbiased, i.e.  $\E \left[ u_2(M_{\Sigma}(\Xb))(\Xb-\boldsymbol{\mu})(\Xb - \boldsymbol{\mu})^\top \right] = \Sigma$, under the assumption that $\Xb \sim N_p(\boldsymbol{\mu},\Sigma)$. At convergence, the $j,k$-th entry of the robust covariance estimator can be written in the form (\ref{cond_S}) with
\begin{align}
	g_{jk}(\Xb_i) &= u_2(M_{\Sigma}(\Xb_i))(X_{ij} - \mu_{j})(X_{ik}-\mu_k) \label{Huber_g}\\
	\begin{split}
		R_{jk} &= \frac{1}{n} \sum_{i=1}^n \left( u_2(M_{\hat{\Sigma}}(\Xb_i)) - u_2(M_{\Sigma}(\Xb_i)) \right) (X_{ij} - \mu_{j})(X_{ik}-\mu_k)- \\ 
		& \quad \quad - \frac{1}{n} \E \left[ \sum_{i=1}^n \left( u_2(M_{\hat{\Sigma}}(\Xb_i)) - u_2(M_{\Sigma}(\Xb_i)) \right) (X_{ij} - \mu_{j})(X_{ik}-\mu_k) \right] 
	\end{split} \label{Huber_R} \\
	b_{jk}(n) &= \frac{1}{n} \E \left[ \sum_{i=1}^n \left( u_2(M_{\hat{\Sigma}}(\Xb_i)) - u_2(M_{\Sigma}(\Xb_i)) \right) (X_{ij} - \mu_{j})(X_{ik}-\mu_k) \right], \label{Huber_b}
\end{align}
where $M_{\Sigma}(\Xb_i)$ is the Mahalanobis distance of $\Xb_i$ from $\boldsymbol{\mu}$. 
In this section, we will verify that the covariance estimator (\ref{def:Huber_scatter}) decomposed as in (\ref{Huber_g})--(\ref{Huber_b}) satisfies the assumptions (A1)--(A3) above.

Assumption (A1) holds due to the choice of the calibration constant $c$ (see (\ref{def:Huber_u2})), which ensured unbiasedness under a multivariate normal distribution, and also since all moments of the normal distribution are finite. The first part of assumption (A3), $\E R_{jk}=0$, yields automatically from the definition of $R_{jk}$. Therefore, it remains to show that $b_{jk}(n) = O(n^{-1})$ and $\E |R_{jk}|^8=O(n^{-4})$. 

Denote
\begin{equation}
	B_{jk} := \frac{1}{n} \sum_{i=1}^n \left( u_2(M_{\hat{\Sigma}}(\Xb_i)) - u_2(M_{\Sigma}(\Xb_i)) \right) (X_{ij} - \mu_{j})(X_{ik}-\mu_k),
\end{equation}
then from (\ref{Huber_R}) and (\ref{Huber_b}), we see that $b_{jk}(n) = \E B_{jk}$ and $R_{jk} = B_{jk}-\E B_{jk}$. By using the relation 
\begin{equation} \label{eq:a_plus_b_na_8}
	(a+b)^8 \leq 2^7(a^8 + b^8)
\end{equation}
for $a=|B|$ and $b=\E |B|$, we obtain
\begin{align}
	\E |B-\E B|^8 \leq \E (|B| + \E |B|)^8 \leq 2^7 \E \left( |B|^8 + (\E |B|)^8 \right) \leq 2^8 \E |B|^8,
\end{align}
where $(\E |B|)^8 \leq \E |B|^8$ comes from the Jensen inequality. 
Hence, we have
\begin{align}
	\E |R_{jk}|^8 
	& \leq 2^8 \E |B_{jk}|^8 \nonumber \\
	&= 2^8  \E \left[ \frac{1}{n} \sum_{i=1}^n \left( u_2(M_{\hat{\Sigma}}(\Xb_i)) - u_2(M_{\Sigma}(\Xb_i)) \right) (X_{ij} - \mu_{j})(X_{ik}-\mu_k) \right]^8.
\end{align}
With this notation, it remains to prove that $\E |B_{jk}|^8=O(n^{-4})$ and $b_{jk}= \E B_{jk}=O(n^{-1})$. First, let's have a look at the later relation.

Denote $\boldsymbol{\theta}_0 := \operatorname{vec} \Sigma$ and $\boldsymbol{\hat{\theta}} := \operatorname{vec} \hat{\Sigma}$ and note that M-estimators can be alternatively computed as a solution of an equation $0=\sum_{i=1}^n \psi(\Xb_i, \boldsymbol{\hat{\theta}})$ for some criterion function $\psi$. By computing the Taylor expansion of $\psi$ around $\boldsymbol{\theta}_0$, we obtain an equation
\begin{equation} \label{eq:Taylor_of_psi}
	0= \frac{1}{n} \sum_{i=1}^n \psi(\Xb_i, \boldsymbol{\theta_0}) + \frac{1}{n} \sum_{i=1}^n A_{i} (\boldsymbol{\hat{\theta}} - \boldsymbol{\theta_0}) + \frac{1}{2} (\boldsymbol{\hat{\theta}} - \boldsymbol{\theta_0})^\top \left[\frac{1}{n} \sum_{i=1}^n B_i \right] (\boldsymbol{\hat{\theta}} - \boldsymbol{\theta_0})+ \text{higher orders},
\end{equation}
where $A_i = \frac{\partial}{\partial \boldsymbol{\theta} }\psi(\Xb_i, \boldsymbol{\theta_0})$ is assumed to exist and be nonsingular and $B_i$ collects the second derivatives of $\psi$. Univariate version of this expansion can be found in \cite[rel. (5.18)]{Vaart-1998}. By the Law of large numbers, for large enough $n$, we have $\frac{1}{n} \sum_{i=1}^n \frac{\partial}{\partial \boldsymbol{\theta} }\psi(\Xb_i, \boldsymbol{\theta_0}) \approx \E_{\boldsymbol{\theta}_0} \left[ \frac{\partial}{\partial \boldsymbol{\theta} }\psi(\Xb, \boldsymbol{\theta_0}) \right] =:A$ and analogously $\frac{1}{n} \sum_{i=1}^n B_i \approx \E_{\boldsymbol{\theta}_0} B_1 =: B$. From this, applying expectation to (\ref{eq:Taylor_of_psi}) yields
\begin{equation*}
	0 \approx \E \psi(\Xb_1, \boldsymbol{\theta_0}) + A \E (\boldsymbol{\hat{\theta}} - \boldsymbol{\theta_0}) + \frac{1}{2} \E \left[ (\boldsymbol{\hat{\theta}} - \boldsymbol{\theta_0})^\top B (\boldsymbol{\hat{\theta}}  - \boldsymbol{\theta_0})\right]+ \E (\text{higher orders}).
\end{equation*}
Due to the choice of the calibration constant $c$ in the Huber estimator, we have $\E \psi(\Xb_1, \boldsymbol{\theta_0}) = 0$ under the assumption of multivariate normal distribution, and hence
\begin{align}
	\E (-\boldsymbol{\hat{\theta}}  + \boldsymbol{\theta_0}) &\approx \frac{1}{2} A^{-1} \E \left[ (\boldsymbol{\hat{\theta}} - \boldsymbol{\theta_0})^\top B (\boldsymbol{\hat{\theta}}  - \boldsymbol{\theta_0})\right] = \frac{1}{2} A^{-1} \E \left[ \operatorname{tr} \left( B  (\boldsymbol{\hat{\theta}}  - \boldsymbol{\theta_0})(\boldsymbol{\hat{\theta}} - \boldsymbol{\theta_0})^\top \right) \right] \nonumber \\
	& = \frac{1}{2} A^{-1} \operatorname{tr} \left[ B \E \left( (\boldsymbol{\hat{\theta}}  - \boldsymbol{\theta_0})(\boldsymbol{\hat{\theta}} - \boldsymbol{\theta_0})^\top \right) \right] = O(1/n),
\end{align}
since $\sqrt{n} (\boldsymbol{\hat{\theta}} - \boldsymbol{\theta_0})$ have normal distribution for large $n$ and $\E \left( (\boldsymbol{\hat{\theta}}  - \boldsymbol{\theta_0})(\boldsymbol{\hat{\theta}} - \boldsymbol{\theta_0})^\top \right) = \var (\boldsymbol{\hat{\theta}}  - \boldsymbol{\theta_0}) + \E (\boldsymbol{\hat{\theta}}  - \boldsymbol{\theta_0}) \E (\boldsymbol{\hat{\theta}}  - \boldsymbol{\theta_0})^\top$, where $\E (\boldsymbol{\hat{\theta}}  - \boldsymbol{\theta_0})= \boldsymbol{0}$ and $\var (\boldsymbol{\hat{\theta}}  - \boldsymbol{\theta_0})$ is a finite constant matrix (see, for example, \cite[p. 418]{Tyler-1983}).
Hence, we showed that the bias of $\boldsymbol{\hat{\theta}} = \operatorname{vec} \hat{\Sigma}$ is of order $1/n$ and so $b_{jk}(n) = O(1/n).$

Now, let us prove that $\E |B_{jk}|^8=O(n^{-4})$. A Taylor expansion of $u_2(d_i)$ around $d_{0i}$ yields
\begin{equation*}
	u_2(d_i) = u_2(d_{0i} + \Delta d_i) = u_2(d_{0i}) + u_2'(d_{0i})\Delta d_i + \frac{1}{2} u_2''(\xi_i)(\Delta d_i)^2
\end{equation*}
for some $\xi_i$ between $d_i$ and $d_{0i}$, where $u_2'(d)$ and $u_2''(d)$ are computed with respect to $d$. Hence
\begin{equation*}
	u_2(M_{\hat{\Sigma}}(\Xb_i)) -  u_2(M_{\Sigma}(\Xb_i)) = u_2'(M_{\Sigma}(\Xb_i))\Delta M_i + \frac{1}{2} u_2''(\xi_i)(\Delta M_i)^2
\end{equation*}
and then from H\"{o}lder's inequality, relation (\ref{eq:a_plus_b_na_8}) and the boundedness of the derivatives of $u_2$ (observe $u_2'(d) = O(d^{-2}) \leq m_1$ and $u_2''(d) = O(d^{-3}) \leq m_2$ for some finite constants $m_1, m_2$), we obtain
\begin{align}
	\E |B_{jk}|^8 &= \E \left\lvert \frac{1}{n} \sum_{i=1}^n \left( u_2(M_{\hat{\Sigma}}(\Xb_i)) - u_2(M_{\Sigma}(\Xb_i)) \right) (X_{ij} - \mu_{j})(X_{ik}-\mu_k) \right\rvert^8 \nonumber \\
	& \leq \frac{1}{n^8}n^7 \sum_{i=1}^n \E \left\lvert \left( u_2(M_{\hat{\Sigma}}(\Xb_i)) - u_2(M_{\Sigma}(\Xb_i)) \right) (X_{ij} - \mu_{j})(X_{ik}-\mu_k) \right\rvert^8 \nonumber \\
	& =  \E \left\lvert \left( u_2'(M_{\Sigma}(\Xb_1))\Delta M_1 + \frac{1}{2} u_2''(\xi_1)(\Delta M_1)^2 \right) (X_{1j} - \mu_{j})(X_{1k}-\mu_k) \right\rvert^8 \nonumber \\
	& \leq 2^7 \E \left \lvert   u_2'(M_{\Sigma}(\Xb_1))\Delta M_1 (X_{1j} - \mu_{j})(X_{1k}-\mu_k) \right\rvert^8 + \nonumber \\
	& \qquad + 2^7 \E \left\lvert \frac{1}{2} u_2''(\xi_1)(\Delta M_1)^2 (X_{1j} - \mu_{j})(X_{1k}-\mu_k) \right\rvert^8 \nonumber \\
	\begin{split}
	&\leq 2^7 m_1^8 \E \lvert \Delta M_1 (X_{1j} - \mu_{j})(X_{1k}-\mu_k) \rvert^8 \\
& \qquad + \frac{1}{2} m_2^8 \E \lvert (\Delta M_1)^2 (X_{1j} - \mu_{j})(X_{1k}-\mu_k) \rvert^8 := K_1 + K_2. \label{def:K1_K2_hub}
	\end{split}
\end{align}
Note that \cite[rel (2.4), p. 118]{Stewart-1990}
\begin{equation*}
	\hat{\Sigma}^{-1} = (\Sigma + \Delta)^{-1} = \Sigma^{-1}+\Sigma^{-1}\Delta \Sigma^{-1}
\end{equation*}
and so
\begin{align*}
	M_{\hat{\Sigma}}(\Xb_i) &= (\Xb_i - \boldsymbol{\mu})^\top \hat{\Sigma}^{-1} (\Xb_i - \boldsymbol{\mu}) = (\Xb_i - \boldsymbol{\mu})^\top (\Sigma^{-1} - \Sigma^{-1} \Delta \Sigma^{-1}) (\Xb_i - \boldsymbol{\mu}) \\
	&= M_{\Sigma}(\Xb_i)-(\Xb_i - \boldsymbol{\mu})^\top \Sigma^{-1} \Delta \Sigma^{-1} (\Xb_i - \boldsymbol{\mu}),
\end{align*}
which yields
\begin{equation} \label{def:Delta_di}
	\Delta M_i = M_{\hat{\Sigma}}(\Xb_i)  - M_{\Sigma}(\Xb_i) = - (\Xb_i - \boldsymbol{\mu})^\top \Sigma^{-1} \Delta \Sigma^{-1} (\Xb_i - \boldsymbol{\mu}).
\end{equation}
By substituting (\ref{def:Delta_di}) into $K_1$ from (\ref{def:K1_K2_hub}), we obtain
\begin{align} \label{eq:K1_v1}
	K_1 &= 2^7 m_1^8 \E \lvert \Delta M_1 (X_{1j} - \mu_{j})(X_{1k}-\mu_k) \rvert^8 \nonumber \\
	& = 2^7 m_1^8 \E \left\lvert (\Xb_i - \boldsymbol{\mu})^\top \Sigma^{-1} (\hat{\Sigma} - \Sigma) \Sigma^{-1} (\Xb_i - \boldsymbol{\mu}) (X_{1j} - \mu_{j})(X_{1k}-\mu_k) \right\rvert^8 \nonumber \\
	& = 2^7 m_1^8 \E \left\lvert \boldsymbol{v}_{\Xb_1}^\top \Delta_{\mathbb{X}} \boldsymbol{v}_{\Xb_1} c_{\Xb_1} \right\rvert^8,
\end{align}
where we denoted $\boldsymbol{v}_{\Xb_1} := \Sigma^{-1} (\Xb_i - \boldsymbol{\mu})$, $\Delta_{\mathbb{X}} := \hat{\Sigma} - \Sigma$, $\mathbb{X} := (\Xb_1,\ldots,\Xb_n)$ and $c_{\Xb_1}:=(X_{1j} - \mu_{j})(X_{1k}-\mu_k)$. Moreover, denote by $\langle \cdot, \cdot  \rangle_{F}$ the Frobenius inner product. By using this notation, we can continue in analyzing (\ref{eq:K1_v1}) by using the definition of a trace and the Cauchy-Schwarz inequality,
\begin{align}
	K_1&= 2^7 m_1^8 \E \left\lvert \operatorname{tr}(\Delta_{\mathbb{X}} \boldsymbol{v}_{\Xb_1} \boldsymbol{v}_{\Xb_1}^\top c_{\Xb_1} ) \right\rvert^8 = 2^7 m_1^8 \E \left\lvert \langle \Delta_{\mathbb{X}} , \boldsymbol{v}_{\Xb_1} \boldsymbol{v}_{\Xb_1}^\top c_{\Xb_1} \rangle_{F} \right\rvert^8 \nonumber \\
	& \leq 2^7 m_1^8 \E \left( \lVert \Delta_{\mathbb{X}} \rVert_F^8 \lVert \boldsymbol{v}_{\Xb_1} \boldsymbol{v}_{\Xb_1}^\top c_{\Xb_1} \rVert_F^8 \right) \nonumber \\
	& \leq 2^7 m_1^8 \sqrt{\E \lVert \Delta_{\mathbb{X}} \rVert_F^{16} \E \lVert \boldsymbol{v}_{\Xb_1} \boldsymbol{v}_{\Xb_1}^\top c_{\Xb_1} \rVert_F^{16}} = \sqrt{O(n^{-8})} = O(n^{-4}), \label{eq:K1_v2}
\end{align}
where we used that $\sqrt{n}(\hat{\sigma}_{jk} - \sigma_{jk}) \sim N(0,r)$ for large $n$, which yields
\begin{align*}
	\E \lVert \Delta_{\mathbb{X}} \rVert_F^{16}  = \E \left( \sum_{j=1}^p \sum_{k=1}^p |\hat{\sigma}_{jk} - \sigma_{jk}|^2 \right)^8 \leq \frac{r^8}{n^8} p^{16} \E |\chi^2|^8 = O(n^{-8}),
\end{align*}
and also the fact that $\E \lVert \boldsymbol{v}_{\Xb_1} \boldsymbol{v}_{\Xb_1}^\top c_{\Xb_1} \rVert_F^{16} < \infty$ since all moments of normal distribution are finite.

Further, due to the Cauchy-Schwarz inequality, we have
\begin{align}
	K_2 &= \frac{1}{2} m_2^8 \E \lvert (\Delta M_1)^2 (X_{1j} - \mu_{j})(X_{1k}-\mu_k) \rvert^8 = \frac{1}{2} m_2^8 \E \lvert (\boldsymbol{v}_{\Xb_1}^\top \Delta_{\mathbb{X}}  \boldsymbol{v}_{\Xb_1})^2 c_{\Xb_1} \rvert^8 \nonumber \\
	& \leq \frac{1}{2} m_2^8 \sqrt{\E (\boldsymbol{v}_{\Xb_1}^\top \Delta_{\mathbb{X}}  \boldsymbol{v}_{\Xb_1})^{32} \E c_{\Xb_1}^{16}} = \sqrt{O(n^{-32})}= O(n^{-16}),
	\label{eq:K2_v1}
\end{align}
where the last but one equality follows from using similar arguments as in the derivation of $K_1$ and from the fact that $\E c_{\Xb_1}^{16} < \infty$.

By substituting (\ref{eq:K1_v2}) and (\ref{eq:K2_v1}) into (\ref{def:K1_K2_hub}), we obtain
\begin{equation}
	\E |B_{jk}|^8 = O(n^{-4}),
\end{equation}
which was to be shown. 


\section{Proofs of statements in Section 2.3} \label{Sec:proofs}

To keep this Appendix self-contained, we restate each lemma and theorem. The numbering corresponds to that used in Section 2.3, while additional statements introduced only in the supplement are numbered with the prefix S. In what follows, $\lVert \cdot \rVert_2$ stands for the $\ell_2$-norm of a vector or the spectral norm of a matrix.

\renewcommand{\thetheorem}{2.1}
\begin{lemma}  \label{lemma:max_norm_N0I}
	Let $\Zb_1,\ldots, \Zb_n$ be a random sample from $N_p(\boldsymbol{0}, I)$. Then, for any $\Delta > 0$,
	\begin{equation} \label{prob:N0I}
		\operatorname{P}\left(\max_{j=1,\ldots,n} \lVert \Zb_j \rVert_1 \geq \sqrt{2 (p+\Delta) \log n}\right) = \smallO(1), \quad n \to \infty.
	\end{equation}
\end{lemma}

\bigskip 

\textit{Proof of Lemma \ref{lemma:max_norm_N0I}.}

The proof is an extension of the univariate case, which can be found in \cite[Lemma~1]{Cui-2014} or \cite[Exercise 2.5.9]{Vershynin-2018}. 
Denote $c=p+\Delta$. By using the Cauchy-Schwarz inequality, we can write
\begin{align}
	& \operatorname{P}\left(\max_{j=1,\ldots,n} \sum_{k=1}^p |z_{jk}| \geq \sqrt{2 c \log n}\right) \leq \operatorname{P}\left(\bigcup_{j=1}^n \left\{ \sum_{k=1}^p |z_{jk}| \geq \sqrt{2 c \log n} \right\} \right) \nonumber \\ 
	&\quad  \leq \sum_{j=1}^n \operatorname{P}\left(\sum_{k=1}^p |z_{jk}| \geq \sqrt{2 c \log n}\right) = n \operatorname{P}\left(\sum_{k=1}^p |z_{1k}| \geq \sqrt{2 c \log n}\right) \nonumber \\
	&\quad \leq  n \operatorname{P}\left(\sqrt{p} \left( \sum_{k=1}^p z_{1k}^2 \right)^{1/2} \geq \sqrt{2 c\log n}\right)= n \operatorname{P} \left(p \sum_{k=1}^p z_{1k}^2  \geq 2 c\log n\right) \nonumber \\
	& \quad = n \operatorname{P}\left( \chi^2_p  \geq \frac{2c}{p}\log n\right). \label{eq:propN0I_v1}
\end{align}
Now, let's use the following tail approximation for $\chi^2_p$ \cite[relation 26.4.12, p. 941]{Abra-1972}
\begin{equation*}
	P(\chi^2_p > t) \sim \frac{t^{p/2-1} e^{-t/2}}{2^{p/2-1} \Gamma(p/2)}, \quad t \to \infty.
\end{equation*}
Then, (\ref{eq:propN0I_v1}) turns into
\begin{align}
	\operatorname{P}\left(\max_{j=1,\ldots,n} \sum_{k=1}^p |z_{jk}| \geq \sqrt{2 c \log n}\right) &\sim n \frac{(\frac{2c}{p} \log n)^{p/2-1} e^{-\frac{c}{p}\log n}}{2^{p/2-1} \Gamma(p/2)} \nonumber \\
	& = \mathcal{O} \left( n^{1-c/p} \left( \frac{2c}{p} \log n \right)^{p/2-1} \right) \nonumber \\
	& = \mathcal{O}\left( n^{-\Delta/p}(1+\smallO(1)) \right) = \smallO(1), \label{eq:propN0I_v2}
\end{align}
since $1-\frac{c}{p} = 1-\frac{p+\Delta}{p} = -\frac{\Delta}{p}<0$. For $p \leq 2$, (\ref{eq:propN0I_v2}) holds automatically, for $p\geq 3$, we use that $\log^a(n) = \smallO(n^b)$ for any $a,b>0$. In our case, $a=p/2-1 >0 $ and $b=c/p-1 >0$.

\qedsymbol

\medskip
\renewcommand{\thetheorem}{S1}
\begin{lemma} \label{lem:Lipschitz_of_CholeskyInv}
	Let $\Sigma=LL^\top>0$ be a covariance matrix and $S=S(n)=\hat{L}\hat{L}^\top >0$ its estimator based on a random sample of size $n$ such that $\lVert \Sigma^{-1}\rVert_2 \lVert \Sigma - S \rVert_2 \leq \alpha < 1$. Then
	\begin{equation}
		\lVert \hat{L}^{-1} - L^{-1} \rVert_2 \leq c(\Sigma) \lVert S-\Sigma \rVert_2 ,
	\end{equation}
	where $c(\Sigma)>0$ is a constant depending only on $\Sigma$.
\end{lemma}

This Lemma states that the inverse of Cholesky decomposition is locally Lipschitz on a~set of positive definite matrices. Therefore, if $\lVert \Sigma - S \rVert_2 =\smallO(1)$, then the Lemma implies that $\lVert \hat{L}^{-1} - L^{-1} \rVert_2 = \smallO(1)$.

\medskip 
\textit{Proof of Lemma \ref{lem:Lipschitz_of_CholeskyInv}.} Consider a polar decomposition of the Cholesky factors $\hat{L}$ and $L$, i.e. $\hat{L} = U_{S} S^{1/2}$ and $L = U_\Sigma \Sigma^{1/2}$, where $U_{S}$ and $U_\Sigma$ are unitary matrices (for details, see, for example, \cite[Section 6.8.4]{Higham-2008}). Then, $\hat{L}^{-1} = S^{-1/2} U_{S}^\top$, $L^{-1} = \Sigma^{-1/2}U_\Sigma^\top$ and by using triangle inequality, submultiplicativity of the spectral norm and the fact that $U_{S}$ is unitary, we can write
\begin{align}
	\lVert \hat{L}^{-1} - L^{-1} \rVert_2 &= \lVert S^{-1/2} U_{S}^\top - \Sigma^{-1/2}U_\Sigma^\top \rVert_2 \nonumber \\
	& \leq \lVert (S^{-1/2} - \Sigma^{-1/2}) U_{S}^\top \rVert_2 + \lVert \Sigma^{-1/2} (U_{S}^\top - U_\Sigma^\top) \rVert_2 \nonumber \\
	& \leq \lVert S^{-1/2} - \Sigma^{-1/2} \rVert_2 \lVert U_{S}^\top \rVert_2 + \lVert \Sigma^{-1/2} \rVert_2 \lVert U_{S}^\top - U_\Sigma^\top \rVert_2 \nonumber \\
	&= \lVert {S}^{-1/2} - \Sigma^{-1/2} \rVert_2 + \lVert \Sigma^{-1/2} \rVert_2 \lVert U_{S}^\top - U_\Sigma^\top \rVert_2. \label{eq:ChL_v1}
\end{align}
Since $U_{S} = \hat{L}S^{1/2}$ and $U_\Sigma=L \Sigma^{-1/2}$, and by using the theorem assumptions, we can continue analyzing (\ref{eq:ChL_v1}) and write
\begin{align}
	\lVert \hat{L}^{-1} - L^{-1} \rVert_2 &\leq \lVert S^{-1/2} - \Sigma^{-1/2} \rVert_2 + \lVert \Sigma^{-1/2} \rVert_2 \lVert S^{-1/2} \hat{L}^\top - \Sigma^{-1/2} L^\top \rVert_2 \nonumber \\
	&= \lVert S^{-1/2} - \Sigma^{-1/2} \rVert_2 + \lVert \Sigma^{-1/2} \rVert_2 \lVert (S^{-1/2}-\Sigma^{-1/2})\hat{L}^\top + \Sigma^{-1/2}(\hat{L}^\top - L^{\top}) \rVert_2 \nonumber \\
	& \leq \left( 1 + \lVert \Sigma^{-1/2} \rVert_2 \lVert \hat{L}^\top \rVert_2 \right) \lVert S^{-1/2} - \Sigma^{-1/2} \rVert_2 +   \lVert \Sigma^{-1/2} \rVert_2^2 \lVert \hat{L}^\top - L^{\top} \rVert_2 . \label{eq:ChL_v2} 
\end{align}

Recall that for large enough $n$, $S$ satisfies the stability condition 
\begin{equation} \label{cond:stability}
	\lVert \Sigma^{-1} \rVert_2 \lVert \Sigma - S \rVert_2 < 1.
\end{equation}
For any real positive definite matrices $A$ and $B$, it holds that \cite[Theorem 6.2]{Higham-2008}
\begin{equation} \label{bound_for_squareRoot}
	\lVert A^{-1/2} - B^{-1/2} \rVert_2 \leq \frac{\lVert A^{-1} - B^{-1}\rVert_2}{\sqrt{\lambda_{min}(A^{-1})} + \sqrt{\lambda_{min}(B^{-1})}},
\end{equation}
where $\lambda_{min}(M)$ stands for the smallest eigenvalue of matrix $M$. Further, for any nonsingular matrices $A$ and $B$ satisfying $\lVert A^{-1}(B-A) \rVert_2 <1$, it holds \cite[p. 381]{Horn-2013}
\begin{equation} \label{bound_for_inverse}
	\lVert A^{-1} - B^{-1} \rVert_2 \leq \frac{\lVert B^{-1} \rVert_2^2 \lVert B - A \rVert_2}{1-\lVert B^{-1} (B - A )\rVert_2} \leq \frac{\lVert B^{-1} \rVert_2^2 \lVert B - A \rVert_2}{1-\lVert B^{-1}\rVert_2 \lVert B - A \rVert_2}.
\end{equation}
By combining (\ref{bound_for_squareRoot}) and (\ref{bound_for_inverse}) for $A = S$ and $B=\Sigma$ and using $\lambda_{min}(\Sigma^{-1}) = 1/\lVert \Sigma \rVert_2$, we obtain
\begin{align}
	\lVert S^{-1/2} - \Sigma^{-1/2} \rVert_2 &\leq \frac{\lVert S^{-1} - \Sigma^{-1}\rVert_2}{\sqrt{\lambda_{min}(S^{-1})} + \sqrt{\lambda_{min}(\Sigma^{-1})}} \leq \frac{\lVert S^{-1} - \Sigma^{-1}\rVert_2}{\sqrt{\lambda_{min}(\Sigma^{-1})}} \nonumber \\
	& \leq \frac{1}{\sqrt{\lambda_{min}(\Sigma^{-1})}}  \frac{\lVert \Sigma^{-1} \rVert_2^2 \lVert \Sigma - S \rVert_2}{1-\lVert \Sigma^{-1}\rVert_2 \lVert \Sigma - S \rVert_2} \leq  \lVert \Sigma \rVert_2^{1/2} \frac{\lVert \Sigma^{-1} \rVert_2^2 \lVert \Sigma - S \rVert_2}{1-\lVert \Sigma^{-1}\rVert_2 \lVert \Sigma - S \rVert_2} \label{eq:ChL_v3b}.
\end{align}

For a real symmetric matrix $A$, it holds that
\begin{equation} \label{real_sym_eig}
	\lVert A \rVert_2 = \max_i |\lambda_i (A)| = \max \{ |\lambda_{min}(A)|, |\lambda_{max}(A)| \} \geq \lambda_{max}(A).
\end{equation}
Using the definition of the spectral norm, 
Weyl's inequality \cite[Theorem 4.3.1]{Horn-2013}, and the relations (\ref{real_sym_eig}) and $\lVert L \rVert_2 = \lVert \Sigma \rVert_2^{1/2}$, we obtain
\begin{equation} \label{eq:bound_Lhat}
	\begin{split}
	\lVert \hat{L} \rVert_2 &= \lambda_{max}(S) \leq \lambda_{max}(\Sigma) + \lambda_{max}(S-\Sigma) \\
	& \leq \lVert L \rVert_2 + \lVert S-\Sigma \rVert_2 = \lVert \Sigma \rVert_2^{1/2} + \lVert S-\Sigma \rVert_2.
		\end{split}
\end{equation}
Furthermore, it follows from \cite[Theorem 1.4]{Sun-1991} that
\begin{align}
	\lVert \hat{L}^\top - L^{\top} \rVert_2 = \lVert \hat{L} - L \rVert_2 \leq \lVert L \rVert_2 \frac{\lVert \Sigma^{-1} \rVert_2}{\sqrt{2(1-\lVert \Sigma^{-1} \rVert_2 \lVert S-\Sigma \rVert_2)}} \lVert S-\Sigma \rVert_2. \label{eq:ChL_v4}
\end{align}
Substituting (\ref{eq:ChL_v3b}), (\ref{eq:ChL_v4}), and the assumptions of the Lemma into (\ref{eq:ChL_v2}) yields
\begin{align*}
	&\lVert \hat{L}^{-1} - L^{-1} \rVert_2 \leq \Bigg( \frac{\left( 1+\lVert \Sigma^{-1/2}\rVert_2 \left(\lVert \Sigma \rVert_2^{1/2}+\lVert S-\Sigma \rVert_2 \right) \right)\lVert \Sigma^{-1} \rVert_2 \lVert \Sigma \rVert_2^{1/2}}{ 1-\lVert \Sigma^{-1} \rVert_2 \lVert \Sigma - S \rVert_2} + \\
	& \hspace{3.5cm} + \frac{\lVert \Sigma^{-1/2} \rVert_2^2 }{\sqrt{2(1-\lVert \Sigma^{-1} \rVert_2 \lVert \Sigma - S \rVert_2))}} \Bigg) \lVert \Sigma^{-1} \rVert_2 \lVert S-\Sigma \rVert_2  \\
	&\quad \leq \Bigg( \frac{\left( 1+\lVert \Sigma^{-1/2}\rVert_2 \left(\lVert \Sigma \rVert_2^{1/2}+\alpha \lVert \Sigma^{-1} \rVert_2 \right) \right)\lVert \Sigma^{-1} \rVert_2 \lVert \Sigma \rVert_2^{1/2}}{ 1-\alpha}+ \frac{\lVert \Sigma^{-1/2} \rVert_2^2 }{\sqrt{2(1-\alpha))}} \Bigg) \times \\
	& \hspace{6cm} \times \lVert \Sigma^{-1} \rVert_2 \lVert S-\Sigma \rVert_2  \\
	&= c(\Sigma) \lVert S-\Sigma \rVert_2,
\end{align*}
which was to be proved.

\qedsymbol

\bigskip
\renewcommand{\thetheorem}{2.2}
\begin{theorem} \label{theorem:upper_bound}
	Let $\beta \in (0,1)$ and $p \in \mathbb{N}$ be fixed numbers. If the parameters $\ab=\ab(n)$ and $\Sigma=LL^\top >0$ in model~(\ref{model}) satisfy the condition
	\begin{equation} \label{cond:a_UB}
		\liminf_{n \to \infty} \frac{\lVert L^{-1}\ab \rVert_1}{\sqrt{p\log n}} > \sqrt{2}(1+\sqrt{1-\beta}),
	\end{equation}
	then the estimator (\ref{def:eta})--(\ref{cond:delta}) satisfies
	\begin{equation} \label{eq:Hamming_risk_T2}
		\liminf_{n \to \infty} \sup_{\etab \in H_{n,\beta}} \E_{\etab} |\etab - \hat{\etab}| = 0.
	\end{equation}
\end{theorem}

\medskip 
\textit{Proof of Theorem \ref{theorem:upper_bound}.}
The maximum Hamming risk of the estimator $\hat{\etab}=(\hat{\eta}_1,\ldots,\hat{\eta}_p)$ satisfies
\begin{align}
	\sup_{\etab \in H_{n,\beta}} \E_\etab |\hat{\etab} - \etab| &= \sup_{\etab \in H_{n,\beta}} \E_\etab \sum_{i=1}^n |\hat{\eta}_i -\eta_i| = \sup_{\etab \in H_{n,\beta}} \left( \sum_{\substack{i=1 \\ \eta_i=0}}^n \E_\etab |\hat{\eta}_i -\eta_i| + \sum_{\substack{i=1 \\ \eta_i=1}}^n \E_\etab |\hat{\eta}_i -\eta_i| \right) \nonumber \\
	&= \sup_{\etab \in H_{n,\beta}} \Bigg( \sum_{\substack{i=1 \\ \eta_i=0}}^n \operatorname{P}_{\boldsymbol{0}} \left(\lVert \hat{L}^{-1}\Xb_i \rVert_1 > \sqrt{2(p+\delta)\log n}\right) + \nonumber \\
	& \hspace{3cm} + \sum_{\substack{i=1 \\ \eta_i=1}}^n \operatorname{P}_{\boldsymbol{a}}\left(\lVert \hat{L}^{-1}\Xb_i \rVert_1 \leq \sqrt{2(p+\delta)\log n}\right) \Bigg) \nonumber \\
	&=:  I_1 + I_2. \label{def:I1_I2}
\end{align}
First, let's have a look at the term $I_1$. Choose $\epsilon >0$ and use the standard triangle inequality for $\hat{L}^{-1} = (\hat{L}^{-1}-L^{-1}) + L^{-1}$, then
\begin{align}
	I_1 &= \sup_{\etab \in H_{n,\beta}} \sum_{\substack{i=1 \\ \eta_i=0}}^n \operatorname{P}_{\boldsymbol{0}}\left(\lVert \hat{L}^{-1}\Xb_i \rVert_1 > \sqrt{2(p+\delta)\log n}\right) \nonumber \\
	&\leq \sup_{\etab \in H_{n,\beta}} \sum_{\substack{i=1 \\ \eta_i=0}}^n \operatorname{P}_{\boldsymbol{0}}\left(\lVert (\hat{L}^{-1}-L^{-1})\Xb_i \rVert_1  + \lVert L^{-1}\Xb_i \rVert_1 > \sqrt{2(p+\delta)\log n}\right) \nonumber \\
	& \leq \sup_{\etab \in H_{n,\beta}} \sum_{\substack{i=1 \\ \eta_i=0}}^n \operatorname{P}_{\boldsymbol{0}}\left(\lVert (\hat{L}^{-1} - L^{-1})\Xb_i \rVert_1 >\epsilon\right) \nonumber \\
	& \hspace{3cm} + \operatorname{P}_{\boldsymbol{0}}\left(\lVert L^{-1}\Xb_i \rVert_1 > \sqrt{2(p+\delta)\log n} - \epsilon\right)  \nonumber \\
	& =: I_1^a + I_1^b. \label{I1_v1}
\end{align}
Then, look at the first term on the right-hand side of (\ref{I1_v1}).
Recall that 
\begin{equation} \label{eq:bound_for_l1_norm}
	\lVert \boldsymbol{v} \rVert_1 \leq \sqrt{p} \lVert \boldsymbol{v} \rVert_2, \quad \boldsymbol{v} \in \mathbb{R}^p,
\end{equation}
which can be seen, for example, from the Cauchy-Schwarz inequality. By using the reverse triangle inequality, relation (\ref{eq:bound_for_l1_norm}), and the definition of the spectral norm, we obtain
\begin{align} \label{eq:LhatX_LX_v1}
	\lVert (\hat{L}^{-1} - L^{-1}) \Xb_i \rVert_1 \leq \sqrt{p} \lVert (\hat{L}^{-1} - L^{-1}) \Xb_i \rVert_2 \leq \sqrt{p} \lVert \hat{L}^{-1} - L^{-1} \rVert_{2} \lVert \Xb_i \rVert_2 .
\end{align}
By using Lemma \ref{lem:Lipschitz_of_CholeskyInv} (its assumption holds for large enough $n$ from the consistency of $\hat{\Sigma}$), relation (\ref{eq:LhatX_LX_v1}) gives 
\begin{align} \label{eq:LhatX_LX_v2}
	\lVert (\hat{L}^{-1} - L^{-1}) \Xb_i \rVert_1 \leq \sqrt{p} C_{\Sigma} \lVert \hat{\Sigma} - \Sigma \rVert_{2} \lVert \Xb_i \rVert_2 ,
\end{align}
where $C_{\Sigma}\geq 0$ is a constant depending only on $\Sigma$.
By using (\ref{eq:LhatX_LX_v2}) and the fact that $\Xb_1,\ldots,\Xb_n$ are iid, and applying Markov's inequality and the Cauchy-Schwarz inequality, we obtain 
\begin{align} \label{eq:I1a_v1}
	I_1^a &= \sup_{\etab \in H_{n,\beta}} \sum_{\substack{i=1 \\ \eta_i=0}}^n \Prob{\boldsymbol{0}}{ \lVert (\hat{L}^{-1} - L^{-1}) \Xb_i \rVert_1 >\epsilon} \leq n \Prob{\boldsymbol{0}}{ \lVert (\hat{L}^{-1} - L^{-1}) \Xb_i \rVert_1>\epsilon} \nonumber \\
	& \leq n\Prob{\boldsymbol{0}}{  \lVert \hat{\Sigma} - \Sigma \rVert_{2} \lVert \Xb_1 \rVert_2  > \frac{\epsilon}{\sqrt{p} C_{\Sigma}}} = n\Prob{\boldsymbol{0}}{  \lVert \hat{\Sigma} - \Sigma \rVert_{2}^4 \lVert \Xb_1 \rVert_2^4  > \frac{\epsilon^4}{p^2 C_{\Sigma}^4}} \nonumber \\
	& \leq \frac{n p^2C_{\Sigma}^4 }{\epsilon^4} \E{} \left( \lVert \hat{\Sigma} - \Sigma \rVert_{2}^4 \lVert \Xb_1 \rVert_2^4 \right) \leq \frac{n p^2C_{\Sigma}^4 }{\epsilon^4} \left( \E{} \lVert \hat{\Sigma} - \Sigma \rVert_{2}^8 \right)^{1/2} \left( \E{} \lVert \Xb_1 \rVert_2^8 \right)^{1/2}.
\end{align}
Since $\E{} \lVert \Xb_1 \rVert_2^8$ is finite and independent of $n$, it remains to show that $n\left( \E{} \lVert \hat{\Sigma} - \Sigma \rVert_{2}^8 \right)^{1/2} = \smallO(1)$.
Since the~spectral norm is bounded from above by the Frobenius norm, we have
\begin{align} \label{eq:ES_Sigma8}
	\E{} \lVert \hat{\Sigma} - \Sigma \rVert_{2}^8 &\leq \E{} \lVert \hat{\Sigma} - \Sigma \rVert_{F}^8 = \E{} \left( \sum_{i,j=1}^p (\hat{\sigma}_{ij} - \sigma_{ij} )^2 \right)^4 \nonumber \\
	& \leq p^4 \max_{j,k} \E{} (\hat{\sigma}_{jk}-\sigma_{jk})^8 =: p^4 \E{} (\hat{\sigma}_{j'k'}-\sigma_{j'k'})^8,
\end{align}
where $\hat{\sigma}_{jk}$, resp. $\sigma_{jk}$, stands for the $j,k$-th element of $\hat{\Sigma}$, resp. $\Sigma$, and where we denoted $\{ j',k' \} = \arg \max_{\{j,k\}} \E{} (\hat{\sigma}_{jk}-\sigma_{jk})^8$. The Rosenthal inequality \cite[Theorem~2.9]{Petrov-1995} states that for independent random variables $Y_1,\ldots,Y_n$ such that $\E{}Y_i = 0$ and $r \geq 2$, it holds that 
\begin{equation} \label{Rosenthal}
	\E{} \left|\sum_{i=1}^n Y_i \right|^r \leq C_r \left( \sum_{i=1}^n \E{} |Y_i|^r + \left(\sum_{i=1}^n \E{ Y_i^2}  \right)^{r/2} \right),
\end{equation}
where $C_r$ is a positive constant depending only on $r$.
Recall that $\sigma_{j'k'} = \E{} (g_{j'k'}(\Xb_1)).$
For $r=8$ and considering $Y_i/n$ instead of $Y_i$, inequality (\ref{Rosenthal}) takes the form
\begin{equation} \label{eq:Es_mu8}
	\E \left|\frac{1}{n}\sum_{i=1}^n Y_i \right|^8 = n^{-8} \E \left( \sum_{i=1}^n Y_i \right)^8 \leq C_8 \left( n^{-7} \E |Y_1|^8 + n^{-4}  (\E Y_1^2)^4 \right) .
\end{equation}
Now, consider $Y_i = g_{j'k'}(\Xb_i) - \sigma_{j'k'}$, $i=1,\ldots,n$, then both $\E |Y_1|^8$ and $\E Y_1^2$ are finite, and so $\E \left|\frac{1}{n}\sum_{i=1}^n Y_i \right|^8 = O(n^{-4})$. Further, by using the inequality $|a+b+c|^8 \leq 3^7 (|a|^8 + |b|^8 + |c|^8), \, a,b,c \in \mathbb{R}$, and combining (\ref{eq:Es_mu8}) with the assumptions that $b_{j'k'}(n) = O(n^{-1})$ and $\E |R_{j'k'}|^8 = O(n^{-4})$, we get
\begin{align}
	\E |\hat{\sigma}_{j'k'} - \sigma_{j'k'}|^8 &= \E \left|\frac{1}{n}\sum_{i=1}^n Y_i + b_{j'k'}(n) + R_{j'k'} \right|^8 \nonumber \\
	& \leq 3^7 \left( \E \left|\frac{1}{n}\sum_{i=1}^n Y_i \right|^8 +\E |b_{j'k'}(n)|^8 + \E |R_{j'k'}|^8 \right)  \nonumber \\
	&= O(n^{-4}) + O(n^{-8}) + O(n^{-4}) =  O(n^{-4}). \label{eq:Es_mu8_On8}
\end{align}
Substituting (\ref{eq:Es_mu8_On8}) into (\ref{eq:ES_Sigma8}) yields $\E \lVert \hat{\Sigma} - \Sigma \rVert_{2}^8 = O(n^{-4})$, and substituting this result into (\ref{eq:I1a_v1}) gives
\begin{align} \label{I1a_is_o1}
	I_1^a = O(n (n^{-4})^{1/2}) = O(n^{-1}) = \smallO(1).
\end{align}

Now, let's have a look at the second term on the right-hand side of (\ref{I1_v1}).
The Gaussian concentration inequality \cite[Theorem 5.6]{Massart-2013}
states that, for an $L$-Lipschitz function $f: \mathbb{R}^p \to \mathbb{R}$,
\begin{equation} \label{eq:tail_ineq_for_Lipschitz}
	\operatorname{P}\left(f(\Zb) - \E{} (f(\Zb)) > t\right) \leq e^{-\frac{t^2}{2L^2}}, \quad \Zb \sim N_p(\boldsymbol{0},I_p), \quad  t >0.
\end{equation}
Furthermore, $L^{-1}\Xb_i \sim N_p(0,I_p)$, and therefore $\E{} \lVert L^{-1}\Xb_i \rVert_1 = p \sqrt{\frac{2}{\pi}}$ \cite[Exercise 2.5.1]{Vershynin-2018}. 
Since the $\ell_1$-norm is $\sqrt{p}$-Lipschitz (see relation (\ref{eq:bound_for_l1_norm})), applying (\ref{eq:tail_ineq_for_Lipschitz}) and using the fact that the number of $\eta_i$ equal to 0 is at most $n$, we obtain from (\ref{I1_v1}) that
\begin{align}
	I_1^b &\leq n \operatorname{P}\left(\lVert L^{-1}\Xb_i \rVert_1 - p \sqrt{\frac{2}{\pi}} > \sqrt{p(2+\delta)\log n} - \epsilon - p \sqrt{\frac{2}{\pi}}\right) \nonumber \\
	& \leq n \exp \left( -\frac{1}{2p} \left( \sqrt{p(2+\delta)\log n} - \epsilon - p \sqrt{\frac{2}{\pi}}\right)^2 \right) \nonumber \\
	& = n \exp \left( -\frac{1}{2p} p(2+\delta)\log n (1+\smallO(1)) \right) \sim n n^{-(2+\delta)/2} = n^{-\delta/2} = \smallO(1). \label{eq:I1b}
\end{align}
Substituting (\ref{I1a_is_o1}) and (\ref{eq:I1b}) into (\ref{I1_v1}) yields
\begin{equation} \label{eq:I1_is_o1}
	I_1 \leq I_1^a + I_1^b = \smallO(1), \quad n \to \infty.
\end{equation}

Similar arguments will be used for the second term on the right-hand side of (\ref{def:I1_I2}). By applying the standard triangle inequality to $\xb=(\xb+\yb)+(-\yb), \, \xb,\yb, \in \mathbb{R}^p,$ we get
\begin{equation} \label{ineq:reverse_triangle}
	\lVert \xb+\yb\rVert_1 \geq \lVert \xb \rVert_1 - \lVert \yb \rVert_1.
\end{equation}
Combining (\ref{ineq:reverse_triangle}) and the standard triangle inequality gives, for $\epsilon>0$
\begin{align}
	I_2 &= \sup_{\etab \in H_{n,\beta}} \sum_{\substack{i=1 \\ \eta_i=1}}^n \Prob{\ab}{ \lVert \hat{L}^{-1} \ab + \hat{L}^{-1} (\Xb_i - \ab) \rVert_1  \leq \sqrt{p(2+\delta)\log n}} \nonumber \\
	&\leq \sup_{\etab \in H_{n,\beta}} \sum_{\substack{i=1 \\ \eta_i=1}}^n \Prob{\ab}{ \lVert \hat{L}^{-1} \ab \rVert_1 - \lVert \hat{L}^{-1} (\Xb_i - \ab) \rVert_1  \leq \sqrt{p(2+\delta)\log n}}  \nonumber \\
	&=\sup_{\etab \in H_{n,\beta}} \sum_{\substack{i=1 \\ \eta_i=1}}^n \Prob{\ab}{ \lVert \hat{L}^{-1} (\Xb_i - \ab) \rVert_1  \geq -\sqrt{p(2+\delta)\log n} + \lVert \hat{L}^{-1} \ab \rVert_1 }  \nonumber \\
	&= \sup_{\etab \in H_{n,\beta}} \sum_{\substack{i=1 \\ \eta_i=1}}^n \Prob{\ab}{ \lVert ((\hat{L}^{-1}-L^{-1}) + L^{-1}) (\Xb_i - \ab) \rVert_1  \geq -\sqrt{p(2+\delta)\log n} + \lVert \hat{L}^{-1} \ab \rVert_1 } \nonumber \\
	& \leq \sup_{\etab \in H_{n,\beta}} \sum_{\substack{i=1 \\ \eta_i=1}}^n \operatorname{P}_{\ab} \Big( \lVert (\hat{L}^{-1}-L^{-1})(\Xb_i - \ab) \rVert_1 + \lVert L^{-1} (\Xb_i - \ab) \rVert_1  \geq \nonumber \\
	& \hspace{8cm} \geq -\sqrt{p(2+\delta)\log n} + \lVert \hat{L}^{-1} \ab \rVert_1 \Big)  \nonumber \\
	& \leq \sup_{\etab \in H_{n,\beta}} \sum_{\substack{i=1 \\ \eta_i=1}}^n \Prob{\ab}{ \lVert (\hat{L}^{-1}-L^{-1})(\Xb_i - \ab) \rVert_1 >\epsilon } + \nonumber \\
	& \hspace{4cm} + \Prob{\ab}{ \lVert L^{-1} (\Xb_i - \ab) \rVert_1  \geq -\sqrt{p(2+\delta)\log n} + \lVert \hat{L}^{-1} \ab \rVert_1 -\epsilon} \nonumber \\
	&=: I_2^a + I_2^b. \label{I2_v1} 
\end{align}
By applying the same arguments as for $I_1^a$ with $\boldsymbol{Y}_i = \Xb_i - \ab$ in place of $\Xb_i$ and using the bound $\sum_{i=1}^n \eta_i \leq n$, we get a result analogous to (\ref{I1a_is_o1}),
\begin{equation} \label{I2a_is_o1}
	I_2^a = \smallO(1).
\end{equation}

Note that for all sufficiently large $n$, we can write $\sqrt{(2+\delta)\log n}$ in the form $\sqrt{2\log n} (1+\delta_1)$ with some $\delta_1=\smallO(1)$. Furthermore, assumption (\ref{cond:a_UB}) implies that for some positive $\delta_2 = \smallO(1)$ and all sufficiently large $n$,
$$ \lVert L^{-1}\ab \rVert_1 > \sqrt{2}(1+\sqrt{1-\beta}) \sqrt{p\log n} (1+\delta_2),$$
where we may take $\delta_2$ such that $\delta_2 > \delta_1$ and $n^{-c(\delta_2-\delta_1)} = \smallO(1)$ for any $c>0$. Now, by applying the reverse triangle inequality (\ref{ineq:reverse_triangle}), we obtain 
\begin{align} \label{eq:lower_bound_Lhatinva}
	\begin{split}
	\lVert \hat{L}^{-1} \ab \rVert_1 &\geq \lVert L^{-1} \ab \rVert_1 - \lVert (\hat{L}^{-1} - L^{-1}) \ab \rVert_1 \\
	&> \sqrt{2p \log n} (1+\sqrt{1-\beta})(1+\delta_2) + \smallO(1).
		\end{split}
\end{align}
For $\lVert \hat{L}^{-1} \ab \rVert_1$ satisfying (\ref{eq:lower_bound_Lhatinva}) and sufficiently small $\epsilon$, it holds that
\begin{align*}
	-\sqrt{p(2+\delta)\log n} + \lVert \hat{L}^{-1} \ab \rVert_1 -\epsilon &= -\sqrt{2p\log n}(1+\delta_1) + \lVert \hat{L}^{-1} \ab \rVert_1 -\epsilon \\ 
	& > -\sqrt{2p\log n}(1+\delta_2) + \lVert \hat{L}^{-1} \ab \rVert_1 -\epsilon  >0,
\end{align*}
so we can apply (\ref{eq:tail_ineq_for_Lipschitz}) to $I_2^b$ with $\Zb_i = L^{-1}\boldsymbol{Y}_i$. 
Using the bound $\sum_{i=1}^n \eta_i \leq c_1n^{1-\beta}$ and arguments analogous to those used in the derivation of (\ref{eq:I1b}), we obtain 
\begin{align} \label{eq:I2b_v1}
	I_2^b \leq c_1 n^{1-\beta} \exp \left\{ -\frac{1}{2p} \left( - \sqrt{p(2+\delta)\log n }(1+\smallO(1))  + \lVert \hat{L}^{-1} \ab \rVert_1 \right)^2 \right\},
\end{align}
and hence, (\ref{eq:I2b_v1}) turns into
\begin{align} \label{eq:I2b_is_o1}
	I_2^b &\sim  \mathcal{O} \left( n^{1-\beta} \exp \left\{ -\frac{1}{2p}\left( -\sqrt{2p\log n}(1+\delta_1) + \sqrt{2p \log n} (1+\sqrt{1-\beta})(1+\delta_2)  \right)^2 \right\} \right) \nonumber \\
	&= \mathcal{O} \left( n^{1-\beta} \exp \left\{ -\log n \left( (1+\sqrt{1-\beta})(1+\delta_2) -(1+\delta_1) \right)^2 \right\}\right) \nonumber \\
	&= \mathcal{O} \left( n^{1-\beta} n^{-\left( (1+\sqrt{1-\beta})(1+\delta_2) -(1+\delta_1) \right)^2} \right) = \mathcal{O} \left( n^{1-\beta} n^{-\left( \sqrt{1-\beta} + (\delta_2-\delta_1) \right)^2} \right) = \smallO(1) . 
\end{align}
Substituting (\ref{I2a_is_o1}) and (\ref{eq:I2b_is_o1}) into (\ref{I2_v1}) yields
\begin{equation} \label{eq:I2_is_o1}
	I_2 = I_2^a + I_2^b = \smallO(1).
\end{equation}
Finally, by substituting (\ref{eq:I1_is_o1}) and (\ref{eq:I2_is_o1}) into (\ref{def:I1_I2}), we obtain 
$$\sup_{\etab \in H_{n,\beta}} \E_{\etab} |\hat{\etab} - \etab| = I_1 + I_2 = \smallO(1) + \smallO(1) = \smallO(1),$$
which completes the proof.

\qedsymbol

\bigskip
\renewcommand{\thetheorem}{2.3}
\begin{theorem} \label{theorem:lower_bound}
	Let $\beta \in (0,1)$ and $p \leq \lfloor \frac{8}{\beta} (1+\sqrt{1-\beta})-4 \rfloor$ be fixed numbers. If the parameters $\ab=\ab(n)$ and $\Sigma=LL^\top$ in model~(\ref{model}) satisfy the condition
	\begin{equation} \label{cond:a_LB}
		\liminf_{n \to \infty} \frac{\lVert L^{-1}\ab \rVert_1}{\sqrt{p\log n}} < \sqrt{2}(1+\sqrt{1-\beta}),
	\end{equation}
	then
	\begin{equation}
		\liminf_{n \to \infty} \inf_{\tilde{\etab}} \sup_{\etab \in H^\pm_{n,\beta}} \E_{\etab} |\etab - \tilde{\etab}| > 0,
	\end{equation}
	where the infimum is taken over all estimators $\tilde{\etab}$ of $\etab$ based on $\Xb_1,\ldots,\Xb_n$.
\end{theorem}

\medskip 
\textit{Proof of Theorem \ref{theorem:lower_bound}.}
The beginning of the proof proceeds along the same lines as that of Theorem~3 in \cite{Cui-2014}, but we recall the argument here to keep the proof self-contained. Denote the maximum risk of $\tilde{\etab}$ by $\mathcal{R}(\tilde{\etab},\etab)$, i.e. $\mathcal{R}(\tilde{\etab},\etab) = \sup_{\etab \in H_{n,\beta}^\pm} \E_{\etab} |\tilde{\etab} - \etab|$.  In order to prove the lower bound on $\inf_{\tilde{\etab}} \mathcal{R}(\tilde{\etab},\etab)$, first, bound $\mathcal{R}(\tilde{\etab},\etab)$ by the Bayes risk with a properly chosen prior distribution of $\etab$. Denote the prior density function of $\etab$ by $\pi(\etab)$ and its domain by $\mathcal{E}^n = \left\{ \etab \in \{0,1\}^n \right\}$. Then, the maximum risk of $\tilde{\etab}$ is always greater than or equal to the Bayes risk
\begin{equation} \label{eq:bound_by_Bayes_risk}
	\max_{\etab \in \mathcal{E}^n} \mathcal{R}(\tilde{\etab},\etab) \geq \int_{\mathcal{E}^n} \mathcal{R}(\tilde{\etab},\etab) \pi(\etab) d \etab. 
\end{equation}
Under model (\ref{model}), the components $\eta_1,\ldots,\eta_n$ of $\etab$ can be considered iid, and the most natural choice of their prior distribution is $Bernoulli(r)$ with $r=n^{1-\beta}/n = n^{-\beta}$. Then, the prior distribution of $\eta_j$ can be written as $\pi_j = r\delta_1 + (1-r)\delta_0$, $j=1,\ldots,n$, where $\delta_x$ is the Dirac measure, which has the total mass of 1 at the point~$x$. Since $\eta_1,\ldots,\eta_n$ are considered independent, the prior distribution of the vector $\etab$ has the product form
\begin{equation}
	\pi(\etab) = \prod_{j=1}^n \pi_j (\eta_j), \quad \quad \pi_j = r\delta_1 + (1-r)\delta_0.
\end{equation}
Denote $B_{n,\beta} := \left\{ \etab \in H_{n,\beta}^\pm \right\} \subseteq \mathcal{E}^n$ and define the probability measure $\hat{\pi}$ on $B_{n,\beta}$ by the relation 
\begin{equation*}
	\hat{\pi}(A) := \pi(A|B_{n,\beta})=\frac{\pi(A \cap B_{n,\beta})}{\pi(B_{n,\beta})}, \quad A \subset \mathcal{E}^n.
\end{equation*}
It holds that (see, for example, Lemma 2 in \cite{Cui-2014}) that
\begin{equation} \label{eq:bound_piB}
	\pi(B_{n,\beta}) = \pi(\etab \in H_{n,\beta}^\pm) = \pi \left( \sum_{j=1}^n \eta_j \in \left[c_0 n^{1-\beta},c_1 n^{1-\beta} \right] \right) \geq 1-\smallO(n^{-1}).
\end{equation}
Therefore, by using (\ref{eq:bound_by_Bayes_risk}) as $n \to \infty$, we can estimate the minimax risk from bellow as follows
\begin{align}
&	\inf_{\tilde{\etab}} \sup_{\etab \in H^\pm_{n,\beta}} \E_{\etab} |\etab - \tilde{\etab}| \geq \inf_{\tilde{\etab}} \int_{B_{n,\beta}} \E_{\etab} |\etab - \tilde{\etab}| d \hat{\pi}(\etab) \nonumber \\
& \qquad = \inf_{\tilde{\etab}} \frac{1}{\pi(B_{n,\beta})} \int_{B_{n,\beta}}  \E_{\etab} |\etab - \tilde{\etab}| d \pi(\etab) \nonumber \\
	&\qquad = (1+\smallO(n^{-1})) \inf_{\tilde{\etab}}  \left( \int_{\mathcal{E}^n} \E_{\etab} |\etab - \tilde{\etab}| d \pi(\etab) - \int_{B^c_{n,\beta}} \E_{\etab} |\etab - \tilde{\etab}| d \pi(\etab) \right) \nonumber \\
	& \qquad \geq (1+\smallO(n^{-1})) \left( \inf_{\tilde{\etab}} \int_{\mathcal{E}^n} \E_{\etab} |\etab - \tilde{\etab}| d \pi(\etab) - \sup_{\tilde{\etab}} \int_{B^c_{n,\beta}} \E_{\etab} |\etab - \tilde{\etab}| d \pi(\etab) \right) \nonumber \\
	&\qquad =: (1+\smallO(n^{-1})) (K_1 - K_2). \label{def:K1_K2}
\end{align}
It is shown in \cite[p. 23]{Cui-2014} that $K_2=\sup_{\tilde{\etab}} \int_{B^c_{n,\beta}} \E_{\etab} |\etab - \tilde{\etab}| d \pi(\etab) \leq \smallO(1)$, and is therefore negligible. Now, consider the main term $K_1$. When $n \to \infty$,
\begin{align}
	K_1=\inf_{\tilde{\etab}} \int_{\mathcal{E}^n} \E_{\etab} |\etab - \tilde{\etab}| d \pi(\etab) &= \inf_{\tilde{\etab}}  \int_{\mathcal{E}^n} \left( \sum_{j=1}^n \E_{\eta_j} |\eta_j - \tilde{\eta}_j|\right)d \pi(\etab)\nonumber \\
	 &= \inf_{\tilde{\etab}}  \sum_{j=1}^n \E_{\pi}\E_{\eta_j} |\eta_j - \tilde{\eta}_j|\nonumber \\
	& \geq \sum_{j=1}^n \inf_{\tilde{\etab}}   \E_{\pi}\E_{\eta_j} |\eta_j - \tilde{\eta}_j| = n \inf_{\tilde{\eta_1}}  \E_{\pi_1}\E_{\eta_1} |\eta_1 - \tilde{\eta}_1|,
\end{align}
where $\eta_1$ has the prior distribution $\pi_1 = Bernoulli(r)$ and $\tilde{\eta}_1$ is an arbitrary estimate of $\eta_1$ based on the first observation $\Xb_1$ in model (\ref{model}). Since $\tilde{\eta}_1$ takes only the values 0 and 1, it may be viewed as a (nonrandomized) test function and hence
\begin{equation} \label{eq:Bayes_risk}
	K_1/n = \inf_{\tilde{\eta_1}}  \E_{\pi_1}\E_{\eta_1} |\eta_1 - \tilde{\eta}_1| = \inf_{\tilde{\eta_1}} \left[ r \E_{\ab}(1-\tilde{\eta}_1) + (1-r) \E_{\boldsymbol{0}}(\tilde{\eta}_1) \right]
\end{equation}
is the minimal Bayes risk in the testing problem $H_0: \boldsymbol{v} = \boldsymbol{0}$ vs. $H_1: \boldsymbol{v} = \boldsymbol{a}$ in the model $\Xb_1 = \boldsymbol{v} + \boldsymbol{\varepsilon}_1$, $\boldsymbol{\varepsilon}_1 \sim N_p(\boldsymbol{0},I)$ with $\boldsymbol{v} = \boldsymbol{a} \eta_1$, $\ab \in (0,\infty)^p.$ It is known (see, for example, \cite{Lehmann-2005}) that 
the minimum of the Bayes risk (\ref{eq:Bayes_risk}) is attained by the test function
\begin{equation}
	\eta^{opt} (\Xb_1) = \mathds{1}\left(\frac{rf_a(\Xb_1)}{(1-r)f_0(\Xb_1)} >1\right),
\end{equation}
where $f_\ab(\Xb_1)$ and $f_{\boldsymbol{0}}(\Xb_1)$ are densities of $N_p(\ab,\Sigma)$ and $N_p(\boldsymbol{0},\Sigma)$, i.e.
\begin{align}
	f_{\ab}(\xb) &= (2\pi)^{-p/2}|\Sigma|^{-1/2} \exp{\left\{ -\frac{1}{2}(\xb - \ab)^\top \Sigma^{-1} (\xb - \ab)\right\} }, \\
	f_{\boldsymbol{0}}(\xb) &= (2\pi)^{-p/2}|\Sigma|^{-1/2} \exp{\left\{ -\frac{1}{2}\xb^\top \Sigma^{-1} \xb\right\} }.
\end{align}
Then, we may write 
\begin{align}
	K_1 &= n\inf_{\tilde{\eta_1}} \left[ r \Prob{\ab}{\tilde{\eta_1} = 0} + (1-r) \Prob{\boldsymbol{0}}{\tilde{\eta_1}=1} \right] \nonumber \\
	& = nr \Prob{\ab}{\eta^{opt}(\Xb_1) = 0} + n(1-r) \Prob{\boldsymbol{0}}{\eta^{opt}(\Xb_1) =1} \nonumber \\
	& = nr \Prob{\ab}{\frac{r f_\ab (\Xb_1)}{(1-r)f_{\boldsymbol{0}}(\Xb_1)}\leq 1} + n(1-r) \Prob{\boldsymbol{0}}{\frac{r f_\ab (\Xb_1)}{(1-r)f_{\boldsymbol{0}}(\Xb_1)}> 1} \nonumber \\
	& =: J_1 + J_2. \label{eq:def_J1_J2}
\end{align}
By substituting the ratio
\begin{align*} 
	\frac{f_{\ab}(\xb)}{f_{\boldsymbol{0}}(\xb)} &=  \exp \left\{ -\frac{1}{2}(\xb - \ab)^\top \Sigma^{-1} (\xb - \ab) + \frac{1}{2}\xb^\top \Sigma^{-1} \xb \right\} = \exp \left\{ -\frac{1}{2} \ab^\top \Sigma^{-1}\ab + \ab^\top \Sigma^{-1} \xb \right\}
\end{align*}
and $r = n^{-\beta}$ into (\ref{eq:def_J1_J2}), we obtain 
\begin{align} \label{eq:J2_v1}
	J_2 &= n(1-n^{-\beta}) \Prob{\boldsymbol{0}}{\frac{n^{-\beta}}{1-n^{-\beta}}  \exp \left\{ -\frac{1}{2} \ab^\top \Sigma^{-1}\ab + \ab^\top \Sigma^{-1} \xb \right\} > 1} \nonumber \\  
	&= n(1-n^{-\beta}) \Prob{\boldsymbol{0}}{\exp \left\{ -\frac{1}{2} \ab^\top \Sigma^{-1}\ab + \ab^\top \Sigma^{-1} \xb \right\} > n^\beta -1} \nonumber \\  
	&= n(1-n^{-\beta}) \Prob{\boldsymbol{0}}{\ab^\top \Sigma^{-1} \Xb_1 > \log\left(n^\beta-1 \right) + \frac{1}{2} \ab^\top \Sigma^{-1}\ab} ,
\end{align}
where $\Xb_1 \sim N_p(\boldsymbol{0},\Sigma)$ and so 
\begin{equation*}
	\frac{1}{\sqrt{\ab^\top \Sigma^{-1}\ab}} \ab^{\top} \Sigma^{-1} \Xb_1 \sim N(0,1).
\end{equation*}
Therefore, (\ref{eq:J2_v1}) turns into
\begin{align} \label{eq:J2_v2}
	J_2 &= n(1-n^{-\beta}) \Phi \left( - \frac{1}{\sqrt{\ab^\top \Sigma^{-1}\ab}} \left( \log\left(n^\beta-1 \right) + \frac{1}{2} \ab^\top \Sigma^{-1}\ab \right)\right) ,
\end{align}
where $\Phi$ stands for the cumulative distribution function of $N(0,1)$. Since $\sqrt{\ab^\top \Sigma^{-1}\ab} = \lVert L^{-1}\ab \rVert_2$, the well-known relation between the $\ell_1$ and $\ell_2$ norms implies that
\begin{equation} \label{eq:bound_sqrt_aSgimainva}
	\frac{1}{\sqrt{p}} \lVert L^{-1} \ab \rVert_1 \leq \sqrt{\ab^\top \Sigma^{-1}\ab}  \leq \lVert L^{-1} \ab \rVert_1.
\end{equation}
Using the asymptotic relation $\log (n^\beta -1) \sim \beta \log(n)$ as $n \to \infty$, together with (\ref{eq:bound_sqrt_aSgimainva}) and condition (\ref{cond:a_LB}), which states that
$$ \lVert L^{-1} \ab \rVert_1 =\gamma \sqrt{p\log n}, \quad \quad 0 < \gamma < \sqrt{2}(1+\sqrt{1-\beta}),$$ 
we obtain from (\ref{eq:J2_v2})
\begin{align}\label{eq:J2_v3}
	J_2 &> n(1-n^{-\beta}) \Phi \left( - \frac{\sqrt{p}}{\lVert L^{-1} \ab \rVert_1} \left( \beta \log n  + \frac{1}{2} \lVert L^{-1} \ab \rVert_1 \right)\right) \nonumber \\
	& > n(1-n^{-\beta}) \Phi \left( - \frac{1}{\gamma \sqrt{\log n}} \left( \beta \log\left(n \right) + \frac{1}{2} \gamma \sqrt{p\log n} \right)\right) \nonumber \\
	&= n(1-n^{-\beta}) \Phi \left( - \frac{\beta}{\gamma} \sqrt{\log n} - \frac{1}{2} \sqrt{p} \right).
\end{align}
Moreover, by applying $n(1-n^{-\beta}) \sim n$, $\sqrt{p}/2 = \smallO(\sqrt{\log n})$ as $n \to \infty$, and the fact \cite[Proposition 2.1.2]{Vershynin-2018} that 
\begin{equation} \label{eq:N01_tail}
	\Phi(-t) \sim \frac{e^{-t^2/2}}{\sqrt{2 \pi} t}, \quad t \to \infty,
\end{equation}
we further obtain from (\ref{eq:J2_v3}), that
\begin{align}\label{eq:J2_v4}
	J_2 &> n \frac{1}{\sqrt{2\pi \log n} \beta/\gamma} \exp \left\{ -\frac{1}{2} \frac{\beta^2}{\gamma^2} \log n \right\} \asymp  \frac{1}{\sqrt{\log n}} n^{1-\beta^2/(2\gamma^2)} .
\end{align}
Since $$1-\frac{\beta^2}{2\gamma^2} > 0 \Leftrightarrow \gamma > \sqrt{\frac{\beta}{2}},$$
we see that
\begin{equation} \label{eq:J2_v5}
	J_2 > 0 \quad \quad \text{for } \sqrt{\frac{\beta}{2}} < \gamma < \sqrt{2}(1+\sqrt{1-\beta}).
\end{equation}

Now, consider the case 
\begin{equation} \label{cond:gamma2}
	\lVert L^{-1} \ab \rVert_1 =\gamma \sqrt{p\log n}, \quad \quad 0 < \gamma \leq \sqrt{\frac{\beta}{2}}.
\end{equation}
By using the same arguments as above and the fact that $\frac{1}{\sqrt{\ab^\top \Sigma^{-1}\ab}} \ab^{\top} \Sigma^{-1} (\Xb_1-\ab) \sim N(0,1)$, we obtain 
\begin{align} \label{eq:J1_v1}
	J_1 > n^{1-\beta} \Phi \left( \frac{\beta \log n}{\gamma \sqrt{p\log n}} - \frac{\gamma}{2} \sqrt{p\log n} \right) = n^{1-\beta} \Phi \left( \left( \frac{\beta}{\gamma p} - \frac{\gamma}{2} \right) \sqrt{p\log n}  \right). 
\end{align}
Moreover, due to (\ref{cond:gamma2}), we can write
\begin{align} \label{eq:J1_v2}
	J_1 > n^{1-\beta} \Phi \left( \left( \frac{\sqrt{2\beta}}{p} - \frac{\sqrt{\beta}}{2\sqrt{2}} \right) \sqrt{p\log n}  \right) =  n^{1-\beta} \Phi \left( \left( \frac{1}{p} - \frac{1}{4} \right) \sqrt{2\beta p\log n}  \right).
\end{align}
It is immediate that $J_1 >0$ for $p\leq 4$. For $p > 4$, we have $1/p-1/4 <0$ and applying (\ref{eq:N01_tail}) again gives
\begin{align} 
	J_1 &> n^{1-\beta} \frac{1}{\sqrt{2\pi} \left( -\frac{1}{ p} + \frac{1}{4} \right) \sqrt{2 \beta p\log n} }\exp \left( -\left( \frac{1}{ p} - \frac{1}{4} \right)^2 \beta p\log n  \right) \nonumber \\
	& \asymp \frac{1}{\sqrt{\log n}} n^{1-\beta - \left( \frac{1}{ p} - \frac{1}{4} \right)^2 \beta p}  >0,
\end{align} 
since $\log^s (n) = \smallO(n^q)$, $n \to \infty$, for any $s>0$ and $q>0$. In our case, $s=\frac{1}{2}$ and $q = 1-\beta - \left( \frac{1}{ p} - \frac{1}{4} \right)^2 \beta p >0$ for $p \in \bigg[4, \lfloor \frac{8}{\beta} (1+\sqrt{1-\beta})-4 \rfloor \bigg) =: [4,p_+)$, where the upper bound on $p$ is guaranteed by the Theorem assumption. Combining the above results, we obtain 
\begin{equation} \label{eq:J1_v4}
	J_1 >0, \quad \quad \text{for } 0 < \gamma \leq \sqrt{\frac{\beta}{2}} \quad \text{and} \quad p=1,\ldots, p_+.
\end{equation}

Substituting (\ref{eq:J2_v5}) and (\ref{eq:J1_v4}) into (\ref{eq:def_J1_J2}) yields $K_1>0$. Therefore,
$$ \inf_{\tilde{\etab}} \mathcal{R}(\tilde{\etab},\etab) >0,$$
which completes the proof.

\qedsymbol

\section{Additional figures for Section 4}
	\label{sec:Sfigures}
	In this section, we provide QQ-plots and autocorrelation function plots from the real data analyses presented in Section 4 of the main paper.

\subsection{Daily step counts}
The following figures are provided to assess the normality and independence of the ARIMA residuals analyzed in Section 4.1 of the main paper. The specific ARIMA model used is indicated in each figure caption.

\clearpage 
\begin{figure}[h!]
		\centering
		\includegraphics[width=0.8\textwidth]{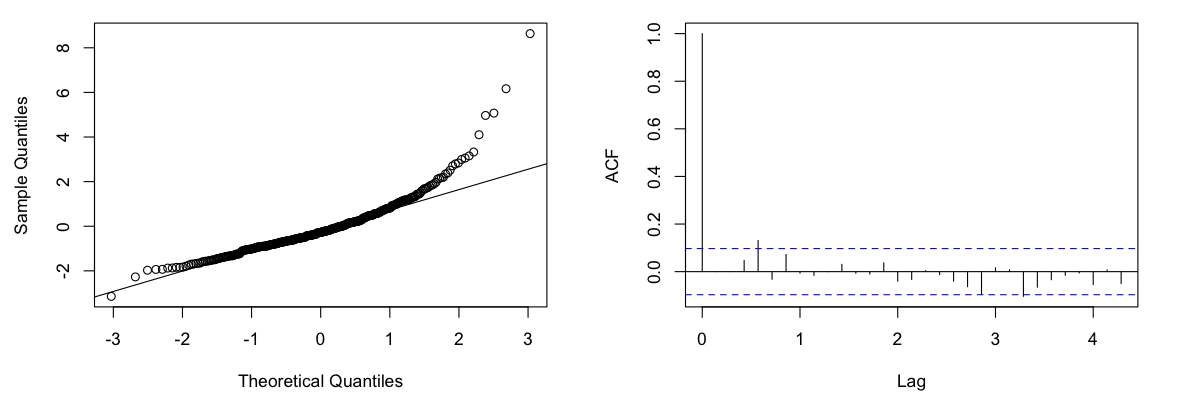}
	\caption{Participant with ID 999. QQ-plot and the autocorrelation function for model residuals. Time period 409 days long. Time series modeled by ARIMA(1,1,2)(0,0,1)[7].}
	\label{fig:steps_id999}
\end{figure}

\begin{figure}[h!]
		\centering
		\includegraphics[width=0.8\textwidth]{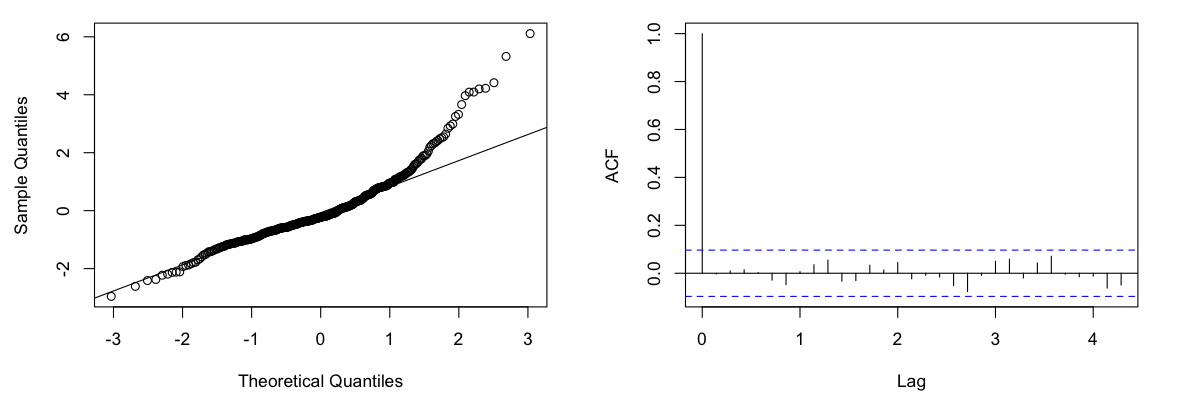}
	\caption{Participant with ID 1429. QQ-plot and the autocorrelation function for model residuals. Time period 412 days long. Time series modeled by ARIMA(1,1,1)(0,0,1)[7]. }
	\label{fig:steps_id1429}
\end{figure}

\begin{figure}[h!]
	\centering
		\includegraphics[width=0.8\textwidth]{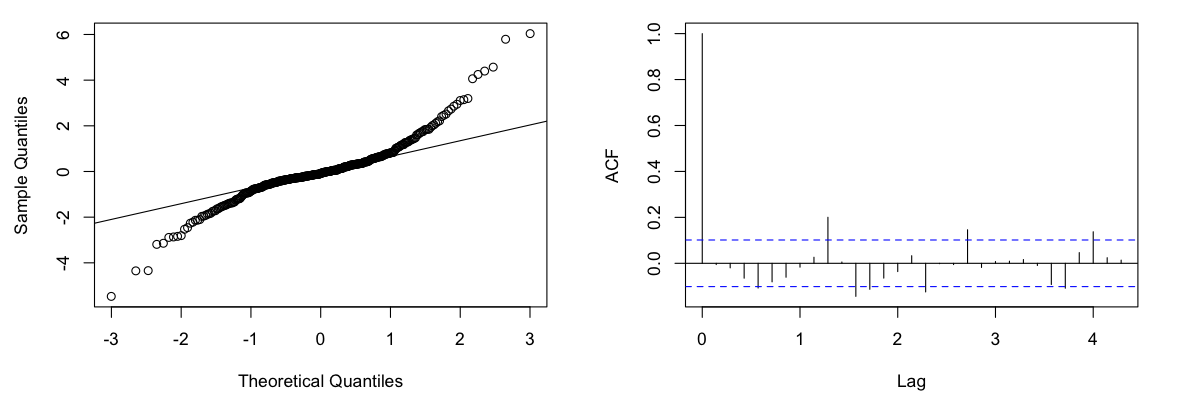}
	\caption{Participant with ID 2110. QQ-plot and the autocorrelation function for model residuals. Time period 373 days long. Time series modeled by ARIMA(5,1,0)(2,0,0)[7]. }
	\label{fig:steps_id2110}
\end{figure}

\subsection{Air pollution data}
The figure below illustrates diagnostic checks of the VAR residuals considered in Section 4.2 of the main paper. 

\begin{figure}[h]
	\centering
	\begin{subfigure}[h]{\textwidth}
		\centering
		\includegraphics[width=0.85\textwidth]{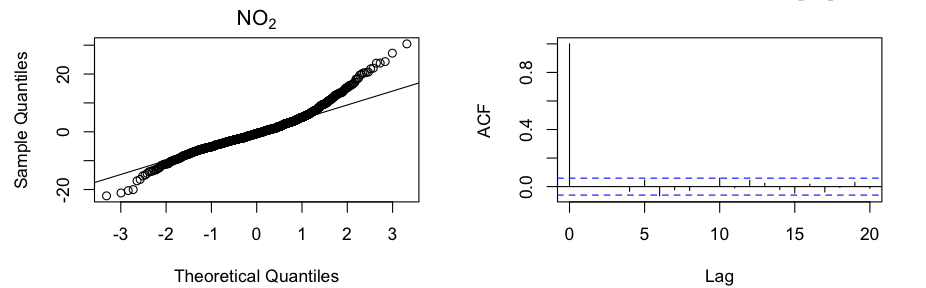}
	\end{subfigure}
	~
	\begin{subfigure}[h]{\textwidth}
		\centering
		\includegraphics[width=0.85\textwidth]{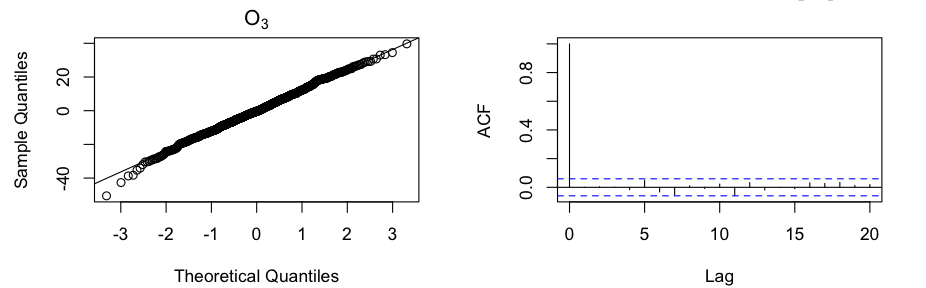}
	\end{subfigure}
	~ 
	\begin{subfigure}[h]{\textwidth}
		\centering
		\includegraphics[width=0.85\textwidth]{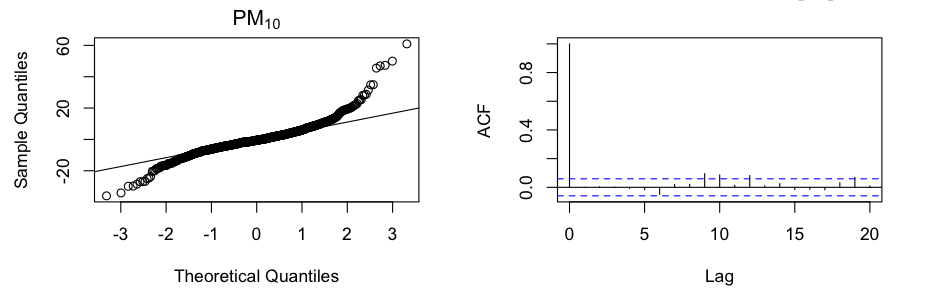}
	\end{subfigure}
	~
	\begin{subfigure}[h]{\textwidth}
		\centering
		\includegraphics[width=0.85\textwidth]{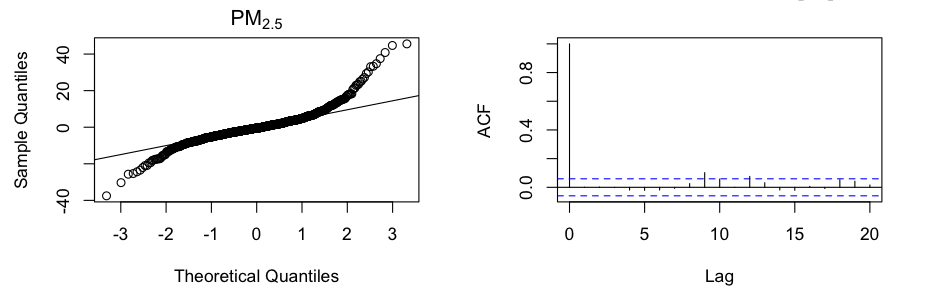}
	\end{subfigure}
	~
	\begin{subfigure}[h]{\textwidth}
		\centering
		\includegraphics[width=0.85\textwidth]{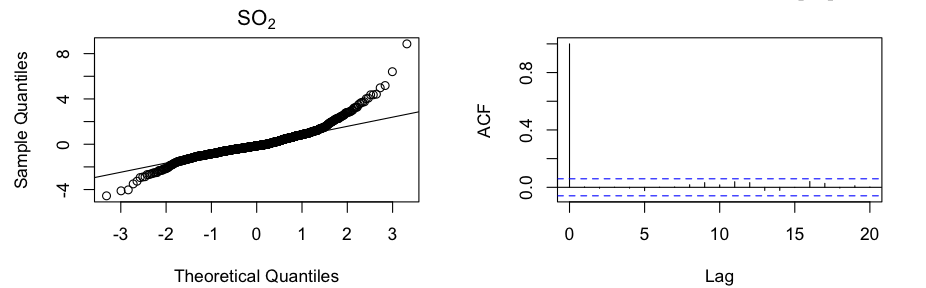}
	\end{subfigure}  
	\caption{QQ-plot and autocorrelation function of the residuals from the VAR(4) model for pollutants NO$_2$, O$_3$, PM$_{10}$, PM$_{2.5}$ and SO$_2$. }
	\label{fig:QQplot_ACF_4_pollutants}
\end{figure}

\end{document}